\documentclass[11pt]{article}
\usepackage{amsfonts}
\usepackage{amssymb,amsmath}
\usepackage{amsthm,amsgen}
\usepackage{graphicx,epsfig}
\usepackage[mathscr]{eucal}
\usepackage{rawfonts}
\usepackage{pictex}
\usepackage{color}
\newtheorem{thm}{THEOREM}[section]
\newtheorem*{thm*}{THEOREM}

\newtheorem{prop}[thm]{Proposition}
\theoremstyle{definition}
\newtheorem{defn}[thm]{Definition}
\newtheorem{rem}[thm]{Remark}

\newcommand{\eps}{\varepsilon}

\newcounter{subequation}
	\newenvironment{subequation}%
	{\addtocounter{equation}{-1}%
	\stepcounter{subequation}%
	\begin{equation}}%
	{\end{equation}%
}

\newcommand{\ff}{\mathfrak{f}}

\newcommand{\fl}{\mathfrak{l}}
\newcommand{\fm}{\mathfrak{m}}

\newcommand{\fz}{\mathfrak{z}}

\newcommand{\fO}{\mathfrak{O}}
\newcommand{\rcyl}{{\mathcal{Z}}}

\newcommand{\ts}{\textstyle}
\newcommand{\sgn}{\mbox{sgn}}
\newcommand{\bA}{\mathbf{A}}
\newcommand{\bAKN}{\mathbf{A}_{\textsc{KN}}}
\newcommand{\AKN}{{A}_{\textsc{KN}}}
\newcommand{\bAKNanom}{\mathbf{A}_{\textsc{KN}}^{\mbox{\tiny\textrm{gen}}}}
\newcommand{\AKNanom}{{A}_{\textsc{KN}}^{\mbox{\tiny\textrm{gen}}}}
\newcommand{\bAKNhydro}{\mathbf{A}_{\mbox{\tiny\textrm{hyd}}}}
\newcommand{\bB}{\mathbf{B}}

\newcommand{\bD}{\mathbf{D}}

\newcommand{\bE}{\mathbf{E}}
\newcommand{\be}{\mathbf{e}}
\newcommand{\bF}{\mathbf{F}}

\newcommand{\bH}{\mathbf{H}}

\newcommand{\bJ}{\mathbf{J}}
\newcommand{\bj}{\mathbf{j}}

\newcommand{\bl}{\mathbf{l}}
\newcommand{\bL}{\mathbf{L}}
\newcommand{\bM}{\mathbf{M}}
\newcommand{\bm}{\mathbf{m}}

\newcommand{\bn}{\mathbf{n}}
\newcommand{\bN}{\mathbf{N}}

\newcommand{\bq}{\mathbf{q}}

\newcommand{\br}{\mathbf{r}}

\newcommand{\bs}{\mathbf{s}}
\newcommand{\bT}{\mathbf{T}}

\newcommand{\bv}{\mathbf{v}}

\newcommand{\bg}{\mathbf{g}}

\newcommand{\beq}{\begin{equation}}
\newcommand{\eeq}{\end{equation}}
\newcommand{\bseq}{\begin{subequation}}
\newcommand{\eseq}{\end{subequation}}
\newcommand{\refeq}[1]{(\ref{#1})}

\newcommand{\fD}{\mathfrak{D}}
\newcommand{\fF}{\mathfrak{F}}

\newcommand{\fR}{\mathfrak{R}}

\newcommand{\p}{\partial}

\newcommand{\siV}{\boldsymbol{\sigma}}

\newcommand{\qV}{\mathbf{q}}
\newcommand{\QV}{\mathbf{Q}}

\newcommand{\brho}{{\rho_{\scriptscriptstyle{\psi}}}}
\newcommand{\bOm}{\boldsymbol{\Omega}}

\newcommand{\cA}{\mathcal{A}}
\newcommand{\cB}{\mathcal{B}}

\newcommand{\cR}{\mathcal{R}}
\newcommand{\cC}{\mathcal{C}}
\newcommand{\cD}{\mathcal{D}}

\newcommand{\cM}{{\mathcal M}}
\newcommand{\cN}{{\mathcal N}}
\newcommand{\cP}{{\mathcal P}}

\newcommand{\cS}{{\mathcal S}}
\newcommand{\cT}{\mathcal{T}}

\newcommand{\cZ}{\mathcal{Z}}

\newcommand{\Cset}{{\mathbb C}}

\newcommand{\Rset}{{\mathbb R}}
\newcommand{\Sset}{{\mathbb S}}
\newcommand{\Zset}{{\mathbb Z}}

\newcommand{\Lsp}{\mathfrak{L}}

\newcommand{\pdt}{{\partial_t^{\phantom{0}}}}

\newcommand{\la}{\lambda}
\newcommand{\La}{\Lambda}
\newcommand{\de}{\delta}
\newcommand{\De}{\Delta}
\newcommand{\al}{\alpha}

\newcommand{\ga}{\gamma}
\newcommand{\Ga}{\Gamma}
\newcommand{\ep}{\epsilon}

\newcommand{\ka}{\kappa}
\newcommand{\om}{\omega}
\newcommand{\Om}{\Omega}
\newcommand{\si}{\sigma}
\newcommand{\Si}{\Sigma}
\newcommand{\nab}{\nabla}
\newcommand{\half}{\frac{1}{2}}

\newcommand{\diag}{\mbox{diag}}

\newcommand{\bna}{\begin{eqnarray}}
\newcommand{\ena}{\end{eqnarray}}
\newcommand{\bea}{\begin{eqnarray*}}
\newcommand{\eea}{\end{eqnarray*}}
\newcommand{\ben}{\begin{enumerate}}
\newcommand{\een}{\end{enumerate}}
\newcommand{\bi}{\begin{itemize}}
\newcommand{\ei}{\end{itemize}}

\newcommand{\RR}{{\mathbb R}}

\newcommand{\CC}{{\mathbb C}}

\newcommand{\grad}{\nabla}

\begin{document}

\title{A novel quantum-mechanical interpretation of the Dirac equation}

\author{\normalsize\sc{M. K.-H. Kiessling and A. S. Tahvildar-Zadeh}\\
	{$\phantom{nix}$}\\[-0.1cm] 
        \normalsize Department of Mathematics\\[-0.1cm]
	Rutgers, The State University of New Jersey\\[-0.1cm]
	110 Frelinghuysen Rd., Piscataway, NJ 08854}
\vspace{-0.3cm}
\date{Orig. version: Nov. 16, 2014; revised: Dec. 1, 2015; printed \today} 
\maketitle
%
%
\begin{abstract}

\noindent
	A novel interpretation is given of Dirac's ``wave equation for the relativistic electron''
as a quantum-mechanical one-particle equation.
        In this interpretation the electron and the positron are merely the two different ``topological spin'' states
of a single more fundamental particle, not distinct particles in their own right. 
        The new interpretation is backed up by the existence of such ``bi-particle'' structures in general relativity,
in particular the ring singularity present in any spacelike section of the spacetime singularity of the maximal-analytically 
extended, 
topologically non-trivial, electromagnetic Kerr--Newman spacetime in the zero-gravity limit (here, ``zero-gravity'' means the limit 
$G\to 0$, where $G$ is Newton's constant of universal gravitation).
        This novel interpretation resolves the dilemma that Dirac's wave equation seems to be capable of describing 
both the electron and the positron in ``external'' fields in many relevant situations, while
the bi-spinorial wave function has only a single position variable in its argument, not two --- as it should if it were 
a quantum-mechanical two-particle wave equation.
        A Dirac equation is formulated for such a ring-like bi-particle which interacts with a static point charge located elsewhere
in the topologically non-trivial physical space associated with the moving ring particle, the motion being governed by a
de Broglie--Bohm type law extracted from the Dirac equation. 
        As an application, the pertinent general-relativistic zero-gravity Hydrogen problem is studied in the usual 
Born--Oppenheimer approximation. 
        Its spectral results suggest that the zero-$G$ Kerr--Newman magnetic moment be identified with the so-called 
``anomalous magnetic moment of the physical electron,'' not with the Bohr magneton, so that the
ring radius is only a tiny fraction of the electron's reduced Compton wavelength.
\end{abstract}
\smallskip

\vfill
\hrule
\smallskip
\copyright{2014/2015. The authors. Reproduction for non-commercial purposes is permitted.}

\medskip

email: miki@math.rutgers.edu, shadi@math.rutgers.edu

\newpage

\section{Introduction}

  We begin with a brief history of Dirac's equation and its quantum mechanical interpretations.
  Readers who are familiar with this subject may wish to skip ahead to section \ref{sec:novel} for an introduction 
to our novel proposal, and for an executive summary of our main results at the end of that section.

\subsection{Quest for the quantum-mechanical interpretation of Dirac's equation}

  Many textbooks and monographs (e.g. \cite{MessiahBOOKii}, \cite{GMRonQED}, \cite{ThallerBOOK}, \cite{SchweberBOOK}) 
tell the story of Dirac's marvelous discovery of the special-relativistic generalization of Pauli's 
non-relativistic spinor wave equation for the ``spinning electron'': his ingenious insight that a first-order
partial differential operator is needed rather than the second-order wave operator in the so-called\footnote{We are alluding to
the history of this equation: first contemplated by Schr\"odinger \emph{before} he discovered his nonrelativistic
equation, and rediscovered by various physicists, amongst them Klein,  Gordon, and also Pauli.}
Klein--Gordon equation; 
his equally ingenious insight that complex four-component bi-spinors, instead of Pauli's complex two-component spinors,
are needed to accomplish this goal; his consequential formulation of the equation and his skillful analysis of the same; 
the $g$-factor 2; the explicit 
solution of the Hydrogen problem in terms of elementary functions (by Darwin \cite{DarwinDIRACspec} 
and Gordon \cite{GordonDIRACspec}), and the surprising exact agreement of the 
Dirac point spectrum with Sommerfeld's fine structure formula (save the labeling of the energy and angular momentum levels);\footnote{For 
a modernized semi-classical Bohr--Sommerfeld approach that leads to the correct labeling, see \cite{KeppelerSEMIclassQM}.} 
the strange occurrence of a negative energy continuum below $-mc^2$ (with $m$ the electron's rest mass, and $c$ the speed of light in
vacuum), leading to Dirac's ultimate triumph: the prediction of the existence of the anti-electron (a.k.a. positron) --- 
based on his ``holes in the Dirac sea'' interpretation.

 Yet, not all is well. 
 The most perplexing (and intriguing) part of the above success story is that it is based on Dirac's changing 
the rules of the game while he was playing it. 
 Thus, originally formulating it as a quantum-mechanical one-particle wave equation on one-particle configuration space,
Dirac --- by postulating that all negative energy continuum states are occupied with electrons --- switched to an infinitely-many 
particle interpretation on physical space:
a ``poor man's quantum field theory'' (yet an important stepping stone towards quantum electrodynamics, one that still is 
the subject of serious studies by mathematical physicists \cite{LiebSeiringerBOOK}).
 And while it is difficult to argue with practical success, Dirac's ad-hoc ``holes in the negative energy sea'' theory can be 
criticized for de-facto bypassing the  conceptual problem of the proper quantum-mechanical interpretation of Dirac's 
equation (assuming there exists one!), rather than addressing it.

 St\"uckelberg \cite{Stueckelberg42} and Feynman \cite{Fey49} revived the quantum-mechanical interpretation of 
Dirac's equation and argued that it is an equation for both, the electron and the positron, cf. Thaller \cite{Thaller}: 
``[U]p to now no particular quantum mechanical interpretation [of Dirac's wave equation] is generally accepted. 
... [T]he St\"uckelberg--Feynman interpretation (St\"uckelberg 1942, Feynman 1949) ... is
intermediate between a one-particle theory and Dirac's hole theory, because it claims
that the Dirac equation is able to describe two kinds of particles, namely electrons and
positrons (but not their interaction; negative energy states are directly observed as
positrons with positive energy).''
 Thaller, who adopts the St\"uckelberg--Feynman interpretation in his scattering theory, goes on to emphasize that it is
``formulated in the language of wave packets ... and does not rely on unobservable objects like the Dirac sea.''
 In a nutshell, the main idea is that wave packets composed of only positive energy eigenfunctions and scattering
states ``describe'' electrons, those composed of only negative energy eigenfunctions and scattering
states ``describe'' positrons; cf. also \cite{BrownRavenhall}.
 The problem with identifying electrons and positrons with respect to the respective subspaces of the ``free'' Hamiltonian
is that
switching on ``external'' fields may induce transitions between the two; or, in mathematical language: the positive 
and negative energy Hilbert subspaces for the free Hamiltonian may not be invariant under the unitary evolution 
of a scattering problem.
 Klein's paradox (see \cite{ThallerBOOK}) shows that these subspaces are not always invariant under the unitary
evolution, indeed.
 Furthermore there are examples of physically interesting electric potentials (e.g. a regularized Coulomb potential)
which can generate negative energy eigenstates which clearly are to be interpreted as belonging to the electron, not 
the positron; see \cite{GMRonQED}.
 Even in favorable situations, mixed initial conditions can pose an interpretational dilemma.

 Thus not all is well with the St\"uckelberg--Feynman interpretation either.
 To Thaller's emphasis of its limitations  we here add another troublesome
criticism, namely that this interpretation is in conflict with established quantum-mechanical many-body concepts.
 Thus, while Dirac's bi-spinorial wave function depends (beside time) only on one position variable, it should have two position 
variables in its arguments, not merely one, if Dirac's equation truly were a quantum-mechanical two-particle equation.
 For then it has to reproduce Pauli's quantum-mechanical two-particle equation for a non-relativistic electron--positron pair in the
non-relativistic limit,
and this Pauli equation --- which is the traditional starting point for computing the leading-order spectral properties of positronium, 
with relativistic corrections computed perturbatively subsequently \cite{BeSaBOOK} ---  does have a complex four-component wave function 
with two position variables in its arguments.\footnote{There is one loose end in this argument --- which can easily be tied up.
  More to the point, as in the quantum-mechanical analog of the Kepler problem, one can change coordinates from the two particle
position vectors (defined w.r.t. some inert fixed point) to center-of-mass coordinates (w.r.t. the same fixed point) plus 
relative coordinates. 
  In the spinless non-relativistic limit the center-of-mass coordinates can be factored out, and the remaining wave equation in the
relative coordinates is effectively a one-body equation for a test particle in an external Coulomb field.
  So one may speculate whether Dirac's equation is perhaps of this type, describing the \emph{relative} motion of
electron and positron in relative coordinates.
  Unfortunately the mass parameter $m$ in Dirac's equation ``for the relativistic electron'' is a factor 2 too big 
for this interpretation to be feasible, putting to rest speculations in this direction.} 

\subsection{A novel proposal}\label{sec:novel}

 Since the enigmatic quantum-mechanical character of ``Dirac's equation for the electron'' manifests itself in the seemingly contradictory 
aspects that, on one hand, it has all the formal hallmarks of a \emph{single-particle} equation, while, on the other,
its set of solutions covers (much of) the physics of \emph{both} the \emph{electron} and the \emph{anti-electron},
the only logical conclusion that avoids this paradox is that electron and anti-electron are not
separate entities but merely ``two different sides of the same medal.'' 
 By this we do not mean the well-known charge-conjugation transformation which maps a ``negative-energy Dirac electron state'' into a 
``positive-energy Dirac positron state,'' and vice versa; neither do we mean the interpretation suggested by St\"uckelberg, Wheeler, 
and Feynman that an ``electron moving backward in time appears as a positron moving forward in time.''
 Rather, 
we propose that \emph{electron and anti-electron are two different ``topological-spin'' states of a single, more structured particle}
 (moving forward in time).

 Our novel interpretation is backed up by the fact that such bi-particle structures exist in some electromagnetic 
spacetime solutions of Einstein's classical theory of general relativity.\footnote{By invoking a classical physical theory, 
in our case Einstein's general relativity, we are simply following the standard protocol of ``first quantization.'' 
 Enjoying the benefit of hindsight we may simplify our task and look for the desired electromagnetic ``bi-particle
structure'' in classical physical theory, then compute its electromagnetic interaction with some other classical elecromagnetic
object, in particular a point charge. 
 This classical interaction then enters the pertinent quantum-mechanical wave equation obtained from ``first quantization,'' in 
our case Dirac's equation.
 The remaining task then is to show that this Dirac equation makes the right physical predictions and avoids any of the 
paradoxical aspects that we have discussed.}
 More concretely, we mean the structure commonly known as the 
``ring\footnote{What is ``ring-like'' is actually a constant-$t$ ``snapshot'' of the spacetime singularity.}
singularity'' of the maximal-analytically extended, topologically non-trivial,\footnote{The 
        complement of a wedding ring in three-dimensional Euclidean space 
		is topologically non-trivial, too, but ``looping through the ring once brings you back to where you began;'' in a 
		spacelike slice of the maximal-analytically extended Kerr--Newman spacetime ``you need to
		loop through the ring twice to get back to square one.''}
electromagnetic spacetime of Newman et al. \cite{NCCEPT65,BoyLin67} --- in its zero-gravity limit.
       Here, ``zero-gravity limit'' means the limit $G\to 0$, where $G$ is Newton's constant of universal gravitation, applied to
the Kerr--Newman metric\footnote{A priori speaking the $G\to0$ limit should also be applied to the expressions for the 
  Kerr--Newman electromagnetic fields, but these are $G$-independent in the coordinates we use.} 
expressed in any one of the most symmetric global coordinate charts.

\noindent{\textbf{1.2.1 The zero-gravity Kerr--Newman (z$G$KN) spacetime}}.

\noindent
        As shown rigorously in an accompanying paper \cite{zGKN}, this z$G$KN spacetime
itself is a flat but topologically nontrivial solution of the vacuum Einstein equations $R_{\mu\nu}=0$, while its electromagnetic 
fields are static Appell--Sommerfeld solutions\footnote{Remarkably, the electromagnetic fields were discovered in 1887 by 
Appell \cite{Appell} by performing a complex translation ``of physical space'' in the Coulomb potential. 
 However, Appell did not realize, apparently, that the complex generalization of the Coulomb potential lives on a topologically
  non-trivial extension of Euclidean $\Rset^3$.  This insight is due to Sommerfeld \cite{Som97}, 
who is credited with introducing the concept of {\em branched Riemann space}, a three-dimensional analog of a Riemann 
surface, on which multi-valued harmonic functions such as the Appell fields become single-valued.}
 of the linear Maxwell equations on this spacetime, with well-behaved sesqui-pole sources 
concentrated on its spacelike ring singularity.
          Most importantly, the ring singularity of a spacelike snapshot of the spacetime appears to be positively charged in one ``sheet,''
yet negatively charged in the other, with pertinent electric monopole asymptotics near spacelike infinity (see Fig.~1); 
analogous statements hold for the magnetic structure, with dipole asymptotics near spacelike infinity.
          The \emph{electromagnetically anti-symmetric spacelike ring singularity of this double-sheeted Sommerfeld space} 
is precisely the kind of bi-particle-like structure which 
vindicates our informal wording that ``electron and anti-electron are two sides of one and the same medal'' as mathematically realized in 
a classical physical field theory, and therefore as putatively physically viable in a quantized theory.

\begin{figure}[ht]
\centering
\includegraphics[scale=0.3]{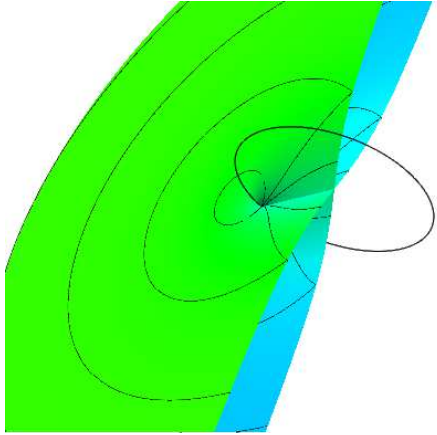}
\caption{
\small An illustration of the z$G$KN spacetime. Shown are the ring singularity
and one of the constant-azimuth sections of a constant-time z$G$KN snapshot immersed into three-dimensional Euclidean space; the 
two sheets of this flat Sommerfeld space are slightly bent for the purpose of visualization.
The apparent intersection of the two sheets along a radius of the ring singularity is a consequence of the immersion and not
a faithful representation.
Also shown are equipotential lines of the Appell--Sommerfeld electric field; the ring singularity appears positively charged in one sheet, 
and negatively in the other,
so that far from the ring the electric field becomes positive, respectively negative monopole-like,
whereas inside a ball having the ring singularity as equator it appears dipolar, with field 
lines transiting from one sheet to the other.}
\end{figure}

         We stress the importance of the \emph{zero-gravity limit} of the \emph{maximal-analytically extended} 
Kerr--Newman spacetime and its elecromagnetic fields \cite{zGKN} for getting our electron/anti-electron bi-particle 
interpretation of its ring singularity off the ground.
        The reason is that the energy-momentum-stress tensor of the electromagnetic Appell--Sommerfeld fields is not locally integrable 
at the ring singularity, so when coupled to spacetime curvature by ``switching on gravity,'' it produces the 
physically ``pathological'' effects of the $G>0$ maximal-analytically extended
Kerr--Newman spacetime (such as closed timelike loops; the $t=$const. ring singularity itself is 
also timelike when $G>0$ and not interpretable as a stationary particle-like source \cite{Car68})\footnote{While ``switching off gravity'' 
  should not be viewed as a troublesome step, 
our statement that the zero-$G$ limit is important to get our electron/anti-electron bi-particle interpretation 
  of its ring singularity off the ground, in the sense that the KN ``ring singularity'' cannot be given a bi-particle interpretation 
for $G>0$, could seem to deal a fatal blow to our proposal.
  However, the $G>0$ problems we have in mind are similar to those affecting the usual textbook formulation of Dirac Hydrogen,
  viz. ``Dirac's equation for   a point (test) electron in the field of a point proton at rest in Minkowski spacetime'' --- as soon as one 
  ``switches on $G$'' to compute the general-relativistic corrections to the Sommerfeld spectrum from eigenstates for Dirac's equation of 
  a point (test) electron in the Reissner--Nordstr\"om spacetime, ``all hell breaks lose'' (see \cite{KieTah15a} for a review).
  The reason is once again the local non-integrability of the energy-momentum-stress tensor of the classical electromagnetic field, which
  in the usual textbook case is the field of a point charge.
  This clearly suggests that the culprit is Maxwell's linear vacuum field equations, and that replacing them with non-linear field
equations that produce an integrable energy-momentum-stress tensor, the $G>0$ pathologies might go away and 
general-relativistic gravitational curvature corrections to the spectral features of special-relativistic Dirac Hydrogen 
may eventually come out not larger than the tiny Newtonian gravitational effects relative to the electromagnetic 
effects on the non-relativistic Schr\"odinger spectrum of  Hydrogen.
 More about this in our concluding section.
 In short, the pathological problems of the $G>0$ Kerr--Newman singularity are bad news for the physical interpretation of
the maximal-analytically extended Kerr--Newman spacetime with $G>0$, but not bad news for our proposal that electron and anti-electron 
might be the two electromagnetic ``sides'' of a bi-particle-like ring singularity associated with a two-sheeted space.}
 --- more on this in section 2.
        The $G\to0$ limit removes all these troublesome features of the Kerr--Newman spacetime while retaining the topological and
electromagnetic features that can be given a clear physical meaning.\footnote{It 
  is our pleasant duty to note here that we are not the first to try to associate the 
  double sheetedness of the z$G$KN spacetime and its anti-symmetric electromagnetism with particle / anti-particle symmetry.
  Namely, after reading an earlier version of our paper, Prof. Ted Newman kindly send us a letter with (amongst other valuable information) 
  the following historical note: 
  ``I had (in about 1966) myself noted --- in the $G=0$ limit --- that by going 
  along the axis of symmetry from large $r$ down to zero and then to negative $r$ that the sign of the charge would change. 
  Kerr and Penrose did clarify this for me. 
  I and one of my Polish colleagues, spent a great deal of time --- talking and trying out ideas concerning the double sheetedness --- 
  mainly trying to put it into the context of parity, time reversal and charge conjugation symmetry.
  It was a very enticing idea. 
  Unfortunately we never could make it into a unified theory or even a coherent idea.
  We eventually gave it up --- not because it was wrong --- but because we were stuck and could not make it work.''
  In this paper we will see that the double sheetedness and its charge (anti-)symmetry are associated not with parity, 
time reversal, and charge conjugation symmetry operations, but with a novel symmetry operation on the Dirac bi-spinor,
associated with what we call ``topo(logical)-spin'' of the z$G$KN ring singularity, viz. the particle/anti-particle symmetry of the
z$G$KN ring singularity is implemented in a topo-spin flip operation and, thus, is somewhat akin to Heisenberg's iso-spin
concept for nucleons.\label{Ted}}

\noindent{\textbf{1.2.2 The physical spacetime of a moving z$G$KN-type ring singularity}}.

\noindent
 Having identified a candidate structure for our electron/anti-electron bi-particle proposal
in form of the ring singularity of a snapshot of the static z$G$KN spacetime, we next describe the
\emph{physical spacetime} associated with such a ring-like bi-particle when moving,
were motion is understood as relative to an inertial frame attached to an arbitrarily chosen origin 
(the analog of Carl Neumann's ``$\alpha$ body frame'' for Newton's mechanics).
We restrict our considerations mostly to the \emph{quasi-static motions} approximation and only briefly comment on more
general motions.

 In the simplest special case, when the ring particle is at rest relative to the $\alpha$ frame, 
the spacetime is simply (a spacelike translation and rotation of) the static z$G$KN spacetime, of course.
 Furthermore, a Lorentz boost of such a static z$G$KN spacetime coordinates produces the spacetime of a ring singularity in
straight inert motion relative to the $\alpha$ frame; see Fig.~2 for a totally stripped down illustration of the physical 
spacetime with a ring singularity in straight motion contrasted with the corresponding illustration of a test particle in 
straight motion in the z$G$KN spacetime.

 \begin{figure}[ht]\label{fig:picApicB}
\centering
\includegraphics[scale=0.7]{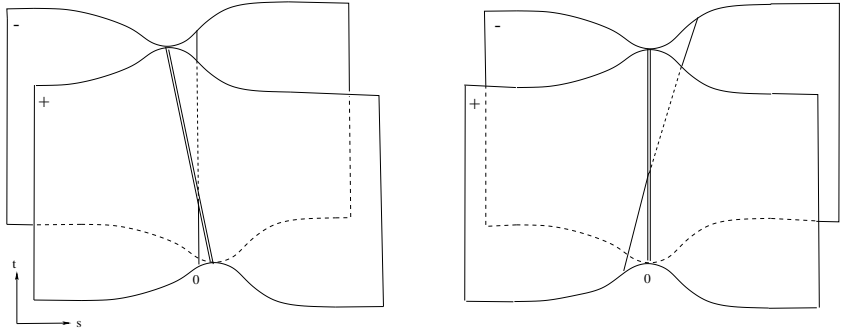}
\caption{Left: The  physical spacetime with a single ring singularity in straight line motion relative to the rest frame of the 
$\alpha$ point (marked $0$), which in the depicted scenario gets ``swept over'' by the ring singularity;
Right: Trajectory of a test particle in straight line motion in the static z$G$KN spacetime having the center of its ring singularity 
as the origin (marked $0$). The test particle transits through the ring from one sheet to the other.
(In both illustrations, the two flat sheets of the double-sheeted spacetimes are bent apart for the purpose of visualization.)}
\end{figure}

\noindent
Note that the ring will be circular generally only in its  rest frame, whereas a Lorentz boost will render it generally elliptical. 
 However, neglecting effects of order $v^2/c^2$ as relatively small in the quasi-static motions 
approximation, a moving ring can  be described as circular to leading order.

 The physical spacetime with a single moving space-like ring singularity is  \emph{generally not} obtained by a
Lorentz boost of z$G$KN, but in the quasi-static approximation it is foliated by constant-time Sommerfeld spaces 
which are all snapshots of the z$G$KN spacetime with a ring singularity translated and rotated relative to the fixed $\alpha$ frame.
 Unfortunately, it's not so easy to generalize the left illustration in Fig.~2 to visualize a double sheeted spacetime with 
a ring singularity not moving in a straight line motion with respect to an $\alpha$ point having a vertical world line. 
 Yet, we hope that the general idea of what this physical spacetime is has become clear.

 Incidentally, the double-sheeted physical zero-$G$ spacetime with a generally moving ring-like singularity will feature 
snapshots having deformed ring singularities, and a proper treatment would presumably require solving the so-called vacuum Einstein
equations.
 It is likely that the spacetime will then also feature ``gravitational'' waves.\footnote{Note that $G=0$ uncouples
the spacetime structure from its matter and field content, but the spacetime metric does not need to be
flat. In particular, it can feature wave-like disturbances propagating at the speed of ``light.''}

\noindent{\textbf{1.2.3 The configuration space(time) of a moving z$G$KN-type ring singularity}}.

\noindent
 Finally, to formulate Dirac's equation for the spinor wave function of such a ring-like bi-particle we need to identify the
\emph{configuration space of a single ring particle} on which the wave function is defined (as  bispinor-valued $L^2$ function),
and this is simply the space of all possible cofigurations for the particle relative to the $\alpha$ frame.
  In case of the ring singularity, classically knowing its configuration means knowing the location of the center of the ring as well 
its orientation, i.e. the unit normal to the plane of the ring. 
 Quantum mechanically however, this orientation vector is itself determined by a bi-spinor wave function (as we will show), so that 
specifying the 
location of the center of the ring, with respect to a frame attached to the $\alpha$ point of the space, should suffice. 
 It follows that the configuration space for a ring-like particle should also be double-sheeted, since one can contemplate the origin to be 
``swallowed'' by the ring as its center moves, causing the origin to end up on a different sheet of the physical space. 
Since that is a different configuration for the ring, one sheet will not be enough to encode all the possible positions of the center of the 
ring; thus the configuration space is also double-sheeted, and in fact isomorphic to the Sommerfeld space obtained by a constant-time snapshot
of the z$G$KN spacetime as one can easily show.
  We remark that \emph{formally} the Dirac equation for a quasi-statically moving ring singularity is therefore a bispinor-valued wave 
equation on the configuration spacetime which is isomorphic to z$G$KN. 
It is important not to confuse this configuration spacetime with a physical z$G$KN spacetime, though. 

\noindent{\textbf{1.2.4 Enter Dirac's equation}}.

\noindent
        To test our bi-particle interpretation of the ring singularity of the constant-time snapshot of the maximal-analytically extended
z$G$KN spacetime as an electron/anti-electron structure whose motion relative to an $\alpha$ frame is governed by Dirac's 
equation we study the pertinent general-relativistic zero-gravity Hydrogen problem in the usual Born--Oppenheimer approximation.
        We first show that the classical electromagnetic interaction of a static z$G$KN ring singularity 
with a point charge located at the $\alpha$ point 
in the static spacelike slice of z$G$KN is given by the minimal coupling formula used to describe the
interaction of a test point charge in the field of a given z$G$KN singularity --- to avoid confusion with the ``naive minimal coupling''
evaluation of the Coulomb field of the point charge at the center of the z$G$KN ring singularity (times its charge parameter),
we will call our interaction formula a ``minimal re-coupling'' term.
        Next we show that although the z$G$KN singularity is not point- but ring-like, \emph{in the quasi-static motions regime}
one and the same Dirac equation covers these two interpretations of the ``point charge plus z$G$KN ring singularity'' system
(i.e., only the  narrative changes); in particular, the eigenstates are captured correctly. 
        In another accompanying paper \cite{KieTah14a} we have rigorously studied this Dirac equation (in the interpretation ``test point
charge in the z$G$KN spacetime''), and in the present paper the results
of \cite{KieTah14a} will be explained in the novel interpretation put forward here.
        In addition, here we will argue --- compellingly as we believe --- that our spectral results suggest as choice for the ring diameter
of the electron/anti-electron bi-particle structure \emph{not} the reduced
 Compton wavelength of the electron, $\hbar/mc$, as proposed in earlier, different studies which aimed at linking the Kerr--Newman 
spacetime with Dirac's electron (see Appendix C), but only a tiny fraction of it.

 Lastly, we will show that with the help of the Dirac bi-spinors de Broglie--Bohm type laws of motion 
and orientation can be formulated for the z$G$KN ring singularity in the quasi-static regime.

\noindent{\textbf{1.2.5 Summary}}.

\noindent
 In this paper \emph{we propose a fundamentally new quantum-mechanical interpretation of Dirac's 
equation as a single-particle equation for a fermionic elementary particle which is both an electron and an anti-electron.}
 We
\begin{itemize}
\item
argue that the primitive ontology of such a bi-particle which is both an electron and a positron is realized by the
general-relativistic electromagnetic ring singularity of a constant-time snapshot of the double-sheeted 
zero-gravity Kerr--Newman (z$G$KN) spacetime;
\item
explain that the {\em physical spacetime} of such a ring-like bi-particle in quasi-static motion with respect to a rest frame
attached to an arbitrarily chosen origin (Neumann's $\alpha$ point) is a two-sheeted spacetime with a ``wiggly'' timelike tubular 
region traced out by the moving ring singularity, such that every constant-time snapshot (w.r.t. the $\alpha$ frame) is a combined
translation and rotation of the constant-time snapshot of the z$G$KN spacetime;
\item
explain that the {\em configuration spacetime} of this ring-like bi-particle in quasi-static motion is isomorphic to the
z$G$KN spacetime --- any point in the pertinent configuration space represents a possible position of the 
{\em center} of the ring singularity, relative to the $\alpha$ frame;
\item
show that a sheet swap in configuration space 
is associated with a \emph{topological spin} operator acting on the bi-spinorial wave function defined on configuration space;
\item
show how the bi-spinorial wave function on configuration space naturally defines a three-frame attached to the center of the ring singularity, 
and thereby an orientation (spin) vector;
\item
formulate quantum laws of motion for the center and orientation of this two-faced particle using
the de Broglie--Bohm approach relative to the $\alpha$ frame, involving a bi-spinorial wave function satisfying Dirac's 
equation on the configurational z$G$KN spacetime;
\item
compute the classical electromagnetic interaction in physical space of a quasi-static ring singularity with 
an infinitely massive point charge (modeling a proton) located at the $\alpha$ point,
showing it is identical to the minimal coupling formula of a point charge in the electromagnetic Appell--Sommerfeld (viz. z$G$KN) fields;
\item
explain that the Dirac spectrum in our model coincides with the spectrum of a  {\em ``Dirac point positron''} in a given, 
fixed background z$G$KN spacetime --- we then invoke the results we have previously obtained in this test-charge situation, namely: 
the self-adjoineness of the Dirac Hamiltonian, the symmetry of its spectrum about 0, 
the essential spectrum being $\RR \setminus [-m,m]$, and the existence of discrete spectrum under suitable smallness assumptions;
\item
argue that the discrete spectrum reproduces the Sommerfeld fine structure formula (in fact, a positive and a negative
copy of it) in the limit of vanishing ring radius, and that the magnitude of effects like hyperfine structure and Lamb(-like) 
shift of the spectral lines put a limit on the size of the ring singularity which corresponds to the identification of the z$G$KN magnetic
moment with the \emph{anomalous magnetic moment of the electron}.
\end{itemize}

\subsection{The structure of the ensuing sections}

 In section 2 we summarize the basic formulas of the z$G$KN spacetimes and their 
electromagnetic fields, and also some straightforward generalizations of the latter. 
  This material is taken from \cite{zGKN}.

 Section 3: We formulate the Dirac equation for a z$G$KN-type ring singularity that interacts with
an infinitely massive static point charge located at the $\alpha$ point in this topologically non-trivial (double-sheeted) manifold. 
 We vindicate the ``minimal re-coupling'' interaction formula introduced here.
 We then summarize the pertinent results obtained in \cite{KieTah14a} for the equivalent Dirac problem of a test point positron in
the field of an infinitely massive z$G$KN singularity:
essential self-adjointness of the Hamiltonian; symmetry of its spectrum about zero; the usual Dirac continuum with a gap; 
a discrete spectrum inside the gap if the ring radius and the coupling constant are small enough ---
in principle numerically computable by ODE methods using the Chandrasekhar--Page--Toop separation-of-variables method in
concert with the Pr\"ufer transform.
  Here, we also show (formally) that in the limit of vanishing diameter of the ring singularity 
the positive part of the discrete spectrum reproduces Sommerfeld's fine structure spectrum (with 
correct Dirac labeling) for Born--Oppenheimer Hydrogen, while the negative part produces the negative of the same
(we also explain why the familiar special-relativistic Dirac spectrum only contains half of it).
 We argue why this result implies that the choice of ring diameter for an electron/anti-electron bi-particle structure 
should only be a tiny fraction of the reduced Compton wavelength of the electron.
 We also explain why a study of finite-size effects of the ring singularity on the spectrum using perturbation
theory may be problematic, while perturbative computations of effects of a Kerr--Newman-anomalous magnetic 
moment on the spectrum is presumably feasible.
 The technical sections \ref{sec:Cartan}, \ref{sec:HamDir}, \ref{sec:Hilbert} are nearly verbatim adapted 
from \cite{KieTah14a} and reproduced here for the convenience of the reader.

Section 4: We formulate de Broglie--Bohm laws of motion and orientation for the z$G$KN-type ring singularity.

Section 5: We conclude with suggesting future work.
 In particular, we include some speculations about the proper two- and many-body theories consisting entirely of 
(semi-classically) electromagnetically interacting z$G$KN ring singularities. 
 The two-body problem is clearly of interest as a putative model for \emph{positronium}, while the many-body 
theory may offer an intriguing novel explanation of why the particle/anti-particle symmetry is broken
in our (part of the) universe.

Appendix: Two appendices supply  technical material, one appendix contrasts our work with earlier proposals to link
the Kerr--Newman spacetime with ``Dirac's equation for the electron.''

 Almost everywhere in this paper we work in spacetime units in which the speed of light in vacuo $c=1$, and in the more 
mathematical parts we will also set $\hbar$, Planck's constant divided by $2\pi$, equal to unity; in some physically important 
formulas we will re-instate both $c$ and $\hbar$. 

\section{$\!\!$z$G$KN spacetimes $\&$ electromagnetic fields, and generalizations}

\subsection{Electromagnetic spacetimes}

An \emph{electromagnetic spacetime} is a triple $(\cM,\bg,\bF)$, consisting of
a four-dimensional manifold $\cM$, a Lorentzian metric $\bg$ on $\cM$ (here with signature $(+,-,-,-)$ in line
with recent mathematical works on the general-relativistic Dirac equation),
and a two-form $\bF$ representing an electromagnetic field defined on $\cM$, altogether solving
the Einstein--Maxwell equations (in units in which $c=1$)
\bna
 R_{\mu\nu}[\bg] - {\ts\frac{1}{2}} R[\bg] g_{\mu\nu} & = & 8\pi G  T_{\mu\nu}[\bg,\bF]\label{eq:Ein}\\
 \nab^\mu F_{\mu\nu} & = & 0.\label{eq:Max}
 \ena
 Here, $R_{\mu\nu}$ denotes the components of the Ricci curvature tensor and $R$ the scalar curvature of the metric $\bg$.
 Finally, $T_{\mu\nu}$ are the components of the trace-free electromagnetic energy(-density)-momentum(-density)-stress tensor:
 \beq\label{eq:Tmunu}
 T_{\mu\nu}
 = {\ts\frac{1}{4\pi}}\left(F_{\mu}^\la {}\star\!F_{\nu\la} -{\ts\frac{1}{4}} g_{\mu\nu} F_{\al\beta} F^{\al\beta}\right),
 \eeq
where $\star$ denotes the Hodge duality map.
 Since $\bT$ is trace-free, viz. $T_{\mu}^{\mu}=0$, the Ricci scalar $R=0$, so that \refeq{eq:Ein} simplifies to
\bna
 R_{\mu\nu}[\bg]  & = & 8\pi G  T_{\mu\nu}[\bg,\bF]\label{eq:EinMax}
 \ena

\subsection{A brief recap of the electromagnetic Kerr--Newman spacetimes}
         
      The ``outer'' {\em Kerr--Newman} (KN for short) family of spacetimes is a three-parameter family of stationary axisymmetric solutions 
$(\cM,\bg,\bF)$ to the Einstein--Maxwell equations.
	The three parameters mentioned are the ADM mass (total energy) $\textsc{m}$, ADM angular momentum per unit mass $a$, and total 
charge $\textsc{q}$; here, ADM stands for Arnowitt, Deser, and Misner, who defined these quantities in terms of surface integrals at
spatial infinity; the charge is similarly defined using Gauss' theorem.
 Of course, the solution $(\cM,\bg,\bF)$ also depends on $G$, though only in combination with $\textsc{m}$ and
$\textsc{q}^2$.
	The KN metric $\bg$ is singular on the timelike cylindrical surface whose cross-section at any fixed $t$ is a circle of (Euclidean) 
radius $\sqrt{a^2+\ka\textsc{q}^2}$, where $\ka = 2G$, and $G$ is Newton's constant of universal gravitation.
	This circle is commonly referred to as the ``ring'' singularity.
   The KN electromagnetic field $\bF$ is also singular on the same ring as the metric, while \emph{asymptotically} (near spatial infinity) 
it becomes indistinguishable from an electric monopole field of moment $\textsc{q}$ and a magnetic dipole field of moment $\textsc{q}a$
in Minkowski spacetime. 

	The causal structure of the maximal-analytical extension of the KN spacetime is quite complex \cite{Car68} (see 
also \cite{HehlREVIEW,ONeillBOOK}), and depends on the relationship 
between the parameters:  For $a^2+\ka\textsc{q}^2 <\ka^2 \textsc{m}^2$ (the subextremal case) there is an ``ergosphere'' region, where 
due to frame dragging the spacetime is no longer stationary.
	In addition, there are two horizons. 
	One is an event horizon (boundary of a black hole region) shielding one asymptotically flat end of the manifold from the
ring singularity and the acausal region surrounding the latter, and the other a Cauchy horizon, beyond which the initial data do not
determine the evolution.
	For $a^2+\ka\textsc{q}^2> \ka^2 \textsc{m}^2$ (the superextremal case) there is still an ergosphere region but there is
no event horizon, and thus the ring singularity is naked.
	The region near the ring is particularly pathological in both sub- and super-extremal cases since it includes closed timelike loops.
	In the superextremal case the presence of these loops turns the entire manifold into a causally vicious set.
        
\subsection{$\!\!$Zero-$G$ limit of Kerr--Newman spacetimes and their electromagnetic fields}

	It is however possible to rid the KN family from all its causal pathologies mentioned above: 
take the limit $G\to 0$ of the KN family in a global chart of asymptotically flat coordinates which 
respect all the symmetries of the spacetime, while fixing the values of the three parameters $\textsc{m,q,}$ and $a$.
        Since $\textsc{m}$ occurs only in the combination $G\textsc{m}$, a
two-parameter family of spacetimes, depending on $a$ and $\textsc{q}$, emerges in the limit, and is denoted as the 
{\em zero-$G$ Kerr--Newman}  (z$G$KN) family.
        In the z$G$KN family there are no horizons, no ergosphere regions, no closed timelike loops, and no causally vicious regions.
	Even though the metrics and fields in this family are still singular on a ring (now of radius $|a|$), the spacetimes do not 
suffer from any of these maladies.
	Indeed, the manifold is locally flat, and the singularity of the metric at a point on the ring is relatively mild, namely conical.

        A key feature associated with this ring singularity which survives the zero-$G$ limit is the {\em Zipoy topology} of the 
double-sheeted maximal-analytical extension of the KN family of spacetimes.
	Remarkably, the entire z$G$KN spacetime can nevertheless be covered by a single chart of coordinates, 
named after Boyer--Lindquist (BL) and usually denoted by $(t,r,\theta,\varphi)$, where the timelike coordinate $t$ and the spacelike 
coordinate $\varphi$ are Killing parameters (i.e. $\frac{\p}{\p t}$ is a timelike and $\frac{\p}{\p\varphi}$ a spacelike Killing field) 
and $(r,\theta,\varphi)$ are {\em oblate spheroidal coordinates} in $t=$const. slices; see Fig.~3.
 Each constant-$t$ slice consists of two copies of $\RR^3$ (the {\em sheets}) that are cross-linked through two 2-discs at $r=0$, as
depicted in Fig.~3; one sheet corresponds to $r>0$ and the other to $r<0$.

\begin{figure}[ht]\centering
\includegraphics[scale=0.3]{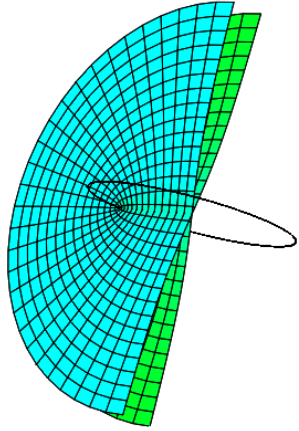} 
\centering
\caption{An immersion of a constant-$\varphi$ section of a constant-$t$ snapshot of the z$G$KN spacetime
together with its ring singularity (parametrized by $\varphi$); also shown is an oblate spheroidal coordinate grid 
--- the semi-ellipses stay on a sheet, while the hyperbolas transit from one sheet to the other.}
\end{figure}

	The KN electromagnetic field (which is in fact independent of $G$), is naturally defined on this branched space.
	To physicists situated far from the ring, the charge carried by it appears positive in one sheet, and negative 
in the other.
	Thus, in the same vein in which a Coulomb singularity in Minkowski spacetime is traditionally interpreted as representing
an electrically charged point particle, the ring singularity of the Kerr--Newman spacetime in its zero-$G$ limit can be 
thought of as a particle which, when viewed from spacelike infinity, appears as an electrically charged point particle with 
a magnetic dipole moment, yet there are two different asymptotically flat ends, and in one such end the particle appears as 
the anti-particle of what is visible in the other sheet.

        \emph{The above description of the electromagnetic Kerr--Newman singularity in its zero-$G$ limit makes it plain that
general relativity supplies a singularity with an electromagnetic structure which suggests the interpretation of electron and 
positron as being just two ``different sides of the same medal,'' or in analogy to Heisenberg's iso-spin concept, 
two different ``iso-spin'' states of one and the same more fundamental particle. 
    Since Heisenberg's iso-spin referred to nucleons (places in the chart of ``isotopes,'' or rather ``isobars''), 
we don't want to use ``iso-spin'' for the suggested electron-positron dichotomy.
    Instead, since the binary character of the Kerr--Newman singularity is associated with the non-trivial topology 
of its spacetime, we will use the term ``topo-spin'' (short for ``topological spin'').
        This term will be vindicated below: we will show that there is a sheet-swap transformation whose
derivative induces a representation of the Lorentz group acting on the tangent bundle which is 
equivalent to the usual spin $1/2$ representation.}

\subsubsection{The metric of the maximal-analytically extended z$G$KN spacetime}

We begin by recalling some standard definitions and results:

Let $\cM = (\Rset^{1,3}\setminus\partial\cZ) \sharp (\Rset^{1,3}\setminus\partial\cZ)$  be a connected sum 
of two copies of the Minkowski spacetime from each of which the boundary $\partial\cZ$ of a timelike cylinder 
$\cZ$ (2-disc $\times$ $\Rset$) has been removed, and which are cross-linked through the interior of the cylinders in 
the manner depicted in Fig.~4 for a constant-time snapshot.

\begin{figure}[ht]
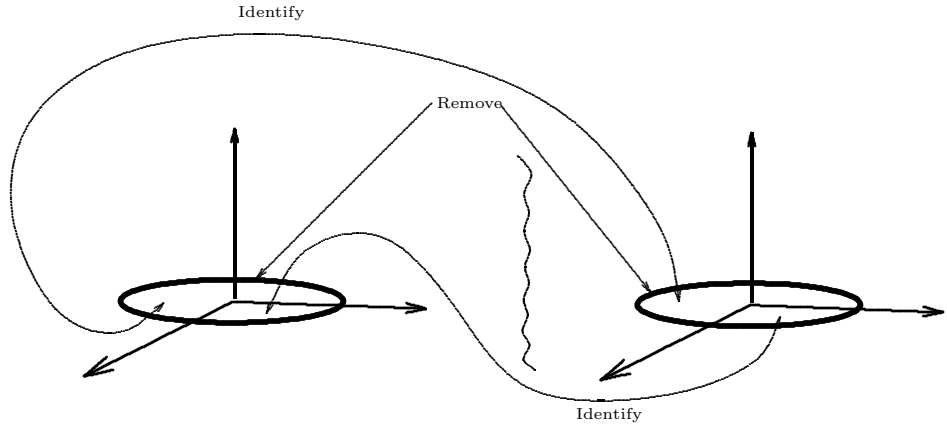

\centering
\font\thinlinefont=cmr5
\begingroup\makeatletter\ifx\SetFigFont\undefined%
\gdef\SetFigFont#1#2#3#4#5{%
  \reset@font\fontsize{#1}{#2pt}%
  \fontfamily{#3}\fontseries{#4}\fontshape{#5}%
  \selectfont}%
\fi\endgroup%
\mbox{\beginpicture
\setcoordinatesystem units <0.50000cm,0.50000cm>
\unitlength=0.50000cm
\linethickness=1pt
\setplotsymbol ({\makebox(0,0)[l]{\tencirc\symbol{'160}}})
\setshadesymbol ({\thinlinefont .})
\setlinear
%
%
\linethickness=2pt
\setplotsymbol ({\makebox(0,0)[l]{\tencirc\symbol{'161}}})
{\color[rgb]{0,0,0}\ellipticalarc axes ratio  2.965:0.574  360 degrees 
	from 17.877 11.779 center at 14.912 11.779
}%
%
%
\linethickness=2pt
\setplotsymbol ({\makebox(0,0)[l]{\tencirc\symbol{'161}}})
{\color[rgb]{0,0,0}\ellipticalarc axes ratio  2.965:0.574  360 degrees 
	from 31.625 11.684 center at 28.660 11.684
}%
%
%
\linethickness= 0.500pt
\setplotsymbol ({\thinlinefont .})
{\color[rgb]{0,0,0}\plot 22.521 15.646 22.528 15.640 /
\plot 22.528 15.640 22.543 15.623 /
\plot 22.543 15.623 22.568 15.598 /
\plot 22.568 15.598 22.600 15.562 /
\plot 22.600 15.562 22.640 15.519 /
\plot 22.640 15.519 22.682 15.471 /
\plot 22.682 15.471 22.722 15.424 /
\plot 22.722 15.424 22.761 15.378 /
\plot 22.761 15.378 22.792 15.335 /
\plot 22.792 15.335 22.820 15.295 /
\plot 22.820 15.295 22.841 15.259 /
\plot 22.841 15.259 22.856 15.227 /
\plot 22.856 15.227 22.866 15.196 /
\plot 22.866 15.196 22.873 15.166 /
\plot 22.873 15.166 22.875 15.136 /
\plot 22.875 15.136 22.873 15.105 /
\plot 22.873 15.105 22.866 15.071 /
\plot 22.866 15.071 22.856 15.037 /
\plot 22.856 15.037 22.841 15.001 /
\plot 22.841 15.001 22.824 14.963 /
\plot 22.824 14.963 22.805 14.927 /
\plot 22.805 14.927 22.784 14.887 /
\plot 22.784 14.887 22.765 14.848 /
\plot 22.765 14.848 22.744 14.812 /
\plot 22.744 14.812 22.727 14.774 /
\plot 22.727 14.774 22.710 14.738 /
\plot 22.710 14.738 22.697 14.704 /
\plot 22.697 14.704 22.689 14.671 /
\plot 22.689 14.671 22.682 14.637 /
\plot 22.682 14.637 22.680 14.605 /
\plot 22.680 14.605 22.684 14.571 /
\plot 22.684 14.571 22.691 14.535 /
\plot 22.691 14.535 22.699 14.499 /
\plot 22.699 14.499 22.712 14.463 /
\plot 22.712 14.463 22.727 14.423 /
\plot 22.727 14.423 22.741 14.385 /
\plot 22.741 14.385 22.756 14.347 /
\plot 22.756 14.347 22.771 14.307 /
\plot 22.771 14.307 22.784 14.271 /
\plot 22.784 14.271 22.796 14.235 /
\plot 22.796 14.235 22.805 14.199 /
\plot 22.805 14.199 22.811 14.165 /
\plot 22.811 14.165 22.813 14.131 /
\putrule from 22.813 14.131 to 22.813 14.099
\plot 22.813 14.099 22.807 14.065 /
\plot 22.807 14.065 22.801 14.031 /
\plot 22.801 14.031 22.790 13.995 /
\plot 22.790 13.995 22.777 13.957 /
\plot 22.777 13.957 22.763 13.921 /
\plot 22.763 13.921 22.748 13.883 /
\plot 22.748 13.883 22.733 13.845 /
\plot 22.733 13.845 22.718 13.807 /
\plot 22.718 13.807 22.705 13.771 /
\plot 22.705 13.771 22.695 13.735 /
\plot 22.695 13.735 22.689 13.701 /
\plot 22.689 13.701 22.682 13.667 /
\putrule from 22.682 13.667 to 22.682 13.636
\plot 22.682 13.636 22.684 13.604 /
\plot 22.684 13.604 22.691 13.570 /
\plot 22.691 13.570 22.699 13.536 /
\plot 22.699 13.536 22.712 13.500 /
\plot 22.712 13.500 22.725 13.464 /
\plot 22.725 13.464 22.741 13.428 /
\plot 22.741 13.428 22.758 13.390 /
\plot 22.758 13.390 22.775 13.352 /
\plot 22.775 13.352 22.790 13.316 /
\plot 22.790 13.316 22.805 13.280 /
\plot 22.805 13.280 22.818 13.246 /
\plot 22.818 13.246 22.826 13.212 /
\plot 22.826 13.212 22.832 13.180 /
\plot 22.832 13.180 22.837 13.149 /
\putrule from 22.837 13.149 to 22.837 13.117
\plot 22.837 13.117 22.832 13.085 /
\plot 22.832 13.085 22.824 13.053 /
\plot 22.824 13.053 22.816 13.020 /
\plot 22.816 13.020 22.803 12.986 /
\plot 22.803 12.986 22.790 12.952 /
\plot 22.790 12.952 22.775 12.916 /
\plot 22.775 12.916 22.761 12.882 /
\plot 22.761 12.882 22.748 12.846 /
\plot 22.748 12.846 22.733 12.814 /
\plot 22.733 12.814 22.722 12.780 /
\plot 22.722 12.780 22.714 12.751 /
\plot 22.714 12.751 22.708 12.719 /
\plot 22.708 12.719 22.705 12.692 /
\putrule from 22.705 12.692 to 22.705 12.662
\plot 22.705 12.662 22.710 12.632 /
\plot 22.710 12.632 22.714 12.601 /
\plot 22.714 12.601 22.722 12.571 /
\plot 22.722 12.571 22.733 12.537 /
\plot 22.733 12.537 22.746 12.503 /
\plot 22.746 12.503 22.758 12.471 /
\plot 22.758 12.471 22.769 12.438 /
\plot 22.769 12.438 22.782 12.404 /
\plot 22.782 12.404 22.792 12.372 /
\plot 22.792 12.372 22.801 12.340 /
\plot 22.801 12.340 22.805 12.308 /
\plot 22.805 12.308 22.809 12.279 /
\putrule from 22.809 12.279 to 22.809 12.249
\plot 22.809 12.249 22.807 12.222 /
\plot 22.807 12.222 22.801 12.192 /
\plot 22.801 12.192 22.792 12.160 /
\plot 22.792 12.160 22.780 12.129 /
\plot 22.780 12.129 22.767 12.095 /
\plot 22.767 12.095 22.752 12.063 /
\plot 22.752 12.063 22.737 12.027 /
\plot 22.737 12.027 22.722 11.993 /
\plot 22.722 11.993 22.708 11.961 /
\plot 22.708 11.961 22.695 11.927 /
\plot 22.695 11.927 22.684 11.898 /
\plot 22.684 11.898 22.676 11.866 /
\plot 22.676 11.866 22.672 11.839 /
\plot 22.672 11.839 22.670 11.811 /
\putrule from 22.670 11.811 to 22.670 11.783
\plot 22.670 11.783 22.674 11.754 /
\plot 22.674 11.754 22.680 11.726 /
\plot 22.680 11.726 22.691 11.697 /
\plot 22.691 11.697 22.701 11.665 /
\plot 22.701 11.665 22.714 11.635 /
\plot 22.714 11.635 22.727 11.604 /
\plot 22.727 11.604 22.739 11.572 /
\plot 22.739 11.572 22.752 11.542 /
\plot 22.752 11.542 22.763 11.510 /
\plot 22.763 11.510 22.771 11.481 /
\plot 22.771 11.481 22.777 11.453 /
\plot 22.777 11.453 22.780 11.424 /
\putrule from 22.780 11.424 to 22.780 11.396
\plot 22.780 11.396 22.777 11.369 /
\plot 22.777 11.369 22.769 11.341 /
\plot 22.769 11.341 22.761 11.309 /
\plot 22.761 11.309 22.748 11.280 /
\plot 22.748 11.280 22.733 11.248 /
\plot 22.733 11.248 22.718 11.214 /
\plot 22.718 11.214 22.701 11.182 /
\plot 22.701 11.182 22.686 11.148 /
\plot 22.686 11.148 22.672 11.117 /
\plot 22.672 11.117 22.659 11.085 /
\plot 22.659 11.085 22.648 11.055 /
\plot 22.648 11.055 22.642 11.026 /
\plot 22.642 11.026 22.638 11.000 /
\plot 22.638 11.000 22.640 10.973 /
\plot 22.640 10.973 22.644 10.943 /
\plot 22.644 10.943 22.655 10.914 /
\plot 22.655 10.914 22.672 10.882 /
\plot 22.672 10.882 22.691 10.850 /
\plot 22.691 10.850 22.712 10.818 /
\plot 22.712 10.818 22.735 10.787 /
\plot 22.735 10.787 22.758 10.755 /
\plot 22.758 10.755 22.780 10.723 /
\plot 22.780 10.723 22.799 10.693 /
\plot 22.799 10.693 22.813 10.664 /
\plot 22.813 10.664 22.822 10.636 /
\plot 22.822 10.636 22.826 10.607 /
\plot 22.826 10.607 22.824 10.579 /
\plot 22.824 10.579 22.816 10.549 /
\plot 22.816 10.549 22.803 10.520 /
\plot 22.803 10.520 22.786 10.488 /
\plot 22.786 10.488 22.767 10.454 /
\plot 22.767 10.454 22.746 10.422 /
\plot 22.746 10.422 22.725 10.389 /
\plot 22.725 10.389 22.705 10.357 /
\plot 22.705 10.357 22.686 10.325 /
\plot 22.686 10.325 22.674 10.295 /
\plot 22.674 10.295 22.665 10.270 /
\plot 22.665 10.270 22.661 10.243 /
\plot 22.661 10.243 22.663 10.221 /
\plot 22.663 10.221 22.670 10.200 /
\plot 22.670 10.200 22.682 10.177 /
\plot 22.682 10.177 22.701 10.154 /
\plot 22.701 10.154 22.727 10.128 /
\plot 22.727 10.128 22.758 10.101 /
\plot 22.758 10.101 22.796 10.071 /
\plot 22.796 10.071 22.837 10.039 /
\plot 22.837 10.039 22.879 10.010 /
\plot 22.879 10.010 22.917  9.984 /
\plot 22.917  9.984 22.947  9.963 /
\plot 22.947  9.963 22.966  9.953 /
\plot 22.966  9.953 22.972  9.946 /
\putrule from 22.972  9.946 to 22.974  9.946
}%
%
%
\linethickness= 0.500pt
\setplotsymbol ({\thinlinefont .})
{\color[rgb]{0,0,0}%
%
\plot 26.707 12.026 26.799 11.781 26.833 12.041 /
\plot 26.799 11.781 26.793 11.836 /
\plot 26.793 11.836 26.789 11.872 /
\plot 26.789 11.872 26.784 11.917 /
\plot 26.784 11.917 26.776 11.974 /
\plot 26.776 11.974 26.767 12.040 /
\plot 26.767 12.040 26.757 12.112 /
\plot 26.757 12.112 26.746 12.192 /
\plot 26.746 12.192 26.734 12.277 /
\plot 26.734 12.277 26.719 12.366 /
\plot 26.719 12.366 26.704 12.454 /
\plot 26.704 12.454 26.687 12.545 /
\plot 26.687 12.545 26.670 12.634 /
\plot 26.670 12.634 26.653 12.721 /
\plot 26.653 12.721 26.634 12.806 /
\plot 26.634 12.806 26.613 12.888 /
\plot 26.613 12.888 26.594 12.969 /
\plot 26.594 12.969 26.571 13.047 /
\plot 26.571 13.047 26.547 13.123 /
\plot 26.547 13.123 26.524 13.197 /
\plot 26.524 13.197 26.499 13.269 /
\plot 26.499 13.269 26.471 13.341 /
\plot 26.471 13.341 26.441 13.413 /
\plot 26.441 13.413 26.410 13.487 /
\plot 26.410 13.487 26.376 13.559 /
\plot 26.376 13.559 26.340 13.633 /
\plot 26.340 13.633 26.302 13.710 /
\plot 26.302 13.710 26.270 13.769 /
\plot 26.270 13.769 26.238 13.828 /
\plot 26.238 13.828 26.204 13.887 /
\plot 26.204 13.887 26.170 13.951 /
\plot 26.170 13.951 26.132 14.012 /
\plot 26.132 14.012 26.094 14.078 /
\plot 26.094 14.078 26.054 14.144 /
\plot 26.054 14.144 26.014 14.211 /
\plot 26.014 14.211 25.969 14.279 /
\plot 25.969 14.279 25.925 14.351 /
\plot 25.925 14.351 25.876 14.421 /
\plot 25.876 14.421 25.828 14.495 /
\plot 25.828 14.495 25.777 14.569 /
\plot 25.777 14.569 25.724 14.643 /
\plot 25.724 14.643 25.671 14.719 /
\plot 25.671 14.719 25.616 14.796 /
\plot 25.616 14.796 25.559 14.874 /
\plot 25.559 14.874 25.499 14.952 /
\plot 25.499 14.952 25.440 15.030 /
\plot 25.440 15.030 25.379 15.107 /
\plot 25.379 15.107 25.317 15.185 /
\plot 25.317 15.185 25.256 15.263 /
\plot 25.256 15.263 25.193 15.342 /
\plot 25.193 15.342 25.129 15.418 /
\plot 25.129 15.418 25.063 15.494 /
\plot 25.063 15.494 24.998 15.568 /
\plot 24.998 15.568 24.934 15.642 /
\plot 24.934 15.642 24.869 15.714 /
\plot 24.869 15.714 24.803 15.786 /
\plot 24.803 15.786 24.737 15.856 /
\plot 24.737 15.856 24.672 15.926 /
\plot 24.672 15.926 24.608 15.991 /
\plot 24.608 15.991 24.543 16.057 /
\plot 24.543 16.057 24.479 16.121 /
\plot 24.479 16.121 24.414 16.182 /
\plot 24.414 16.182 24.350 16.243 /
\plot 24.350 16.243 24.285 16.303 /
\plot 24.285 16.303 24.221 16.360 /
\plot 24.221 16.360 24.158 16.417 /
\plot 24.158 16.417 24.092 16.472 /
\plot 24.092 16.472 24.028 16.525 /
\plot 24.028 16.525 23.963 16.578 /
\plot 23.963 16.578 23.895 16.631 /
\plot 23.895 16.631 23.829 16.681 /
\plot 23.829 16.681 23.760 16.732 /
\plot 23.760 16.732 23.692 16.783 /
\plot 23.692 16.783 23.620 16.832 /
\plot 23.620 16.832 23.548 16.883 /
\plot 23.548 16.883 23.476 16.931 /
\plot 23.476 16.931 23.400 16.978 /
\plot 23.400 16.978 23.324 17.026 /
\plot 23.324 17.026 23.247 17.073 /
\plot 23.247 17.073 23.167 17.122 /
\plot 23.167 17.122 23.086 17.168 /
\plot 23.086 17.168 23.004 17.213 /
\plot 23.004 17.213 22.921 17.259 /
\plot 22.921 17.259 22.837 17.304 /
\plot 22.837 17.304 22.752 17.348 /
\plot 22.752 17.348 22.665 17.391 /
\plot 22.665 17.391 22.578 17.433 /
\plot 22.578 17.433 22.490 17.475 /
\plot 22.490 17.475 22.401 17.515 /
\plot 22.401 17.515 22.312 17.556 /
\plot 22.312 17.556 22.223 17.594 /
\plot 22.223 17.594 22.134 17.632 /
\plot 22.134 17.632 22.043 17.670 /
\plot 22.043 17.670 21.954 17.706 /
\plot 21.954 17.706 21.865 17.742 /
\plot 21.865 17.742 21.774 17.776 /
\plot 21.774 17.776 21.685 17.808 /
\plot 21.685 17.808 21.596 17.841 /
\plot 21.596 17.841 21.507 17.873 /
\plot 21.507 17.873 21.419 17.903 /
\plot 21.419 17.903 21.330 17.932 /
\plot 21.330 17.932 21.243 17.962 /
\plot 21.243 17.962 21.154 17.990 /
\plot 21.154 17.990 21.065 18.017 /
\plot 21.065 18.017 20.976 18.045 /
\plot 20.976 18.045 20.887 18.072 /
\plot 20.887 18.072 20.798 18.098 /
\plot 20.798 18.098 20.720 18.121 /
\plot 20.720 18.121 20.640 18.142 /
\plot 20.640 18.142 20.559 18.165 /
\plot 20.559 18.165 20.479 18.189 /
\plot 20.479 18.189 20.396 18.210 /
\plot 20.396 18.210 20.312 18.233 /
\plot 20.312 18.233 20.227 18.254 /
\plot 20.227 18.254 20.140 18.275 /
\plot 20.140 18.275 20.051 18.299 /
\plot 20.051 18.299 19.962 18.320 /
\plot 19.962 18.320 19.869 18.343 /
\plot 19.869 18.343 19.776 18.364 /
\plot 19.776 18.364 19.683 18.385 /
\plot 19.683 18.385 19.586 18.407 /
\plot 19.586 18.407 19.490 18.430 /
\plot 19.490 18.430 19.391 18.451 /
\plot 19.391 18.451 19.291 18.472 /
\plot 19.291 18.472 19.190 18.491 /
\plot 19.190 18.491 19.088 18.512 /
\plot 19.088 18.512 18.984 18.534 /
\plot 18.984 18.534 18.881 18.553 /
\plot 18.881 18.553 18.777 18.572 /
\plot 18.777 18.572 18.673 18.591 /
\plot 18.673 18.591 18.567 18.610 /
\plot 18.567 18.610 18.462 18.629 /
\plot 18.462 18.629 18.356 18.646 /
\plot 18.356 18.646 18.250 18.663 /
\plot 18.250 18.663 18.146 18.680 /
\plot 18.146 18.680 18.040 18.694 /
\plot 18.040 18.694 17.937 18.709 /
\plot 17.937 18.709 17.833 18.724 /
\plot 17.833 18.724 17.729 18.739 /
\plot 17.729 18.739 17.628 18.752 /
\plot 17.628 18.752 17.526 18.764 /
\plot 17.526 18.764 17.424 18.777 /
\plot 17.424 18.777 17.325 18.788 /
\plot 17.325 18.788 17.228 18.800 /
\plot 17.228 18.800 17.128 18.809 /
\plot 17.128 18.809 17.033 18.819 /
\plot 17.033 18.819 16.935 18.828 /
\plot 16.935 18.828 16.840 18.836 /
\plot 16.840 18.836 16.747 18.843 /
\plot 16.747 18.843 16.654 18.849 /
\plot 16.654 18.849 16.561 18.855 /
\plot 16.561 18.855 16.468 18.862 /
\plot 16.468 18.862 16.375 18.866 /
\plot 16.375 18.866 16.279 18.872 /
\plot 16.279 18.872 16.182 18.874 /
\plot 16.182 18.874 16.085 18.879 /
\plot 16.085 18.879 15.987 18.881 /
\plot 15.987 18.881 15.890 18.883 /
\plot 15.890 18.883 15.790 18.885 /
\putrule from 15.790 18.885 to 15.691 18.885
\putrule from 15.691 18.885 to 15.589 18.885
\putrule from 15.589 18.885 to 15.490 18.885
\plot 15.490 18.885 15.386 18.883 /
\plot 15.386 18.883 15.284 18.881 /
\plot 15.284 18.881 15.181 18.879 /
\plot 15.181 18.879 15.075 18.874 /
\plot 15.075 18.874 14.971 18.870 /
\plot 14.971 18.870 14.865 18.864 /
\plot 14.865 18.864 14.760 18.857 /
\plot 14.760 18.857 14.652 18.851 /
\plot 14.652 18.851 14.546 18.843 /
\plot 14.546 18.843 14.438 18.834 /
\plot 14.438 18.834 14.332 18.826 /
\plot 14.332 18.826 14.226 18.815 /
\plot 14.226 18.815 14.120 18.804 /
\plot 14.120 18.804 14.014 18.792 /
\plot 14.014 18.792 13.909 18.779 /
\plot 13.909 18.779 13.805 18.766 /
\plot 13.805 18.766 13.703 18.752 /
\plot 13.703 18.752 13.602 18.737 /
\plot 13.602 18.737 13.500 18.722 /
\plot 13.500 18.722 13.401 18.705 /
\plot 13.401 18.705 13.303 18.688 /
\plot 13.303 18.688 13.208 18.671 /
\plot 13.208 18.671 13.115 18.652 /
\plot 13.115 18.652 13.022 18.635 /
\plot 13.022 18.635 12.931 18.616 /
\plot 12.931 18.616 12.842 18.595 /
\plot 12.842 18.595 12.755 18.576 /
\plot 12.755 18.576 12.670 18.555 /
\plot 12.670 18.555 12.588 18.534 /
\plot 12.588 18.534 12.505 18.512 /
\plot 12.505 18.512 12.425 18.489 /
\plot 12.425 18.489 12.347 18.466 /
\plot 12.347 18.466 12.268 18.443 /
\plot 12.268 18.443 12.192 18.419 /
\plot 12.192 18.419 12.118 18.394 /
\plot 12.118 18.394 12.031 18.364 /
\plot 12.031 18.364 11.946 18.335 /
\plot 11.946 18.335 11.864 18.303 /
\plot 11.864 18.303 11.779 18.269 /
\plot 11.779 18.269 11.697 18.235 /
\plot 11.697 18.235 11.614 18.199 /
\plot 11.614 18.199 11.532 18.163 /
\plot 11.532 18.163 11.449 18.125 /
\plot 11.449 18.125 11.369 18.085 /
\plot 11.369 18.085 11.288 18.045 /
\plot 11.288 18.045 11.206 18.002 /
\plot 11.206 18.002 11.125 17.958 /
\plot 11.125 17.958 11.047 17.913 /
\plot 11.047 17.913 10.966 17.867 /
\plot 10.966 17.867 10.888 17.818 /
\plot 10.888 17.818 10.812 17.769 /
\plot 10.812 17.769 10.736 17.719 /
\plot 10.736 17.719 10.662 17.668 /
\plot 10.662 17.668 10.588 17.615 /
\plot 10.588 17.615 10.516 17.562 /
\plot 10.516 17.562 10.446 17.509 /
\plot 10.446 17.509 10.376 17.454 /
\plot 10.376 17.454 10.310 17.399 /
\plot 10.310 17.399 10.245 17.344 /
\plot 10.245 17.344 10.183 17.289 /
\plot 10.183 17.289 10.122 17.234 /
\plot 10.122 17.234 10.063 17.177 /
\plot 10.063 17.177 10.008 17.122 /
\plot 10.008 17.122  9.953 17.065 /
\plot  9.953 17.065  9.902 17.010 /
\plot  9.902 17.010  9.851 16.954 /
\plot  9.851 16.954  9.804 16.897 /
\plot  9.804 16.897  9.758 16.842 /
\plot  9.758 16.842  9.716 16.787 /
\plot  9.716 16.787  9.673 16.732 /
\plot  9.673 16.732  9.633 16.675 /
\plot  9.633 16.675  9.595 16.620 /
\plot  9.595 16.620  9.559 16.563 /
\plot  9.559 16.563  9.519 16.499 /
\plot  9.519 16.499  9.481 16.436 /
\plot  9.481 16.436  9.445 16.370 /
\plot  9.445 16.370  9.409 16.305 /
\plot  9.409 16.305  9.375 16.237 /
\plot  9.375 16.237  9.343 16.167 /
\plot  9.343 16.167  9.311 16.097 /
\plot  9.311 16.097  9.282 16.023 /
\plot  9.282 16.023  9.254 15.949 /
\plot  9.254 15.949  9.227 15.875 /
\plot  9.227 15.875  9.201 15.797 /
\plot  9.201 15.797  9.178 15.718 /
\plot  9.178 15.718  9.157 15.638 /
\plot  9.157 15.638  9.136 15.558 /
\plot  9.136 15.558  9.119 15.475 /
\plot  9.119 15.475  9.102 15.392 /
\plot  9.102 15.392  9.087 15.308 /
\plot  9.087 15.308  9.072 15.225 /
\plot  9.072 15.225  9.061 15.141 /
\plot  9.061 15.141  9.053 15.056 /
\plot  9.053 15.056  9.045 14.971 /
\plot  9.045 14.971  9.038 14.889 /
\plot  9.038 14.889  9.036 14.804 /
\plot  9.036 14.804  9.034 14.721 /
\plot  9.034 14.721  9.032 14.641 /
\plot  9.032 14.641  9.034 14.558 /
\plot  9.034 14.558  9.038 14.478 /
\plot  9.038 14.478  9.042 14.400 /
\plot  9.042 14.400  9.049 14.321 /
\plot  9.049 14.321  9.055 14.243 /
\plot  9.055 14.243  9.066 14.167 /
\plot  9.066 14.167  9.076 14.089 /
\plot  9.076 14.089  9.087 14.012 /
\plot  9.087 14.012  9.102 13.936 /
\plot  9.102 13.936  9.114 13.862 /
\plot  9.114 13.862  9.131 13.784 /
\plot  9.131 13.784  9.150 13.708 /
\plot  9.150 13.708  9.169 13.629 /
\plot  9.169 13.629  9.191 13.551 /
\plot  9.191 13.551  9.212 13.470 /
\plot  9.212 13.470  9.237 13.392 /
\plot  9.237 13.392  9.263 13.310 /
\plot  9.263 13.310  9.290 13.229 /
\plot  9.290 13.229  9.320 13.147 /
\plot  9.320 13.147  9.351 13.062 /
\plot  9.351 13.062  9.383 12.979 /
\plot  9.383 12.979  9.417 12.897 /
\plot  9.417 12.897  9.453 12.812 /
\plot  9.453 12.812  9.489 12.730 /
\plot  9.489 12.730  9.529 12.647 /
\plot  9.529 12.647  9.567 12.567 /
\plot  9.567 12.567  9.608 12.486 /
\plot  9.608 12.486  9.650 12.408 /
\plot  9.650 12.408  9.692 12.330 /
\plot  9.692 12.330  9.735 12.253 /
\plot  9.735 12.253  9.779 12.181 /
\plot  9.779 12.181  9.821 12.109 /
\plot  9.821 12.109  9.866 12.040 /
\plot  9.866 12.040  9.910 11.974 /
\plot  9.910 11.974  9.955 11.910 /
\plot  9.955 11.910  9.999 11.849 /
\plot  9.999 11.849 10.046 11.790 /
\plot 10.046 11.790 10.090 11.733 /
\plot 10.090 11.733 10.135 11.680 /
\plot 10.135 11.680 10.179 11.629 /
\plot 10.179 11.629 10.224 11.578 /
\plot 10.224 11.578 10.268 11.532 /
\plot 10.268 11.532 10.312 11.487 /
\plot 10.312 11.487 10.365 11.441 /
\plot 10.365 11.441 10.416 11.394 /
\plot 10.416 11.394 10.469 11.352 /
\plot 10.469 11.352 10.524 11.309 /
\plot 10.524 11.309 10.579 11.271 /
\plot 10.579 11.271 10.634 11.233 /
\plot 10.634 11.233 10.693 11.199 /
\plot 10.693 11.199 10.751 11.165 /
\plot 10.751 11.165 10.810 11.134 /
\plot 10.810 11.134 10.871 11.106 /
\plot 10.871 11.106 10.933 11.079 /
\plot 10.933 11.079 10.994 11.053 /
\plot 10.994 11.053 11.055 11.032 /
\plot 11.055 11.032 11.117 11.011 /
\plot 11.117 11.011 11.180 10.992 /
\plot 11.180 10.992 11.242 10.977 /
\plot 11.242 10.977 11.303 10.964 /
\plot 11.303 10.964 11.362 10.952 /
\plot 11.362 10.952 11.424 10.943 /
\plot 11.424 10.943 11.481 10.937 /
\plot 11.481 10.937 11.538 10.930 /
\plot 11.538 10.930 11.595 10.928 /
\putrule from 11.595 10.928 to 11.648 10.928
\putrule from 11.648 10.928 to 11.701 10.928
\plot 11.701 10.928 11.752 10.930 /
\plot 11.752 10.930 11.803 10.935 /
\plot 11.803 10.935 11.849 10.941 /
\plot 11.849 10.941 11.898 10.947 /
\plot 11.898 10.947 11.942 10.956 /
\plot 11.942 10.956 11.987 10.964 /
\plot 11.987 10.964 12.040 10.979 /
\plot 12.040 10.979 12.093 10.996 /
\plot 12.093 10.996 12.145 11.013 /
\plot 12.145 11.013 12.196 11.036 /
\plot 12.196 11.036 12.247 11.060 /
\plot 12.247 11.060 12.300 11.087 /
\plot 12.300 11.087 12.355 11.119 /
\plot 12.355 11.119 12.410 11.155 /
\plot 12.410 11.155 12.467 11.195 /
\plot 12.467 11.195 12.529 11.240 /
\plot 12.529 11.240 12.590 11.286 /
\plot 12.590 11.286 12.653 11.337 /
\plot 12.653 11.337 12.719 11.390 /
\plot 12.719 11.390 12.783 11.445 /
\plot 12.783 11.445 12.844 11.498 /
\plot 12.844 11.498 12.903 11.549 /
\plot 12.903 11.549 12.954 11.595 /
\plot 12.954 11.595 13.001 11.635 /
\plot 13.001 11.635 13.037 11.669 /
\plot 13.037 11.669 13.094 11.722 /
%
%
\plot 12.950 11.503 13.094 11.722 12.864 11.596 /
}%
%
%
\linethickness= 0.500pt
\setplotsymbol ({\thinlinefont .})
{\color[rgb]{0,0,0}%
%
\plot 29.466 11.086 29.483 11.347 29.345 11.125 /
\plot 29.483 11.347 29.462 11.282 /
\plot 29.462 11.282 29.449 11.242 /
\plot 29.449 11.242 29.432 11.193 /
\plot 29.432 11.193 29.411 11.136 /
\plot 29.411 11.136 29.388 11.072 /
\plot 29.388 11.072 29.362 11.005 /
\plot 29.362 11.005 29.335 10.935 /
\plot 29.335 10.935 29.307 10.867 /
\plot 29.307 10.867 29.278 10.799 /
\plot 29.278 10.799 29.248 10.736 /
\plot 29.248 10.736 29.218 10.674 /
\plot 29.218 10.674 29.187 10.617 /
\plot 29.187 10.617 29.155 10.564 /
\plot 29.155 10.564 29.123 10.513 /
\plot 29.123 10.513 29.089 10.465 /
\plot 29.089 10.465 29.053 10.418 /
\plot 29.053 10.418 29.015 10.376 /
\plot 29.015 10.376 28.975 10.334 /
\plot 28.975 10.334 28.933 10.291 /
\plot 28.933 10.291 28.888 10.249 /
\plot 28.888 10.249 28.850 10.217 /
\plot 28.850 10.217 28.812 10.188 /
\plot 28.812 10.188 28.772 10.156 /
\plot 28.772 10.156 28.730 10.124 /
\plot 28.730 10.124 28.685 10.092 /
\plot 28.685 10.092 28.639 10.061 /
\plot 28.639 10.061 28.588 10.027 /
\plot 28.588 10.027 28.535  9.995 /
\plot 28.535  9.995 28.480  9.961 /
\plot 28.480  9.961 28.423  9.929 /
\plot 28.423  9.929 28.361  9.895 /
\plot 28.361  9.895 28.298  9.862 /
\plot 28.298  9.862 28.230  9.828 /
\plot 28.230  9.828 28.160  9.796 /
\plot 28.160  9.796 28.090  9.762 /
\plot 28.090  9.762 28.016  9.730 /
\plot 28.016  9.730 27.940  9.699 /
\plot 27.940  9.699 27.862  9.667 /
\plot 27.862  9.667 27.781  9.635 /
\plot 27.781  9.635 27.701  9.605 /
\plot 27.701  9.605 27.618  9.576 /
\plot 27.618  9.576 27.534  9.546 /
\plot 27.534  9.546 27.449  9.519 /
\plot 27.449  9.519 27.362  9.491 /
\plot 27.362  9.491 27.275  9.466 /
\plot 27.275  9.466 27.189  9.440 /
\plot 27.189  9.440 27.100  9.417 /
\plot 27.100  9.417 27.011  9.394 /
\plot 27.011  9.394 26.920  9.373 /
\plot 26.920  9.373 26.829  9.351 /
\plot 26.829  9.351 26.736  9.332 /
\plot 26.736  9.332 26.642  9.313 /
\plot 26.642  9.313 26.566  9.299 /
\plot 26.566  9.299 26.488  9.284 /
\plot 26.488  9.284 26.408  9.269 /
\plot 26.408  9.269 26.327  9.256 /
\plot 26.327  9.256 26.245  9.243 /
\plot 26.245  9.243 26.160  9.231 /
\plot 26.160  9.231 26.073  9.218 /
\plot 26.073  9.218 25.984  9.205 /
\plot 25.984  9.205 25.895  9.195 /
\plot 25.895  9.195 25.802  9.184 /
\plot 25.802  9.184 25.709  9.176 /
\plot 25.709  9.176 25.614  9.167 /
\plot 25.614  9.167 25.516  9.159 /
\plot 25.516  9.159 25.419  9.150 /
\plot 25.419  9.150 25.322  9.144 /
\plot 25.322  9.144 25.220  9.140 /
\plot 25.220  9.140 25.121  9.136 /
\plot 25.121  9.136 25.019  9.131 /
\plot 25.019  9.131 24.920  9.129 /
\plot 24.920  9.129 24.818  9.127 /
\putrule from 24.818  9.127 to 24.716  9.127
\putrule from 24.716  9.127 to 24.617  9.127
\plot 24.617  9.127 24.517  9.129 /
\plot 24.517  9.129 24.418  9.133 /
\plot 24.418  9.133 24.320  9.138 /
\plot 24.320  9.138 24.225  9.142 /
\plot 24.225  9.142 24.130  9.150 /
\plot 24.130  9.150 24.037  9.157 /
\plot 24.037  9.157 23.948  9.167 /
\plot 23.948  9.167 23.859  9.176 /
\plot 23.859  9.176 23.772  9.188 /
\plot 23.772  9.188 23.688  9.201 /
\plot 23.688  9.201 23.605  9.214 /
\plot 23.605  9.214 23.525  9.229 /
\plot 23.525  9.229 23.446  9.246 /
\plot 23.446  9.246 23.372  9.263 /
\plot 23.372  9.263 23.298  9.282 /
\plot 23.298  9.282 23.226  9.303 /
\plot 23.226  9.303 23.158  9.324 /
\plot 23.158  9.324 23.091  9.345 /
\plot 23.091  9.345 23.017  9.373 /
\plot 23.017  9.373 22.945  9.400 /
\plot 22.945  9.400 22.877  9.432 /
\plot 22.877  9.432 22.807  9.464 /
\plot 22.807  9.464 22.741  9.500 /
\plot 22.741  9.500 22.674  9.536 /
\plot 22.674  9.536 22.610  9.574 /
\plot 22.610  9.574 22.545  9.614 /
\plot 22.545  9.614 22.481  9.658 /
\plot 22.481  9.658 22.420  9.703 /
\plot 22.420  9.703 22.356  9.749 /
\plot 22.356  9.749 22.297  9.798 /
\plot 22.297  9.798 22.236  9.851 /
\plot 22.236  9.851 22.176  9.904 /
\plot 22.176  9.904 22.117  9.957 /
\plot 22.117  9.957 22.060 10.014 /
\plot 22.060 10.014 22.003 10.071 /
\plot 22.003 10.071 21.948 10.130 /
\plot 21.948 10.130 21.893 10.192 /
\plot 21.893 10.192 21.840 10.253 /
\plot 21.840 10.253 21.787 10.315 /
\plot 21.787 10.315 21.734 10.376 /
\plot 21.734 10.376 21.683 10.439 /
\plot 21.683 10.439 21.634 10.503 /
\plot 21.634 10.503 21.586 10.566 /
\plot 21.586 10.566 21.539 10.630 /
\plot 21.539 10.630 21.493 10.693 /
\plot 21.493 10.693 21.446 10.757 /
\plot 21.446 10.757 21.402 10.820 /
\plot 21.402 10.820 21.357 10.884 /
\plot 21.357 10.884 21.315 10.945 /
\plot 21.315 10.945 21.273 11.009 /
\plot 21.273 11.009 21.228 11.070 /
\plot 21.228 11.070 21.186 11.132 /
\plot 21.186 11.132 21.143 11.195 /
\plot 21.143 11.195 21.101 11.256 /
\plot 21.101 11.256 21.059 11.318 /
\plot 21.059 11.318 21.016 11.379 /
\plot 21.016 11.379 20.972 11.443 /
\plot 20.972 11.443 20.925 11.506 /
\plot 20.925 11.506 20.881 11.570 /
\plot 20.881 11.570 20.834 11.633 /
\plot 20.834 11.633 20.786 11.699 /
\plot 20.786 11.699 20.737 11.764 /
\plot 20.737 11.764 20.688 11.832 /
\plot 20.688 11.832 20.637 11.898 /
\plot 20.637 11.898 20.585 11.966 /
\plot 20.585 11.966 20.534 12.033 /
\plot 20.534 12.033 20.479 12.101 /
\plot 20.479 12.101 20.426 12.169 /
\plot 20.426 12.169 20.371 12.234 /
\plot 20.371 12.234 20.316 12.302 /
\plot 20.316 12.302 20.261 12.368 /
\plot 20.261 12.368 20.204 12.433 /
\plot 20.204 12.433 20.149 12.497 /
\plot 20.149 12.497 20.091 12.560 /
\plot 20.091 12.560 20.036 12.622 /
\plot 20.036 12.622 19.979 12.683 /
\plot 19.979 12.683 19.924 12.740 /
\plot 19.924 12.740 19.869 12.797 /
\plot 19.869 12.797 19.814 12.852 /
\plot 19.814 12.852 19.759 12.903 /
\plot 19.759 12.903 19.706 12.954 /
\plot 19.706 12.954 19.653 13.003 /
\plot 19.653 13.003 19.602 13.047 /
\plot 19.602 13.047 19.550 13.092 /
\plot 19.550 13.092 19.499 13.132 /
\plot 19.499 13.132 19.450 13.172 /
\plot 19.450 13.172 19.399 13.208 /
\plot 19.399 13.208 19.351 13.244 /
\plot 19.351 13.244 19.302 13.276 /
\plot 19.302 13.276 19.255 13.305 /
\plot 19.255 13.305 19.192 13.343 /
\plot 19.192 13.343 19.130 13.377 /
\plot 19.130 13.377 19.067 13.407 /
\plot 19.067 13.407 19.006 13.434 /
\plot 19.006 13.434 18.942 13.462 /
\plot 18.942 13.462 18.876 13.483 /
\plot 18.876 13.483 18.811 13.504 /
\plot 18.811 13.504 18.745 13.523 /
\plot 18.745 13.523 18.680 13.538 /
\plot 18.680 13.538 18.612 13.551 /
\plot 18.612 13.551 18.546 13.561 /
\plot 18.546 13.561 18.479 13.570 /
\plot 18.479 13.570 18.411 13.576 /
\plot 18.411 13.576 18.345 13.581 /
\plot 18.345 13.581 18.280 13.583 /
\putrule from 18.280 13.583 to 18.214 13.583
\plot 18.214 13.583 18.150 13.581 /
\plot 18.150 13.581 18.089 13.576 /
\plot 18.089 13.576 18.028 13.570 /
\plot 18.028 13.570 17.968 13.564 /
\plot 17.968 13.564 17.911 13.555 /
\plot 17.911 13.555 17.856 13.545 /
\plot 17.856 13.545 17.803 13.534 /
\plot 17.803 13.534 17.752 13.523 /
\plot 17.752 13.523 17.702 13.511 /
\plot 17.702 13.511 17.653 13.498 /
\plot 17.653 13.498 17.606 13.485 /
\plot 17.606 13.485 17.562 13.470 /
\plot 17.562 13.470 17.509 13.454 /
\plot 17.509 13.454 17.458 13.437 /
\plot 17.458 13.437 17.410 13.418 /
\plot 17.410 13.418 17.359 13.399 /
\plot 17.359 13.399 17.310 13.377 /
\plot 17.310 13.377 17.261 13.354 /
\plot 17.261 13.354 17.213 13.331 /
\plot 17.213 13.331 17.164 13.305 /
\plot 17.164 13.305 17.117 13.280 /
\plot 17.117 13.280 17.069 13.252 /
\plot 17.069 13.252 17.022 13.223 /
\plot 17.022 13.223 16.978 13.193 /
\plot 16.978 13.193 16.933 13.161 /
\plot 16.933 13.161 16.889 13.130 /
\plot 16.889 13.130 16.847 13.096 /
\plot 16.847 13.096 16.806 13.062 /
\plot 16.806 13.062 16.768 13.028 /
\plot 16.768 13.028 16.730 12.992 /
\plot 16.730 12.992 16.694 12.958 /
\plot 16.694 12.958 16.660 12.922 /
\plot 16.660 12.922 16.626 12.884 /
\plot 16.626 12.884 16.595 12.848 /
\plot 16.595 12.848 16.563 12.810 /
\plot 16.563 12.810 16.533 12.770 /
\plot 16.533 12.770 16.506 12.734 /
\plot 16.506 12.734 16.478 12.696 /
\plot 16.478 12.696 16.451 12.656 /
\plot 16.451 12.656 16.423 12.613 /
\plot 16.423 12.613 16.394 12.567 /
\plot 16.394 12.567 16.366 12.518 /
\plot 16.366 12.518 16.336 12.467 /
\plot 16.336 12.467 16.305 12.412 /
\plot 16.305 12.412 16.273 12.351 /
\plot 16.273 12.351 16.237 12.287 /
\plot 16.237 12.287 16.203 12.220 /
\plot 16.203 12.220 16.165 12.145 /
\plot 16.165 12.145 16.127 12.069 /
\plot 16.127 12.069 16.089 11.991 /
\plot 16.089 11.991 16.049 11.913 /
\plot 16.049 11.913 16.010 11.834 /
\plot 16.010 11.834 15.974 11.758 /
\plot 15.974 11.758 15.941 11.688 /
\plot 15.941 11.688 15.911 11.627 /
\plot 15.911 11.627 15.886 11.574 /
\plot 15.886 11.574 15.867 11.534 /
\plot 15.867 11.534 15.835 11.466 /
%
%
\plot 15.885 11.723 15.835 11.466 16.000 11.669 /
}%
%
%
\linethickness= 0.500pt
\setplotsymbol ({\thinlinefont .})
{\color[rgb]{0,0,0}\plot 22.066 16.986 26.035 12.033 /
%
%
\plot 25.827 12.192 26.035 12.033 25.926 12.271 /
}%
%
%
\linethickness= 0.500pt
\setplotsymbol ({\thinlinefont .})
{\color[rgb]{0,0,0}\plot 20.225 17.050 15.589 12.414 /
%
%
\plot 15.724 12.639 15.589 12.414 15.814 12.549 /
}%
%
%
\linethickness=1pt
\setplotsymbol ({\makebox(0,0)[l]{\tencirc\symbol{'160}}})
{\color[rgb]{0,0,0}\plot 14.942 11.764 20.064 11.578 /
%
%
\plot 19.552 11.470 20.064 11.578 19.561 11.724 /
}%
%
%
\linethickness=1pt
\setplotsymbol ({\makebox(0,0)[l]{\tencirc\symbol{'160}}})
{\color[rgb]{0,0,0}\putrule from 14.986 11.843 to 14.986 16.351
%
%
\plot 15.071 16.012 14.986 16.351 14.901 16.012 /
}%
%
%
\linethickness=1pt
\setplotsymbol ({\makebox(0,0)[l]{\tencirc\symbol{'160}}})
{\color[rgb]{0,0,0}\plot 14.912 11.748 10.986  9.779 /
%
%
\plot 11.581 10.291 10.986  9.779 11.752  9.950 /
}%
%
%
\linethickness=1pt
\setplotsymbol ({\makebox(0,0)[l]{\tencirc\symbol{'160}}})
{\color[rgb]{0,0,0}\plot 28.689 11.669 33.812 11.483 /
%
%
\plot 33.300 11.374 33.812 11.483 33.309 11.628 /
}%
%
%
\linethickness=1pt
\setplotsymbol ({\makebox(0,0)[l]{\tencirc\symbol{'160}}})
{\color[rgb]{0,0,0}\putrule from 28.734 11.748 to 28.734 16.256
%
%
\plot 28.819 15.917 28.734 16.256 28.649 15.917 /
}%
%
%
\linethickness=1pt
\setplotsymbol ({\makebox(0,0)[l]{\tencirc\symbol{'160}}})
{\color[rgb]{0,0,0}\plot 28.660 11.652 24.733  9.684 /
%
%
\plot 25.329 10.196 24.733  9.684 25.500  9.855 /
}%
%
%
\put{\SetFigFont{6}{7.2}{\rmdefault}{\mddefault}{\updefault}{\color[rgb]{0,0,0}Remove}%
} [lB] at 20.352 16.923
%
%
\put{\SetFigFont{6}{7.2}{\rmdefault}{\mddefault}{\updefault}{\color[rgb]{0,0,0}Identify}%
} [lB] at 15.081 19.336
%
%
\put{\SetFigFont{6}{7.2}{\rmdefault}{\mddefault}{\updefault}{\color[rgb]{0,0,0}Identify}%
} [lB] at 24.066  8.636
\linethickness=0pt
\putrectangle corners at  9.006 19.704 and 33.881  8.479
\endpicture}
\caption{An illustration of the Zipoy topology of the z$G$KN spacetime at a constant-time snapshot, using two copies of $\Rset^3$. 
The ring singularity, which appears positively charged in one copy of $\Rset^3$, and negatively charged in the other, is not part 
of this Riemannian constant-time manifold.}
\end{figure}

  In terms of coordinate charts, the manifold can be described thus.
  Let $a$ be fixed.
  Let $(t,\varrho,z,\varphi)$ denote cylindrical coordinates on $\Rset^{1,3}$ and let $(t,r,\theta,\varphi)$ be 
BL coordinates on the z$G$KN manifold $\cM$, with the same $(t,\varphi)$ as in cylindrical 
coordinates, and with $(r,\theta)$ elliptical coordinates which are related to $(\varrho,z)$ by 
$$
\varrho = \sqrt{r^2+a^2}\sin\theta,\qquad z = r \cos \theta.
$$
  The identification in $\cM$ is to be done in such a way that in each fixed timelike half-plane $t=t_0,\varphi=\varphi_0$, 
one of the sheets is described by $r\geq 0$ and the other by $r\leq 0$ (the linking at $t=$constant is through a double disk 
at $r=0$; cf. Fig.~4) with all the BL coordinates having smooth transitions from one sheet to the other.
 Note that the coordinate map $(t,r,\theta,\varphi) \mapsto (t,\varrho,z,\varphi)$ is $2:1$.

 We can also view $\cM$ as a bundle over the base manifold $\Rset^{1,3}\setminus\partial\cZ$, with the projection map 
$\Pi: \cM \to \Rset^{1,3}\setminus\partial\cZ$ being $\Pi(t,r,\theta,\varphi) = (t,\varrho,z,\varphi)$.
  The fiber over a point in the base is thus a discrete set of two points; note that $\Pi$ degenerates at
the boundary $\partial\cZ$ of $\cZ$,
$$
\partial\cZ_t = \{(t,r,\theta,\varphi)\ |\ t\in\Rset,\ r=0, \ \theta = \pi/2, 0\leq \varphi\ \leq 2\pi \},
$$
where only one point lies above the extended base manifold $\Rset^{1,3}$ (see Fig.~5).
 Such a bundle is also known as a ``branched cover'' of  $\Rset^{1,3}$, with two copies of $\Rset^{1,3}$ (``branches'') which
are ``branching off each other over a ring.''

\begin{figure}[ht]
\includegraphics[scale=0.3]{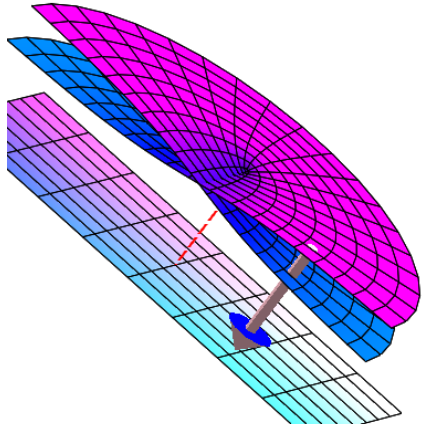}
\centering
\caption{The constant-$t$--constant-$\varphi$ section of the z$G$KN spacetime $\cM$ (with its $(r,\theta)$ coordinate grid)
as a branched cover over the pertinent constant-$t$--constant-$\varphi$ section of $\Rset^{1,3}\setminus\partial\cZ$ (with its
 $(\varrho,z)$ coordinate grid). 
The endpoints of the dashed line mark the locations of the ring singularity in $\cM$ (at $(0,\pi/2)$)
and its base point in $\Rset^{1,3}\setminus\partial\cZ$ (at $(1,0)$); neither are part of the pertinent Lorentz manifolds. 
The arrow indicates the projection $\Pi: \cM \to \Rset^{1,3}\setminus\partial\cZ$.}
\end{figure}

 The pullback of the Minkowski metric $\boldsymbol{\eta}$ under $\Pi$ endows $\cM$ with a flat Lorentzian metric 
$\bg = \Pi^*\boldsymbol{\eta}$, which solves the Einstein vacuum equations, and which in BL coordinates has the line element
\begin{equation}\label{def:metricBL}
ds_\bg^2 = dt^2 - (r^2+a^2) \sin^2\theta d\varphi^2 - \frac{r^2+a^2 \cos^2\theta}{r^2+a^2} (dr^2 + (r^2+a^2) d\theta^2).
\end{equation}

 Incidentally, the spatial part of the BL coordinate system is one representative of so-called \emph{oblate spheroidal} 
coordinates.
 All such coordinate systems differ from each other only in the labeling of the level surfaces; any constant $\varphi$ 
section of these consists of confocal ellipses and hyperbolas.
 Thus, an alternative choice of oblate spheroidal coordinates are the {\em ring-centered} coordinates, defined by 
$$
\xi = \frac{r}{a},\qquad \eta = \cos\theta,
$$ 
and with $(t,\varphi)$ as before.
   In these coordinates the line element (\ref{def:metricBL}) takes the more symmetric form
$$
ds_\bg^2 = dt^2 - a^2(1+\xi^2)(1-\eta^2)d\varphi^2 - a^2(\xi^2+\eta^2) \left(\frac{d\xi^2}{1+\xi^2} + \frac{d\eta^2}{1-\eta^2}\right).
$$
  The metric $\bg$ is singular on the cylindrical surface $\partial\cZ$ whose cross section at any $t$ is the {\em ring} 
$\{r=0, \theta=\pi/2,\varphi\in[0,2\pi)\}$.
 The points on the ring are {\em conical singularities} for the metric, meaning that the limit, as the radius goes to zero, of the 
circumference of a small circle that is centered at a point of the ring and is lying in a meridional plane $\varphi=$const., is 
different from $2\pi$ --- in the case of z$G$KN it is $4\pi$.
  See \cite{zGKN} for details.

 The spacetime $(\cM,\bg)$ is the limiting member, in the limit $G\to 0$, of the {\em Kerr--Newman} family of spacetimes.
 The spacetime $(\cM,\bg)$ introduced above is also the zero-$G$ limit of another family, namely the {\em Kerr} family of 
spacetimes, which 
outside some ergosphere horizon are stationary, axisymmetric solutions to Einstein's vacuum equations.
 The Kerr family is simply the limit $\textsc{q}\to0$ of the KN family (in BL coordinates), and as long as only the spacetime 
structure itself is of interest, we may call $(\cM,\bg)$ the (maximal-analytically extended) z$G$K spacetime.
 This z$G$K spacetime $(\cM,\bg)$ is also the 
zero-$G$ limit of another family, namely the oblate spheroidal {\em Zipoy} family of spacetimes, which are maximal-analytically 
extended static, axisymmetric solutions to Einstein's vacuum equations and not otherwise related to the KN family except for 
having the same zero-$G$ limit (so ``z$G$KN = z$G$K = z$G$Z'').
 In fact, the spacetime $(\cM,\bg)$ introduced above is static, in the sense that it features a timelike Killing field which 
is everywhere hypersurface-orthogonal.
 This shows that the stationary Kerr and Kerr--Newman spacetimes become static when $G\to 0$.

\subsubsection{The z$G$KN electromagnetic fields and some of their generalizations}\label{sec:subsubFIELDSzGKNgen}

The z$G$KN spacetime $(\cM,\bg,\bF)$ is an electromagnetic spacetime, i.e. it comes already equipped with
the  electromagnetic field $\bF_{\textsc{KN}} =d\bAKN$ of the Kerr--Newman family,\footnote{Interestingly,
 the KN fields are independent of $G$ and therefore survive intact in the zero-$G$ limit.} 
with
\beq\label{def:AKN}
\bAKN = - \frac{r}{r^2+a^2\cos^2\theta} (\textsc{q}dt - \textsc{q}a \sin^2\theta\, d\varphi).
\eeq
  The field $\bF$ is singular on the same ring $\{r=0, \theta=\pi/2,\varphi\in[0,2\pi)\}$ as the metric, 
while for $r$ very large it exhibits an electric monopole of strength $\textsc{q}$ and a magnetic dipole moment of strength $\textsc{q}a$;
for $r$ very large negative it exhibits an electric monopole of strength $-\textsc{q}$ and a magnetic dipole moment of strength 
$-\textsc{q}a$.
 As mentioned in the introduction, this static electromagnetic field was discovered by Appell \cite{Appell} in 1887, while
Sommerfeld \cite{Som97} realized that it lives on z$G$KN.

\begin{rem}\emph{The Kerr--Newman spacetime is famously known to have a gyromagnetic ratio $\textsc{q/m} = g \textsc{q/2m}$ corresponding 
to a $g$-factor of $g_{\mbox{\tiny\textrm{KN}}}=2$, the terminology being borrowed from quantum mechanics and motivated by the facts that 
the Kerr--Newman spacetime is associated with an ADM spin angular momentum $a\textsc{m}$, a charge $\textsc{q}$, and a
magnetic moment $\textsc{q}$ (as seen asymptotically ``from infinity''); see \cite{Car68,StraumannBOOK,New02}.
 Since the z$\,G$KN spacetime (or rather any of its $\pdt$-orthogonal space slices) is static and not ``gyrating,'' 
it becomes problematic to speak of a \emph{gyromagnetic ratio} for z$\,G$KN;
also, since no $\textsc{m}$ features in the z$\,G$KN (i.e. z$\,G$K) metric, one would need to introduce new notions of ``mass and
angular momentum of z$\,G$KN.''
 One could try to argue that not the spacetime but the ring singularity is gyrating, with a spin angular momentum equal to $a\textsc{m}$, 
and with $\textsc{m}$ the inert mass of the ring singularity, yet
this has to be taken with a grain of salt, for the electromagnetic field energy of z$\,G$KN is infinite, so that according to 
Einstein's mass-energy equivalence also $\textsc{m}$ ought to be infinite.
 For a resolution of this mass-energy puzzle in a general relativistic spacetime with a single point charge using the nonlinear 
Einstein--Maxwell--Born--Infeld (and other nonlinear) equations, see \cite{Tah11}.}
\end{rem}

 The Kerr--Newman electric field $\bE_{\textsc{KN}}$ and the Kerr--Newman magnetic field $\bB_{\textsc{KN}}$ are gradients,
$$
 \bE_{\textsc{KN}} = d \phi_{\textsc{KN}},\qquad \bB_{\textsc{KN}} = d  \psi_{\textsc{KN}},
$$
where the scalar potentials $\phi_{\textsc{KN}}$ and $\psi_{\textsc{KN}}$ have remarkably simple, and symmetric, expressions 
in oblate 
spheroidal $(\xi,\eta,\varphi)$ coordinates on $\cN$, the $t=0$ slice of z$G$K, namely\footnote{Even though one obtains a 
  remarkable complex structure with these
 formulas (cf. \cite{zGKN}), the representation of $\bB_{\textsc{KN}}$ as a gradient of a scalar potential is problematic at the
 ring singularity because of the condition that $\bB$ be divergence-free.}
$$
\phi_{\textsc{KN}} = \frac{\textsc{q}}{a} \frac{\xi}{\xi^2 + \eta^2},\qquad  
\psi_{\textsc{KN}} =  \frac{\textsc{q}}{a} \frac{\eta}{\xi^2 + \eta^2}.
$$
 Note in particular that these two are anti-symmetric with respect to the ``toggle'' map that swaps the two sheets, viz.
$\varsigma: (\xi,\eta) \mapsto (-\xi,-\eta)$.
 It is therefore evident that in the sheet where $\xi>0$ the asymptotic behavior of $\phi_{\textsc{KN}}$ is
$\frac{\textsc{q}}{|{\bq}-\bq_{\mathrm{rg}}|}$ 
while in the other sheet, where $\xi<0$, the asymptotic behavior becomes 
$\frac{-\textsc{q}\ }{|{\bq}-\bq_{\mathrm{rg}}|}$.
 Thus by Gauss's law the charge carried by the ring is $\textsc{q}$ in the first sheet and $-\textsc{q}$ in the second.

  Now we note that by the decoupling of spacetime structure from its matter/field content in the zero-$G$ limit, 
by the linearity of Maxwell's vacuum equations, and by the decoupling of their electric
and magnetic subsystems, we can generalize the electromagnetic potential field (\ref{def:AKNanomalous}) supported on z$G$K by adding 
any almost everywhere (on z$G$K) harmonic electric or magnetic potential field solving the Maxwell equations on z$G$K, 
see \cite{Evans51}.

  In particular, the KN electromagnetic field can readily be generalized to exhibit a \emph{KN-anomalous magnetic moment},
\beq\label{def:AKNanomalous}
\bAKNanom = - \frac{r}{r^2+a^2\cos^2\theta} (\textsc{q}dt - \textsc{i}\pi a^2 \sin^2\theta\, d\varphi),
\eeq
with which one can decorate the z$G$K spacetime of the same ring radius $|a|$; here,
$\textsc{i}$ is the electrical current which produces a magnetic dipole moment $\textsc{i}\pi a^2$ when viewed from
spacelike infinity in the $r>0$ sheet.
 Our terminology for the case $\textsc{i}\pi a\neq \textsc{q}$ is in analogy to the physicists' ``anomalous magnetic moment 
of the electron;'' so, the \emph{KN-anomalous magnetic moment} is $\textsc{i}\pi a^2-\textsc{q}a$.

 Incidentally, notice that ${\bAKNanom}$ in \refeq{def:AKNanomalous} satisfies the Coulomb gauge.

 Furthermore, we can add the electric potential of a point charge source.
 The electrostatic field $\bE_{\textrm{pt}}$ generated by a positive point charge of magnitude $\textsc{q}'$ must be curl-free, 
and thus is a gradient:
$$
 \bE_{\textrm{pt}} = d \phi_{\textrm{pt}},
$$
where $\phi_{\textrm{pt}}$ solves Poisson's equation on $\cN$ with a point-source located at ${\bq}_{\textrm{pt}},$
$$
- \De_{\cN}^{} \phi_{\textrm{pt}} = 4 \pi \textsc{q}' \de_{{\bq}_{\textrm{pt}}}.
$$
 Now, $\cN$ is a two-sheeted Riemann space branched over the ring, and the fundamental solution of the Laplacian on such a manifold 
has been known for a long time \cite{Hob00,Neu51,DavRei71,EfiVor74}.
  It is best described in terms of {\em peripolar} \cite{Neu1880} (sometimes called toroidal) coordinates $(\zeta, \chi, \varphi)$.
  Their definition is as follows:  Let ${\bq}$ be a point in $\cN$, and set $\br = {\bq} - \bq_{\mathrm{rg}}$ (recall 
that $\bq_{\mathrm{rg}}$ is the center of the electron ring).
  Consider the plane spanned by $\br$ and $\bn_{\mathrm{rg}}$, the normal to the disc spanned by the ring.
  It intersects the ring at two antipodal points $\bq_1$ and $\bq_2$, the smaller index always reserved for the point 
closer to ${\bq}$.
 Let $d_1$ and $d_2$ denote the distances of ${\bq}$ from $\bq_1$ and $\bq_2$ respectively.
 Then the peripolar coordinate $\zeta := \ln(d_2/d_1) \geq 0$, and the coordinate $\chi$ is simply the angle between 
vectors ${\bq}-\bq_2$ and ${\bq}-\bq_1$.
  Note that $\chi$ thus defined will be double-valued on ordinary space, since as ${\bq}$ is moved through the
 ring the value of $\chi$ will jump from $-\pi$, its value on the top side of the  disc $\cD$, to $\pi$, its value 
on the bottom side of the disc.
 The angle  $\chi$ is an example of a multi-valued harmonic function in $\RR^3$, first studied by Sommerfeld \cite{Som97},
who is credited with introducing the concept of a {\em branched Riemann space}, i.e. a three-dimensional analog of a Riemann 
surface, on which multi-valued harmonic functions such as the peripolar $\chi$ become single-valued.
  Other examples of such {\em branched potentials} (the term used by Sommerfeld) include the oblate spheroidal coordinate 
functions $\xi$ and $\eta$.
  Just as in the case of oblate spheroidal coordinates, the system of  coordinates $(t,\zeta,\chi,\varphi)$, with $\varphi$ being 
the same azimuthal angle introduced before, form a single chart that covers the maximal extension of z$G$KN,  with 
$-\pi<\chi<\pi$ in one sheet and $\pi<\chi<3\pi$ in the other sheet of that space.
  In terms of peripolar coordinates, the electrostatic potential due to the point charge of magnitude $\textsc{q}'$ is  \cite{Hob00,Neu51,DavRei71,EfiVor74}
\beq\label{phipt}
\phi_{\textrm{pt}} = \frac{\textsc{q}'}{R} \left( \half + \frac{1}{\pi}
\sin^{-1}\frac{\cos\frac{ \chi- \chi_{\textrm{pt}}}{2}}{\cosh\frac{\vartheta}{2}} \right),
\eeq
where 
\begin{eqnarray}
R &:=& |\bq - \bq_{\textrm{pt}}| = \frac{a\sqrt{2}\sqrt{\cosh\vartheta-\cos( \chi- \chi_{\textrm{pt}})}}{\sqrt{\cosh\zeta-\cos \chi}\sqrt{\cosh\zeta_{\textrm{pt}} -\cos \chi_{\textrm{pt}}}},\\
\cosh\vartheta & := & \cosh\zeta\cosh\zeta_{\textrm{pt}} -\sinh\zeta\sinh\zeta_{\textrm{pt}}\cos(\varphi-\varphi_{\textrm{pt}}),
\end{eqnarray}
and where $\bq_{\textrm{pt}} = (\zeta_{\textrm{pt}},\chi_{\textrm{pt}},\varphi_{\textrm{pt}})$ is the position of the point charge in peripolar coordinates.
 Amending this electric potential field to $\bAKNanom$ yields the \emph{hydrogenic electromagnetic potential} on a z$G$K spacetime, thus
\beq\label{def:AKNhydrogen}
\bAKNhydro = \bAKNanom - \phi_{\textrm{pt}} dt;
\eeq
we will use it later to study the Born--Oppenheimer Hydrogen problem with our Dirac equation for a
z$G$KN singularity (with / without anomaly) of charge $\textsc{q}=-e$ interacting with a positive point charge of magnitude
$\textsc{q}'=e$ supported elsewhere on the z$G$K spacetime; here, $e$ is the empirical \emph{elementary charge} used by
physicists.

\section{Dirac's equation for a zero-$G$ Kerr--Newman type singularity}

       In this section we associate the single-particle Dirac wave function in a consistent manner first
with the neutral ring singularity of the topologically non-trivial z$G$K spacetime consisting of two cross-linked copies of 
Minkowski spacetime, and subsequently  --- in a compelling manner --- with the charge- and current-carrying ring singularity 
of the topologically non-trivial z$G$KN spacetime.

       We begin with some group theoretical preliminaries, discussing the spinorial representation of the Lorentz group as well as
the sheet swap map associated with the non-trivial topology of the z$G$K and a$G$KN spacetimes.
       This is based on what one does when a Dirac equation is to be formulated  for a point particle in the z$G$K or z$G$KN spacetimes.
       
       Subsequently we invoke \emph{the principle of relativity} to show that this formalism also covers the 
one-body Dirac equation for a ``free'' ring singularity, i.e. one not interacting with any other electromagnetic 
object in the manifold (the fact that the z$G$KN ring singularity carries charge and current does not yet enter the formalism); 
here we also benefit from the pioneering works of Schiller \cite{Sch62} and others (see \cite{Hol88}, and refs. therein), who 
first investigated whether non-pointlike, axisymmetric structures can be represented by Dirac bi-spinors.

        Then we generalize to the one-body Dirac equation for a zero-$G$ Kerr--Newman singularity which interacts electromagnetically with
an additional electromagnetic field that can be supported by the z$G$K manifold.
  In particular, this extra field may be generated by a point particle, which will be assumed to have such a 
large mass that it can be treated in Born-Oppenheimer approximation as infinitely massive and thereby as ``classical;'' 
more precisely, its location $\bq_{\textrm{pt}}$ (and possibly its magnetic moment $\boldsymbol{\mu}_{\textrm{pt}}$)
enter the equations as parameters, not as operators.
  This formulation generalizes readily to the situation where an anomalous magnetic moment is added to the KN magnetic moment.
  Since the ring singularity is not a point, its interaction with the point charge (or any other electromagnetic object, for that matter)
is \emph{not} given by a minimal coupling formula that multiplies the ring's charge with the Coulomb potential of the point charge 
evaluated at the ring's center,  but by a ``minimal re-coupling'' formula; our
interaction formula is a natural relativistic extension of the usual formalism employed to calculate the many-charges Coulomb
interaction from the classical field-energy integral as carried out, e.g., in \cite{JacksonBOOK}.
  We shall explicitly compute the electromagnetic interaction of the zero-$G$ Kerr--Newman singularity --- in fact, its generalization to
generate the fields (\ref{def:AKNanomalous}) --- with a point charge.

  We then show that in the limit of vanishing ring radius $|a|$ the Dirac point spectrum reproduces a positive plus a negative
Sommerfeld fine structure spectrum (with the correct labeling of the levels) --- we also explain, why in the traditional
special-relativistic calculations one only obtains half of it.

  Finally, we point out problems with perturbation theory as a tool for computing corrections to the Sommerfeld fine structure
formula in a ``small $a$'' regime, and we also comment on the perturbative approach to compute corrections to the 
z$G$KN-Dirac spectrum at finite-$a$ coming from a KN-anomalous magnetic moment.

\subsection{Group theoretical considerations}

\subsubsection{{Spinorial Representations of the Lorentz group, and topology of z$G$KN}}
Except  for the last paragraph, the material in this section is classical and can be found, e.g., in \cite{ThallerBOOK}, pp.68-77.

Let $H(2)$ denote the Hermitian matrices in $\Cset^{2\times 2}$.
  It is a real vector space of dimension four, and a basis is $\{\sigma_\mu\}_{\mu=0}^3$ where $\si_0 = I_2$ and $\si_i$ are the 
Pauli matrices:
\beq\label{eq:PAULIsigma}
\si_1 = \left(\begin{array}{cc} 0 & 1\\ 1 & 0\end{array}\right),\quad
 \si_2 = \left(\begin{array}{cc} 0 & -i\\ i & \;\, 0\end{array}\right),\quad
 \si_3 = \left(\begin{array}{cc} 1 & \;\, 0\\ 0 & -1\end{array}\right).
\eeq
  Let $\Rset^{1,3}$ denote the Minkowski spacetime, with metric represented by $(\boldsymbol{\eta}) = \diag(1,-1,-1,-1)$.
  Let the two mappings  $\si,\si':\Rset^{1,3}\to H(2)$ be defined by
 $$
\si(X) = X^\mu\si_\mu,\qquad \si'(X) = X^0 \si_0 - \sum_{i=1}^3 X^i\si_i.
$$  
Each of these mappings gives rise to a representation of the proper Lorentz group $SO_0(1,3)$ 
(the connected component of the identity in $O(1,3)$) by matrices in $SL(2,\Cset)$, in the following way:
  If $A \in SL(2,\Cset)$, let ${A}^\dagger$ denote its Hermitian adjoint.
  Now let $Y \in \Rset^{1,3}$ be such that 
$$
\si(Y) = A \si(X) A^\dagger,
$$ 
 Then $Y = L_A X$, where $L_A$ is a member of the proper Lorentz group.
 Note that this shows $SL(2,\Cset)$ to be the double cover of  the proper Lorentz group, 
because both $A=I_2$ and $A=-I_2$ give $L_A=I_4$.

 The maps $\si$ and $\si'$ are chosen in such a way that the two representations they give are inequivalent.
  This is because there is no matrix $S\in SL(2,\Cset)$ such that $S\si(X)S^{-1} = \si'(X)$.
  Note that $\si'(X) = \si (PX)$ where $P=\left(\begin{array}{cc} 1 & 0\\ 0 & -I_3\end{array}\right)$ is an element 
of the Lorentz group responsible for {\em space reflection}.
  There is no element in $SL(2,\Cset)$ that can represent $P$, because $\det P = -1$.
  A similar statement  is true about the {\em time reversal} matrix $T = -P$.

Let $\ga:\Rset^{1,3} \to \Cset^{4\times 4}$ be defined as 
$\ga(X) = \left(\begin{array}{cc} 0 & \si(X)\\ \si'(X) & 0\end{array} \right)$.
  For $A \in SL(2,\Cset)$ let $\La_A := \left(\begin{array}{cc} A & 0 \\ 0 &(A^\dagger)^{-1}\end{array} \right).$ 
Then one checks that $\La_A \ga(X) \La_A^{-1} = \ga(L_A X)$ where $L_A$ is as before.
  Thus the mapping $\ga$ gives another representation of the proper Lorentz group by the special linear group $SL(2,\Cset)$.
  Let also $\La_P := \left(\begin{array}{cc} 0 & I_2 \\ I_2 &0 \end{array}\right)$ and 
$\La_T := \left(\begin{array}{cc} 0 & -iI_2 \\ iI_2 &0\end{array} \right)$.
  It is easy to see that $\La_M \ga(X) \La_M^{-1} = \ga(MX)$ holds for $M = P$, $M=T$, and $M=PT$.
 Thus $\ga$ gives a representation of the full Lorentz group, in the sense that
$$
 O(1,3) = \{ \La_A,\La_P\La_A,\La_T\La_A,\La_{PT}\La_A \ | \ A \in SL(2,\Cset)\}.
$$
Setting $\gamma(X) = \gamma_\mu X^\mu$ defines the Dirac matrices $\{\ga_\mu\}_{\mu=0}^3$.
 We have  
$$
\ga_0 = \ga^0= \Big(\begin{array}{ll} 0 & I_2 \\ I_2 & 0 \end{array}\Big),\qquad \ga_i =-\ga^i 
=  \Big(\begin{array}{ll} \quad 0 & \sigma_i \\ -\sigma_i & 0 \end{array}\Big),\quad i=1,2,3.
$$
Also define $\al_W^k := \ga^0 \ga^k =  \Big(\begin{array}{ll} \si_k & \quad 0 \\  0 & -\si_k\end{array}\Big)$ and $\beta_W^{} := \ga^0$.
 Thus $\{\beta_W^{},\al_W^k\}$ provides another basis for the same representation of the Clifford algebra by the $\gamma$ matrices; they
can be re-expressed as 
$$
\ga^0 = \beta_W^{}, \ga^k = \beta_W^{}\al_W^k.
$$ 
  This is called the {\em spinorial}, or {\em Weyl}, representation.

  A basis for another representation of the Clifford algebra, unitarily equivalent to the above one, is given by the matrices 
$\beta_{DP}^{} := \left(\begin{array}{ll} I_2 & \;\; 0 \\ 0 & -I_2\end{array}\right)$ and 
$\al_{DP}^k :=  \left(\begin{array}{ll} 0 & \sigma_k \\ \sigma_k & 0\end{array}\right)$.
  This is called the {\em standard}, or {\em Dirac--Pauli}, representation.
  Note that one still has $\ga^k = \beta_{DP}^{}\al_{DP}^k$; however, $\ga^0$ in the Dirac--Pauli
representation is $\beta_{DP}^{}$, which is not the same as $\ga^0$ in the Weyl representation, which is the same as $\beta_W^{}$.

  The above calculation was done on the Minkowski spacetime $\Rset^{1,3}$.
  Let $\cM$ be {\em any} Lorentzian manifold.
  Then the tangent space and the cotangent space at each point on the manifold are copies of $\Rset^{1,3}$.
  Thus all of the above can be replicated on the tangent and cotangent bundles of the manifold.

  In particular, let $\cM$ be the z$G$KN or z$G$K spacetime, and let $\{E_\mu\}_{\mu=0}^3$ denote 
an orthonormal basis (with respect to $\bg$) for $T_p\cM$.
  A vector $\Xi \in T_p\cM$, $\Xi = \Xi^\mu \frac{\p}{\p x^\mu}$ has an expansion in 
this basis $\Xi = X^\mu E_\mu$ and identifying $(X^\mu)$ with a point in the Minkowski spacetime, one has $\si(X) =  X^\mu \si_\mu$.
 In this way both the tangent and the cotangent bundle of $\cM$ have a representation 
as {\em order two} spinors (i.e. objects with two spinor indices, or in other words, 
operators that act on {\em  order one} spinors, to be defined below).

 In addition, the \emph{sheet swap map} $\varsigma:\cM \to \cM$ acts by 
$\varsigma(t,\xi,\eta,\varphi) = (t,-\xi,-\eta,\varphi)$; it is a {\em bundle map}, 
i.e. $\Pi\circ \varsigma = \Pi$.
 It is an isometric {\em involution} on $\cM$: its pullback acts on the metric like
$\varsigma^* g = g$, while $\varsigma^2 = id$, and it fixes the ring: $\left.\varsigma\right|_{\cR_t} = id$.
 Its differential $d\varsigma(p) : T_p\cM \to T_{\varsigma(p)}\cM$ induces an equivalent representation because:
$$ 
\si(d\varsigma(X)) = X^0\si_0 - X^1\si_1 - X^2 \si_2 + X^3\si_3 = \si_3 \si(X) \si_3.
$$
  
\subsubsection{Bi-Spinors and Sheet Swaps}

Again, up to (\ref{spintrans}) and accompanying text, 
the material in this section is classical and can be found, e.g., in  \cite{CartanBOOK}.

In particular, the following basic definitions are due to Cartan \cite{CartanBOOK}.
\begin{defn} 
A vector $w$ in a vector space $V$ (over $\Cset$) on which a non-degenerate bi-linear form $\langle~,~\rangle$ is defined, 
is called {\em isotropic} with respect to that bi-linear form if 
$$
\langle w,w\rangle = 0.
$$
(Note that the form is assumed to be bi-linear, not {\em sesqui}-linear.) 
\end{defn}
 For example, the Minkowski metric $\boldsymbol{\eta}$ is a non-degenerate bi-linear form on $\Cset^4$.
  A vector $w\in \Cset^4$ is thus isotropic with respect to $\boldsymbol{\eta}$ if $w_0^2 - \sum_{i=1}^3 w_i^2 = 0$, 
i.e. if it is a (complexified) null vector.
  Let $w\ne 0$ be such a vector.
 It is easy to see that $W := \gamma(w)$ will be singular, i.e. $\det W = 0$.

\begin{defn} 
A vector $\Psi \in \Cset^4$ is called a {\em bi-spinor} if there exists a non-zero vector $w \in \Cset^4$ isotropic with 
respect to the Minkowski metric $\boldsymbol{\eta}$ such that 
$$ 
W\Psi = 0.
$$
\end{defn}
 The definition of  a bi-spinor makes it clear that it is defined {\em projectively}, i.e. $\Psi$ is equivalent to  $\lambda \Psi$ for 
$\la \in \Cset\setminus\{0\}$.

  Recall that for $M \in O(1,3)$ we have shown that
\beq\label{vectrans}\ga(Mw)=\Lambda_M \gamma(w)\Lambda_M^{-1},
\eeq
where $\Lambda_M$ is the matrix corresponding to $M$ in the representation of the Lorentz group given by $\ga$-matrices described above.
  On the other hand, since $M$ preserves the Minkowski bi-linear form $\boldsymbol{\eta}(Mx,My) = \boldsymbol{\eta}(x,y)$, it thus follows that 
$w$ is isotropic with respect to $\boldsymbol{\eta}$ iff $Mw$ is.
  Moreover it is easy to see that if the bi-spinor $\Psi'$ generates the nullspace of $Mw$, then we must have
\beq
\label{spintrans}
\Psi' = \Lambda_M \Psi.
\eeq 
This is the rule of transformation of (rank-one) bi-spinors.
  Comparing \refeq{spintrans} with \refeq{vectrans} we note that, unlike the spinorial representation $\ga(w)$ of a vector $w$, a bi-spinor 
is transformed by the left action alone, not the conjugate action, of the group.
  In particular, let $M$ correspond to a space rotation of any angle about any given axis.
 Then, because of \refeq{vectrans}, $\Lambda_M$ would have to correspond to a rotation of half of that angle, 
and thus by \refeq{spintrans} a bi-spinor would be rotated through {\em half} of that angle.
  A rotation through the angle $2\pi$, which leaves all vectors $w\in V$ invariant, takes $\Psi$ to $-\Psi$ instead.

  Finally, we consider the action of $\varsigma$ on a bi-spinor. 
  As explained above, it corresponds to a sheet swap, which can be alternatively described as replacing each point on one sheet with
an associated point on the other sheet reached by looping through the ring once (a $2\pi$ circle). 
  Since $\varsigma^2 = id$, two full loopings through the ring brings one back to ``square one,'' which is analogous to a spin-1/2 
(bi-)spinor rotation of angle 4$\pi$ corresponding to a full 2$\pi$ rotation in Euclidean space.
  Thus we appropriately may speak of the particle associated to our bi-spinor as having ``topo-spin'' 1/2.

\begin{rem}\emph{
The ``looping through the ring'' visualization of topo-spin should not be confused with some kind of rotation in $\cN$, the constant-$t$
snapshot of z$\,G$KN;  it is independent of the notion of a spinorial representation of the rotation group.
 In particular, ``scalar'' particles can have topo-spin, too:
 a scalar wave function $\Psi\in \Lsp^2(\cN,\Cset)$ can be viewed as depending on $\bq\in\cN$, equivalently on a vector position 
$\br\in\Rset^3$ plus a discrete variable $\varkappa\in\{-1,1\}$ which indicates on which copy of $\Rset^3$ the vector $\br$ lives. 
 The similarity with Pauli's original way of writing spin variables as arguments rather than components of $\Psi$ is evident. 
 Still, the action of the topo-spin operator on such a $\Psi$ (see next subsection) represents a sheet swap, not a rotation.
}
\end{rem}

\subsubsection{Generalized Cayley--Klein representation of a bi-spinor}

 In this subsection we describe a  representation of Dirac bi-spinors that is a generalization of the Cayley--Klein representation of 
Pauli spinors.
 Therefore we first consider the two-component Pauli spinors.  

 Let $\psi :\RR\times \RR^3 \to \CC^2$ be a (two-component) Pauli spinor field.
 Set 
$$
R^2 := \psi^\dag \psi,
$$ 
where for any column vector $\psi\in \Cset^k$ we have defined $\psi^\dag = \psi^{*t}$ to be the complex conjugate-transpose of $\psi$.
 Then the unit spinor $\check{\psi}$ has the following $SU(2)$ (Cayley--Klein) representation \cite{Gol80}:
\beq\label{CK}
\check{\psi}:=\frac{1}{R} \psi = \left(\begin{array}{c} \cos(\Theta/2)
e^{i(\Phi-\Omega)/2}\\ \sin(\Theta/2)
e^{i(\Phi+\Om)/2}\end{array}\right);
\eeq
here, for each point in the configuration space, $(\Phi, \Theta,\Omega)$ are a triplet of Eulerian angles\footnote{Here 
    and elsewhere in the paper we are using the ZYZ convention for 
    Euler angles, whereby any rotation in $\RR^3$ can be decomposed into a rotation around the $z$-axis, followed by one around the (new) 
    $y$-axis, followed by another one around the (new) $z$-axis} 
corresponding to a rotation $\cR(\Phi,\Theta,\Omega) \in SO(3)$; clearly,
$(\Phi,\Theta,\Om)$ are real-valued functions on the one-body configuration space.
 The above representation is obtained as follows:
Given $\psi\in \Cset^2$, there is a unitary matrix $U^\psi \in SU(2)$ such that 
$$
U^\psi \left(\begin{array}{c} 1 \\ 0 \end{array}\right)  = \check{\psi}.
$$ 
   The map
\beq\label{def:CKmap}
(\Phi,\Theta,\Omega) \to U 
:= e^{-i\frac{\Om}{2}\si_3}e^{-i\frac{\Theta}{2}\si_2}e^{-i\frac{\Phi}{2}\si_3} 
=\left(\begin{array}{cc}
\cos(\Theta/2) e^{i(\Phi-\Omega)/2} & -\sin(\Theta/2) e^{-i(\Phi+\Omega)/2}\\ 
\sin(\Theta/2) e^{i(\Phi+\Om)/2}    & \quad \cos(\Theta/2) e^{-i(\Phi-\Omega)/2}
\end{array}\right)
\eeq 
is a map from $SO(3)$ into its universal cover $SU(2)$ that takes any such triplet of Euler angles $(\Phi,\Theta,\Omega)$ to 
(one of the two) $SU(2)$ elements that comprise the inverse image of $\cR(\Phi,\Theta,\Omega)$ under the covering map.

\begin{rem}\emph{
 Incidentally, in  \cite{BST55} this representation of  Pauli spinors is used to give an ontological fluid interpretation of
Pauli's equation in the spirit of Madelung's fluid interpretation of Schr\"odinger's equation; subsequently 
\cite{BohmHileyBOOK} it was noted that it supplies a non-relativistic law for the orientation of a rigid, spinning (spherical) 
model electron.
 We also direct the reader to \cite{Dankel70} for a related discussion of spinning spherical particles in the context of Nelson's 
stochastic mechanics.}
\end{rem}

 Given any ${\psi} = \left(\begin{array}{c} \fz_1\\ \fz_2\end{array}\right) \in \Cset^2$, 
one can find a triplet of Eulerian angles $(\Phi,\Theta,\Omega)$, (unique up to the 
usual ambiguity in Eulerian angles), such that \refeq{CK} holds.
  More precisely, we show below that every non-zero $\psi$ determines an orthonormal frame in $\RR^3$, 
denoted by $\{\bl(\psi),\bm(\psi),\bn(\psi)\}$ (cf. \cite{Hol88}) and thus a unique element of the real rotation group $SO(3)$ 
that takes the standard basis for $\RR^3$ to the basis $\{\bl,\bm,\bn\}$.  
 Set
$$
\bn(\psi) := \frac{\psi^\dag \siV \psi}{\psi^\dag \psi} = 
\frac{1}{|\fz_1|^2+|\fz_2|^2} \left(\begin{array}{c} 2 \mbox{Re} (\fz_1{\fz_2^*})\\
2 \mbox{Im}(\fz_1{\fz_2^*})\\ 
|\fz_1|^2 - |\fz_2|^2\end{array}\right);
$$
here, the ${}^*$ denotes complex conjugation.
 It is easy to see that $\bn$ is a unit vector:
 using the Cayley--Klein form of $\psi$ \refeq{CK} we obtain
$$
\bn(\psi) = \left(\begin{array}{c} \sin\Theta\cos\Om\\
\sin\Theta\sin\Omega\\
\cos\Theta\end{array}\right).
$$

 Next, let us define the {\em flip map} $\ff :\Cset^2 \to \Cset^2$ as follows:
$$
\ff \left(\begin{array}{c} \fz_1\\ \fz_2 \end{array}\right) 
=  \left(\begin{array}{c} -{\fz_2^*}\\ \ {\fz_1^*} \end{array}\right).
$$
  We note that
$$
\bn(\ff \psi) = - \bn(\psi),
$$
which is why $\ff$ is called the flip map.
  Next we define vectors $\bl,\bm \in \RR^3$ by
$$
\bl + i \bm := \frac{ (\ff \psi)^\dag \siV\psi}{\psi^\dag \psi} =
\left(\begin{array}{c} \fz_1^2 - \fz_2^2\\ i(\fz_1^2 + \fz_2^2)\\ -2\fz_1\fz_2\end{array}\right),
$$
and in terms of the Cayley--Klein parameters:
$$
\bl(\psi) = \left(\begin{array}{c} \cos\theta\cos\Om\cos\Phi+\sin\Om\sin\Phi\\
\cos\theta\sin\Om\cos\Phi - \cos\Om\sin\Phi\\
-\sin\theta\cos\Phi\end{array}\right),\qquad
\bm(\psi) = \left(\begin{array}{c}
\cos\Theta\cos\Om\sin\Phi-\sin\Om\cos\Phi\\
\cos\Theta\sin\Om\sin\Phi +\cos\Om\cos\Phi\\
-\sin\Theta\sin\Phi\end{array}\right).
$$
 One checks that $\{\bl,\bm,\bn\}$ are indeed orthonormal.
  Putting the three vectors together gives the rotation matrix $\cR(\Phi,\Theta,\Om)$ mentioned above, which takes the standard frame 
$\{\be_x,\be_y,\be_z\}$ of $\RR^3$ into the frame $\{\bl,\bm,\bn\}$.  

 Now let $\Psi \in\Cset^4$ be a bi-spinor. 
 Thus 
$$
\Psi = \left(\begin{array}{c} \psi_1 \\ \psi_2 \end{array}\right),
$$
with $\psi_1,\psi_2 \in \Cset^2$.
  Using the Cayley--Klein representations of $\psi_1$ and $\psi_2$, we may write
$$
\Psi = \left(\begin{array}{c} 
R_1 e^{i \Phi_1/2} \left(\begin{array}{c} \cos(\Theta_1/2) e^{-i\Om_1/2} \\
\sin(\Theta_1/2) e^{i\Om_1/2}\end{array}\right) \\
R_2 e^{i \Phi_2/2} \left(\begin{array}{c} \cos(\Theta_2/2) e^{-i\Om_2/2} \\
\sin(\Theta_2/2) e^{i\Om_2/2}\end{array}\right)
\end{array}\right).
$$
 Let $R:= \sqrt{R_1^2+R_2^2}$,\quad $\Si := 2\tan^{-1}\frac{R_2}{R_1}$,\quad $S := \half(\Phi_2+\Phi_1)$, and $\Phi = \half(\Phi_2-\Phi_1)$.
  Then,
\beq\label{canspin}
\Psi =   R e^{iS} \left(\begin{array}{c} \cos\frac{\Si}{2} e^{-i\Phi/2} \left(\begin{array}{c}   \cos(\Theta_1/2) e^{-i\Omega_1/2}\\ 
\sin(\Theta_1/2) e^{i\Om_1/2}   \end{array}\right)\\
                                            \sin\frac{\Si}{2}e^{i\Phi/2}  \left(\begin{array}{c}   \cos(\Theta_2/2) e^{-i\Omega_2/2} \\   
 \sin(\Theta_2/2) e^{i\Omega_2/2}      \end{array}\right)     \end{array}\right).
\eeq
 We call this the \emph{generalized Cayley--Klein representation} of the Dirac bi-spinor $\Psi$.  

 We now define the {\em orientation vector field} of $\Psi$ by
\beq\label{eq:classSPIN}
\bn_\Psi := \frac{\Psi^\dag \mathbf{S} \Psi}{\Psi^\dag \Psi} = \cos^2\frac{\Si}{2} \bn_1 + \sin^2\frac{{\Si}}{2} \bn_2.
\eeq
 Here, $\mathbf{S} = (S_1,S_2,S_3)^t$, with
$$
S_k:= \left(\begin{array}{cc} \si_k & 0 \\ 0 & \si_k \end{array}\right),\qquad k = 1,2,3,
$$
while $\bn_1 := \bn(\psi_1)$ and $\bn_2 := \bn(\psi_2)$.
 One readily checks that
$$
\|\bn_\Psi\|^2 = \half\left( 1 + \cos^2\Si + (\bn_1 \cdot \bn_2)\sin^2\Si \right) \leq 1,
$$
with equality holding if and only if either $\Si = 0$ or $\pi$, or  if the vectors $\bn_1$ and $\bn_2$ are {\em parallel},
i.e. $\bn_1 \cdot \bn_2 = 1$.  
 Of interest is also when this vector field vanishes: $\bn_\Psi =0$ iff $\psi_2 = e^{i\Phi} \ff\psi_1$.

 Finally, by analogy with the Pauli spinor, for a bi-spinor we can define
$$
\bl'_\Psi + i \bm'_\Psi := \frac{(\fF\Psi)^\dag \mathbf{S} \Psi}{\Psi^\dag \Psi},
$$
where 
$$
\fF \Psi := 
Re^{iS}\left(\begin{array}{c}
\cos\frac{\Si}{2} e^{-i\Phi/2}  \left(\begin{array}{c}   -\sin(\Theta_1/2) e^{i\Omega_1/2} \\   
\ \cos(\Theta_1/2) e^{-i\Omega_1/2}      \end{array}\right)    
\\
 \sin\frac{\Si}{2}e^{\;i\Phi/2} \left(\begin{array}{c}   -\sin(\Theta_2/2) e^{i\Omega_2/2}\\ 
\ \cos(\Theta_2/2) e^{-i\Om_2/2}   \end{array}\right)
\end{array}\right).
$$
 Thus, 
$$
\bl_\Psi'  = \cos^2\frac{\Si}{2}\ \bl_1 +\sin^2
\frac{\Si}{2}\ \bl_2,\qquad 
\bm_\Psi' = \cos^2\frac{\Si}{2}\ \bm_1 +\sin^2
\frac{\Si}{2}\ \bm_2.
$$
 Unfortunately, the triplet $\{\bl'_\Psi,\bm'_\Psi,\bn_\Psi\}$ forms an orthogonal frame only if $\bn_1\times\bn_2=\boldsymbol{0}$, 
in general. 
 But, when $\bn_1\times\bn_2\neq \boldsymbol{0}$, then we can use these vectors to define
$$
\bl_\Psi := \bn_1\times\bn_2,
$$
which is orthogonal to $\bn_\Psi$ because $\bn_\Psi$ is a linear combination of $\bn_1$ and $\bn_2$, and 
$$
\bm_\Psi := \bn_\Psi\times\bl_\Psi.
$$
 The triplet $\{\bl_\Psi,\bm_\Psi,\bn_\Psi\}$ forms an orthogonal frame if $\bn_1\times\bn_2\neq\boldsymbol{0}$.
 In each case one can obtain an orthonormal frame by normalization.

 After our group-theoretical preparations we are ready to formulate the Dirac equation for a z$G$KN ring singularity.
 We begin with the simpler case of a ``free'' z$G$KN ring singularity; the interacting case is treated thereafter.

\subsection{Dirac's equation for a free z$G$KN ring singularity}\label{sec:DiracEQzGKNrg}

 Recall that any quantum-mechanical wave equation for  a ``quantum particle,'' when formulated ``in position space representation,'' 
is formulated on the configuration space $\cC$ of the corresponding ``classical (point) particle.''
 For a ``free'' massive test particle whose classical worldlines (according to Einstein's general theory of relativity) are timelike 
geodesics in some static background manifold $(\cM,\bg)$ its configuration space is simply a spacelike constant-time slice $\cN$ of 
$\cM$, where ``time'' parametrizes the timelike Killing field of $(\cM,\bg)$; for a single quantum particle we may more casually write
that its wave equation is formulated on $(\cM,\bg)$.
 In particular, the Dirac equation for a ``free'' spin-$1/2$ particle of rest mass $m$ with classical worldline 
in a spacetime $(\cM,\bg)$, when formulated w.r.t. arbitrary coordinates $(x^\mu)_{\mu=0}^3$ of $\cM$ (with $c=1$) reads 
(cf.  \cite{BrillCohen66}, \cite{KieTah14a}):
\beq\label{eq:DirEqFREE}
{\tilde\ga}^\mu p_\mu \Psi + m \Psi = 0;
\eeq
here, $\Psi = \Psi(x^1,x^2,x^3,x^4)$ is a bi-spinor field on $(\cM,\bg)$, while 
$p_\mu = -i\hbar \nab_\mu$ with $\nab$ the covariant derivative (on spinors) associated to the spacetime metric $\bg$, and
$({\tilde\ga^\mu})_{\mu=0}^3$ are Dirac matrices associated to this metric, i.e. satisfying the anti-commutation relations
\beq
\tilde{\ga}^\mu \tilde{\ga}^\nu + \tilde{\ga}^\nu \tilde{\ga}^\mu = 2 g^{\mu\nu}{\boldsymbol{1}}_{4\times4}.
\eeq

 Now, what is a ``free'' z$G$KN ring singularity, and what is its configuration space?
 These may at first sound like perplexingly difficult questions, but by the principle of relativity 
the first one has a straightforward answer, and surprisingly also the second question has a simple answer based on
the principle of relativity and the group-theoretical considerations of the previous subsection.

 First, let's clarify what a ``free'' z$G$KN ring singularity is.
 To see this, start from the picture of a test particle moving freely in $\cN$, its worldline being a timelike geodesic in $(\cM,\bg)$, which
for $\cM=$z$G$KN (and z$G$K for that matter, too) is a straight line in Euclidean sense, possibly changing sheets by going through the 
(geometrical) disc spanned by the ring. 
 Any static frame attached to the static z$G$KN (and z$G$K) spacetimes is an inertial frame, but any worldpoint of the 
freely moving test particle is the origin of some other inertial frame, too; so by a simple Lorentz change of inertial frames 
one can speak of the z$G$K or z$G$KN singularity as moving freely w.r.t. an inertial frame in which the point test particle is at the 
origin.

 Next, we clarify what the configuration space is for such a free z$G$KN ring singularity.
 This is a more subtle issue, for relative to a Dreibein with origin at 
$\bq_{\mathrm{pt}} = (\br_{\mathrm{pt}},\varkappa)\in\Rset^3\times\{-1,1\}$ 
in a constant-time snapshot $\cN$ of z$G$KN its axisymmetric ring singularity of fixed radius $|a|$ has a geometrical center at
$\bq_{\mathrm{rg}} = (\br_{\mathrm{rg}},\varkappa)\in\Rset^3\times\{-1,1\}$ 
and a normal $\bn_{\mathrm{rg}}\in \Rset^3$ to the disc spanned by the ring; in addition, if one ``marks off'' a 
reference point on the ring, then also an azimuthal angle is needed to specify where that point is.
 Since this requires three spatial plus a two-valued discrete coordinate for $\bq_{\mathrm{rg}}$,
 and two angles for $\bn_{\mathrm{rg}}$, and possibly an azimuth, one could come to the conclusion 
that the configuration space for the ring singularity has to be five- or possibly six-dimensional.
 However, this counting tacitly assumes a scalar wave function description.
 Since we are working with the bi-spinorial wave functions of Dirac, the two angles for $\bn_{\mathrm{rg}}$ and the azimuth
are encoded in this structure already (see the previous subsection), so that the configuration space for the z$G$KN
 ring singularity is indeed merely the set of locations of its center $\bq_{\mathrm{rg}}$.
 Note that also for the standard Dirac point particle one usually speaks of ``its classical spin,'' which is simply formula
\refeq{eq:classSPIN} multiplied by $\hbar/2$, but this way of talking suggests more structure of a structureless point than there 
really is --- for a Dirac point electron, ``spin'' is purely in the bi-spinorial wave function (acted on by the $\alpha_k$ matrices);
by contrast, the same bi-spinorial information acquires ontological meaning for the structured object that the z$G$KN ring singularity is.

  Now we know that the configuration space $\cC$ for the z$G$K and z$G$KN ring singularity is the set of locations $\bq_{\mathrm{rg}}$ 
of its center, but which set is this, mathematically? 
  Since the set of all snapshots of where a freely moving test particle can be in $\cN$ relative to the ring singularity is $\cN$
itself, by relativity (turning the perspective around) the same is true for the center of the ring singularity of z$G$K and z$G$KN ---
so $\cC=\cN$ for a z$G$KN and z$G$K ring singularity.

 To obtain the domain for the Dirac wave equation of a free z$G$K or z$G$KN ring singularity we only have to stack up 
$\Rset$ copies of $\cN$ (indexed by time) and obtain a manifold which is isomorphic to z$G$KN (or z$G$K, without electromagnetism).
 Because of this isomorphism, the Dirac equation for the bi-spinorial wave function $\Psi$ of a free z$G$K or z$G$KN ring singularity 
``of mass $m$'' is identical to the Dirac equation \refeq{eq:DirEqFREE} for a free spin-1/2 particle of mass $m$ on a fixed z$G$K 
(or z$G$KN) background:  only the narrative of the variables in the argument of $\Psi$ changes.
 
 The last remark requires some elaboration, though. 
 In our discussion above we have resorted to a notation involving ``Euclidean'' vectors (plus the discrete variable $\varkappa$), 
which are defined relative to some Dreibein attached to an ``origin.''
 But clearly, the location of the origin and the orientation of the Dreibein attached to it are merely auxiliary constructs
which allow one to invoke vector-algebra and -calculus.
 As in Newtonian point mechanics of Newton's universe, the physics does not depend on the choice of origin nor on the Dreibein.
 In the same vein, the coordinates $\bq_{\mathrm{pt}} = (\br_{\mathrm{pt}},\varkappa)\in\Rset^3\times\{-1,1\}$ enter the Dirac 
equation \refeq{eq:DirEqFREE} only \emph{relative to the center and axis of the z$\,G$KN ring singularity, and relative to which sheet 
it is associated with}.
 In fact, by axisymmetry only two coordinates enter the bi-spinor field $\Psi(t,\,.\,)$ evaluated at 
$\bq_{\mathrm{pt}}$, namely $|\bq_{\mathrm{pt}}-\bq_{\mathrm{rg}}|$ and $(\bq_{\mathrm{pt}}-\bq_{\mathrm{rg}})\cdot\bn_{\mathrm{rg}}$, which
can be expressed in oblate spheroidal coordinates based on the ring, without $\varphi$, (see Appendix A).
 Indeed, since Chandrasekhar \cite{ChandraDIRACinKERRgeom,ChandraDIRACinKERRgeomERR} showed that the Dirac equation \refeq{eq:DirEqFREE} 
separates in the oblate spheroidal coordinates (BL coordinates) which provide a single chart for its maximal-analytical extension,
this has become the coordinate system of choice for essentially all ensuing studies of \refeq{eq:DirEqFREE}.
 Clearly, the triplet $(r,\theta,\varphi)$ of oblate spheroidal coordinates of a point in $\cN$ are neither its Cartesian coordinates,
nor do they refer to a Dreibein attached to the geometrical center of the ring singularity or attached to the point with coordinates $(r,\theta,\varphi)$.
 Yet they contain all the relevant information.
 More to the point, if one knows where a point particle is located in $\cN$ relative to the ring singularity, given by the triplet
$(r,\theta,\varphi)$, then one knows where a point of the ring singularity is located relative to the point particle (the $\varphi$
variable in each case relative to a specifically marked ``reference azimuth $\varphi_0$'' on the ring singularity); and by axisymmetry,
the $\Sset^1$ orbit of the $\varphi$ variable yields the location of the whole ring (both ``locations'' only modulo a rotation of $SO(3)$) ---
this is equivalent to giving $(r,\theta)$.
 And if one knows the location of the ring singularity relative to the point particle, then one can retrieve its center 
relative to the point particle. 
 Relative to some Dreibein attached to the point particle that center can of course be anywhere on an $SO(3)$ orbit.
 The $SO(3)$ ambiguity is now fixed by the orientational information extracted from the bi-spinor.

 Lastly, we note that instead of merely \emph{re-interpreting} 
the oblate spheroidal coordinates $(r,\theta,\varphi)$ of a point particle 
in $\cN$ as coordinates of a point on the ring singularity relative to that point particle, as just explained, one can also 
change variables to other oblate spheroidal coordinates for the center of the ring singularity; see Appendix~A. 

\subsection{Dirac's equation for a z$G$KN ring singularity (with / without anomaly) 
interacting with a point charge located elsewhere in the spacetime}

 To set up the Dirac equation for a spin-half z$G$KN ring singularity with ``inert mass $m$'' and charge $\textsc{q}$ (as seen asymptotically 
in the $r>0$ sheet) in the electric field of a static point charge $\textsc{q}'$ located at $\bq_{\textrm{pt}}$ in 
$\cN\subset$z$G$KN, we follow the strategy of 
the previous subsection and begin by recalling the Dirac equation for a spin-half test particle of charge $\textsc{q}'$ and mass $m$ in the 
z$G$KN background spacetime of charge $\textsc{q}$ as seen from infinity in the $r>0$ sheet; it was studied in \cite{KieTah14a}, and reads
(we allow an anomalous magnetic moment, and set $c=1$) 
\beq\label{eq:DirEqA}
{\tilde\ga}^\mu 
\left(-i\hbar \nab_\mu  - \textsc{q}' {\AKNanom}_\mu\right)\Psi + m \Psi = 0, 
\eeq
where the Dirac matrices $({\tilde\ga^\mu})_{\mu=0}^3$ satisfy the same anti-commutation relations as for the ``free'' Dirac particle,
and where $\Psi(t,\,.\,)$ and ${\AKNanom}_\mu(\,.\,)$ are evaluated at the location  $\bq_{\mathrm{pt}}\in\cN$ of the point charge,
relative to the location of the ring singularity. 
 By following Chandrasekhar's work \cite{ChandraDIRACinKERRgeom,ChandraDIRACinKERRgeomERR} on 
the Dirac equation for the ``free'' Dirac particle \refeq{eq:DirEqFREE}, Page \cite{Page76} and Toop \cite{Too76} showed 
that \refeq{eq:DirEqA}, with $\AKN$ in place of $\AKNanom$, is separable when the position $\bq_{\mathrm{pt}}\in\cN$ of the point charge 
relative to the ring singularity is coordinatized by the oblate spheroidal coordinates $(r,\theta,\varphi)$.

 Now, the conclusions of the previous subsection about the ``free'' spin-half z$G$KN singularity should extend to the spin-half z$G$KN 
singularity in the field of a point charge: it should be possible to re-interpret (\ref{eq:DirEqA}) as 
the Dirac equation for a spin-half z$G$KN ring singularity with inert mass $m$ and charge $\textsc{q}$ (as seen asymptotically 
in the $r>0$ sheet) and possibly a KN-anomalous magnetic moment, which interacts with
the electric field of a static point charge $\textsc{q}'$ located elsewhere in $\cN\subset$z$G$KN.
 Indeed, all that would seem to be necessary once again is to re-interpret the  oblate 
spheroidal coordinates $(r,\theta,\varphi)$ of the point charge relative to the ring singularity as coordinates of a point on
the ring singularity relative to the point charge.
 As an important spin-off, the Dirac equation of a spin-half zero-$G$ Kerr--Newman singularity in the field of a point charge
would become essentially identical to \refeq{eq:DirEqA}, the Dirac equation of a spin-half point charge in the field of a 
zero-$G$ Kerr--Newman singularity, which we have studied in  \cite{KieTah14a}. 
 
 However, all this can only be strictly valid for stationary situations; fortunately, this covers the totality of the bound states 
associated with the discrete spectrum.
 Our Dirac equation should still be approximately valid in the quasi-static regime, but not beyond.
 More on that in section 4.

 The reader may object that whereas the minimal coupling interaction term in (\ref{eq:DirEqA}) is 
easily understandable as describing the effect of $\bAKNanom$ on a test point charge, it definitely needs to be justified 
as describing also the effect of the electric field of a static point charge on a ``test z$G$KN singularity'' (possibly 
with KN-anomalous moment) --- even for stationary situations.

\begin{rem}\emph{
 In the interpretation of ``a z$\,G$KN ring singularity in a given electrostatic field of a point charge''
we appropriately should call the term 
$\textsc{q}' {\AKNanom}_\mu(\bq_{\mathrm{pt}})\Psi(\bq_{\mathrm{pt}})$ a ``\emph{minimal re-coupling} term'' rather than
a ``minimal coupling term'' ---
 this also pays attention to the fact that in the interpretation of $\Psi$ as bi-spinor wave function of the z$\,G$KN ring 
singularity the coordinates of $\bq_{\mathrm{pt}}$ relative to the ring singularity are used to locate the ring singularity
relative to the point, without requiring a notational change.}
\end{rem}

 In the next subsubsection we shall vindicate the minimal re-coupling term in (\ref{eq:DirEqA}) as the correct interaction
term for a spin-half z$G$KN ring singularity of  charge $\textsc{q}$ (as seen asymptotically in the $r>0$ sheet), possibly
appended by a Kerr--Newman-anomalous magnetic moment, in the electric field of a static point charge $\textsc{q}'$ located 
elsewhere in z$G$KN.

\subsubsection{Vindication of minimal re-coupling}

 To justify the above form of the Dirac equation for a spin-half z$G$KN ring singularity with inert 
mass $m$ and charge $\textsc{q}$ (as seen asymptotically in the $r>0$ sheet) in the electric field of a static point 
charge $\textsc{q}'$ located elsewhere in z$G$KN, we show that the traditional minimal coupling term 
in (\ref{eq:DirEqA}) which describes the effects of $\bAKNanom$ on a test point charge is actually obtained from an electromagnetic
interaction integral which is completely symmetric in the two objects ``point charge'' and ``z$G$KN ring singularity'' (appended by
a KN-anomalous magnetic moment).

 We start with the following Dirac equation,
\beq\label{eq:Dir}
 \tilde{\ga}^\mu\left( -i\hbar \nab_\mu  - \cP_\mu \right)\Psi  + m \Psi= 0,
\eeq
where $\nab$ and $\tilde{\ga}$ are as before, and where the four-covector $\cP$ is the {\em interaction energy-momentum ``vector''} 
describing the mutual electromagnetic interaction of the z$G$KN singularity and the point charge. 
 The interaction energy-momentum vector is defined as follows.

 Recall the energy(-density)-momentum(-density)-stress tensor $\bT$ with components $T_{\mu\nu}$ given by \refeq{eq:Tmunu};
for our stationary fields it is time-independent, depending only on the space variable $\bs$, and in addition also on 
$|{\bq}_{\textrm{pt}}- {\bq}_{\mathrm{rg}}|, ({\bq}_{\textrm{pt}}- {\bq}_{\mathrm{rg}})\cdot \bn_{\mathrm{rg}}$, $a$, 
and on $\textsc{qq}'$ and $\textsc{iq}'$ (if a Kerr--Newman-anomalous magnetic moment is added), as parameters.
 Then $\cP$ is by definition the four-vector field integral
$$
\cP_\mu := \mbox{f.p.} \int_{\cN} T_{\mu 0}(\bq) d^3s,
$$
where $\mbox{f.p.}$ stands for {\em finite part}.
  More to the point, we resort to the early ``pedestrian'' renormalization recipe, as explained e.g. in \cite{JacksonBOOK} for the
electrostatic $N$-body Coulomb interaction, and as explained in our context next.
 Explicitly, we consider the classical electromagnetic problem of computing, {\em in the quasi-static approximation}, the interaction 
energy of a (generalized) z$G$KN singularity of charge $\textsc{q}$ (as seen in the $r>0$ sheet) with a point charge
$\textsc{q}'$ that sits elsewhere in the branched Riemann space whose branch curve is the ring singularity.
 Recall that with Maxwell's vacuum law $\bD = \bE$, $\bH = \bB$, 
the energy-momentum-stress tensor $T_{\mu\nu}$ of an electromagnetic field tensor $\bF = d\bA$ is defined to be
$$
T_{\mu\nu} = \frac{1}{4\pi} \left\{ F_\mu^\la\star F_{\la\nu} - \frac{1}{4} F_{\mu\nu}F^{\mu\nu} \right\},
$$
and thus in particular
$$
T_{00} = \frac{1}{8\pi} \left( |\bE|^2 + |\bB|^2 \right),\qquad
T_{0j} = \frac{1}{4\pi}( \bE \times \bB)_j.
$$
 Here for convenience we have already switched to the Euclidean vector notation (the local tangent space of z$G$K is always 
Minkowski spacetime); we note that we commit a slight abuse of notation, as we already defined $\bE$ and $\bB$ and $\bA$ as one-forms.

 Be that as it may, let the spacetime $\cM$ be a copy of z$G$KN, with its ring singularity centered at $\bq_{\mathrm{rg}}$ and of 
radius $|a|$ and orientation $\bn_{\mathrm{rg}}$, and let us suppose that a point charge is located at ${\bq}_{\textrm{pt}} \in \cN$.
 We note that $\bE$ and $\bB$ are of course the {\em total} fields, computed from $\bAKNhydro$,
 so that for the case at hand, namely the (generalized) z$G$KN 
singularity plus point charge system, they include self-field contributions from both the (generalized) z$G$KN singularity and 
the point charge, in addition to interaction terms.
 By the linearity of the Maxwell vacuum system, we can write $\bE = \bE_{\textrm{pt}} + \bE_{\textsc{KN}}$ 
and likewise 
$\bB =  \bB_{\textsc{KN}}^{\mbox{\tiny\textrm{gen}}}$,
 where $\bE_{\textsc{KN}}$ and $\bB_{\textsc{KN}}^{\mbox{\tiny\textrm{gen}}}$ are the (generalized) Kerr--Newman fields 
associated with (generalized) z$G$KN, and similarly $\bE_{\textrm{pt}}$ is the electric field generated by the static point 
charge located elsewhere in the z$G$K spacetime; we have simplified matters by taking $\bB_{\textrm{pt}} = 0,$ as
appropriate for a point charge, and since the addition of a magnetic point dipole field $\bB_{\textrm{pt}}$ would 
in fact be catastrophic.
 Thus we have
$$
4\pi T_{00} = \half |\bE_{\textrm{pt}}|^2 + \half |\bE_{\textsc{KN}}|^2 + \half |\bB_{\textsc{KN}}^{\mbox{\tiny\textrm{gen}}}|^2 +
\bE_{\textrm{pt}} \cdot \bE_{\textsc{KN}}.
$$
 The first three terms in the above are self-energy terms, the integrals of which over the whole space diverges.
 At the same time, these quantities would be finite for instance if the charges were smeared out over a small region, 
and then the result would be a (large) constant, but in particular independent of the locations of the particles.
 Thus, {\em in calculations where only energy differences are important} (recall that only differences of eigenvalues
of the Hamiltonian have physical spectral significance),
these infinite self-energy terms may be ignored, so that only the last term $\bE_{\textrm{pt}}\cdot \bE_{\textsc{KN}}$ needs to be evaluated.
 Similarly,
$$
4 \pi T_{0j} = (\bE_{\textrm{pt}}\times \bB_{\textsc{KN}}^{\mbox{\tiny\textrm{gen}}})_j + 
(\bE_{\textsc{KN}} \times \bB_{\textsc{KN}}^{\mbox{\tiny\textrm{gen}}})_j,
$$
and this time only the first term needs to be computed, the second one integrating to an infinite self-interaction term.
  Thus,
$$
 4\pi \cP_0 = \int_\cN \bE_{\textrm{pt}} \cdot \bE_{\textsc{KN}} d^3s,
\qquad 4\pi \cP_j = \int_\cN (\bE_{\textrm{pt}} \times \bB_{\textsc{KN}}^{\mbox{\tiny\textrm{gen}}})_j d^3s.
$$

\begin{rem}\emph{
In our narrative we  have been talking about (generalized) z$\,G$KN and point charge fields, but the interaction energy-momentum integrals
can be used to compute the quasi-static interaction of any two point- or not-point-like electromagnetic objects (subscripts ${}_1$ and 
${}_2$) --- in general this will then involve integrands of the type $\bE_1 \cdot \bE_2$, $\bB_1 \cdot \bB_2$, 
$\bE_1 \times \bB_2$, and $\bE_2 \times \bB_1$.}
\end{rem}

\begin{rem}\emph{
With this $\cP$ in our Dirac equation \refeq{eq:Dir} its gauge invariance seems lost.
This seeming contradiction is resolved by noting that energy-momentum conservation still holds if to $T_{0,\mu}$ we add any
divergence-free four gradient $\p_\mu\Upsilon$, corresponding to allowing gauge transformations satisfying the
Lorenz--Lorentz gauge condition --- which one may want to use to keep the theory Lorentz invariant.}
\end{rem}

 We represent 
$$
 \bE_{\textsc{KN}} = -\nabla \phi_{\textsc{KN}}, 
\qquad \bB_{\textsc{KN}}^{\mbox{\tiny\textrm{gen}}} = \nabla\times \bA_{\textsc{KN}}^{\mbox{\tiny\textrm{gen}}}.
$$
 On the other hand, $\bE_{\textrm{pt}} =-\grad \phi_{\textrm{pt}}$, the electrostatic field generated by the point charge 
located at ${{\bq}_{\textrm{pt}}}$, is a gradient, too, with $\phi_{\textrm{pt}}$ as in \refeq{phipt}.

 We are now ready to compute the interactions $\cP_\mu$.  
 We shall show that $\cP_\mu = \textsc{q}' {\AKNanom}_\mu$. 

\begin{prop}\label{prop:P0}
$$
 \cP_0  = \textsc{q}' \phi_{\textsc{KN}}(\bq_{{\mathrm{pt}}}). 
$$
\end{prop}

\begin{proof}
 We need to compute 
$$
\int_\cN \grad \phi_{\textrm{pt}} \cdot \grad \phi_{\textsc{KN}} d^3s.
$$
 The potential $\phi_{\textrm{pt}}$ is singular at its pole ${\bq}_{\textrm{pt}}$, while the potential $\phi_{\textsc{KN}}$ is singular 
on the ring $(\xi=0,\eta=0, 0\leq \varphi\leq 2\pi)$ (in oblate spheroidal coordinates $(\xi,\eta,\varphi)$).  
 The singularities are sufficiently mild so that the their gradients are locally integrable over $\cN$; i.e., the 
above energy integral exists.
 To evaluate it, we excise $\ep$-neighborhoods of the singularities and write the above energy integral as limit
$\eps\to0$ of the integral over the remaining domain.
 We can perform integration by parts and use the fact that $\phi_{\textrm{pt}}$ and $\phi_{\textsc{KN}}$ satisfy the Poisson equation with sources
supported in the excised regions to  convert the integral over the remaining region into a sum of integrals over the
excised regions. 
 Using Poisson's equation the evaluation is immediate.

 Thus let $B_\ep$ be the Euclidean ball of radius $\ep$ centered at ${\bq}_{\textrm{pt}}$, and let $T_\ep $
be the connected sum of two Euclidean tori centered on the ring singularity which are cut and re-glued exactly like $\cN$.
  Let $\cN_\ep := \cN \setminus (B_\ep \cup T_\ep)$.
Then, since $\phi_{\textsc{KN}}$ is harmonic away from the ring and $\phi_{\textrm{pt}}$ is finite away from the point charge, and
using also that $\phi_{\textrm{pt}}$ satisfies Poisson's equation with a point source inside $B_\ep$ while 
$\phi_{\textsc{KN}}$ is harmonic inside $B_\ep$, 
and noting the convention that $\bn$ is always the \emph{outward} normal to the indicated oriented domain of integration,
we find
\begin{eqnarray*}
        \int_{\cN_\ep} \grad \phi_{\textrm{pt}} \cdot \grad \phi_{\textsc{KN}} d^3s 
& = & 
-\int_{\cN_\ep}\phi_{\textsc{KN}}\De_{\cN}^{}\phi_{\textrm{pt}}d^3s +\int_{\p \cN_\ep} \phi_{\textsc{KN}} \grad\phi_{\textrm{pt}} \cdot \bn dS,\\
& = & 
0 - \int_{\p B_\ep} \phi_{\textsc{KN}} \grad\phi_{\textrm{pt}} \cdot \bn dS -
 \int_{\p T_\ep} \phi_{\textsc{KN}} \grad\phi_{\textrm{pt}} \cdot \bn dS,\\
& = & 
-\int_{B_\ep}\grad\phi_{\textsc{KN}}\cdot \grad\phi_{\textrm{pt}} d^3s - \int_{ B_\ep} \phi_{\textsc{KN}} \De_{\cN}^{}\phi_{\textrm{pt}} d^3s  \\
&&
-\int_{T_\ep} \grad\phi_{\textsc{KN}}\cdot \grad\phi_{\textrm{pt}} d^3s - \int_{ T_\ep} \phi_{\textsc{KN}} \De_{\cN}^{}\phi_{\textrm{pt}} d^3s  \\
& = & 
O(\ep) +O(\ep^{1/2}) - \int_{ B_\ep} \phi_{\textsc{KN}} \De_{\cN}^{}\phi_{\textrm{pt}} d^3s  \\
& = & 
O(\ep^{1/2}) + 4\pi \textsc{q}'\phi_{\textsc{KN}}(\bq_{\textrm{pt}}).
\end{eqnarray*}
Finally, letting $\eps\to0$ and dividing by $4\pi$ we obtain
$$
\cP_0 = \textsc{q}' \phi_{\textsc{KN}}(\bq_{\textrm{pt}}) .
$$
\end{proof}

\begin{prop}\label{prop:Pj}
 For $j\in\{1,2,3\}$, we have
$$ 
\cP_j 
=
 \textsc{q}' \bA_{\textsc{KN}}^{\mbox{\tiny\textrm{gen}}}(\bq_{\mathrm{pt}})_j.
$$
\end{prop}

\begin{proof}
 We proceed by analogy to our proof of the previous proposition.
 Thus, with $\bE_{\textrm{pt}} = - \grad \phi_{\textrm{pt}}$ and 
$\bB_{\textsc{KN}}^{\mbox{\tiny\textrm{gen}}} = \grad\times \bA_{\textsc{KN}}^{\mbox{\tiny\textrm{gen}}}$, and 
with $\grad\cdot \bA_{\textsc{KN}}^{\mbox{\tiny\textrm{gen}}} =0$ (the Coulomb gauge), we compute
\begin{eqnarray*}
\int_{\cN_\ep}\!\! \grad \phi_{\textrm{pt}} \times \grad \times \bA_{\textsc{KN}}^{\mbox{\tiny\textrm{gen}}} d^3s 
&\!\! =\!\! &
\int_{\cN_\ep} \left[\grad\times \left(\phi_{\textrm{pt}} \grad \times \bA_{\textsc{KN}}^{\mbox{\tiny\textrm{gen}}} \right) 
-  \phi_{\textrm{pt}} \grad \times \grad \times \bA_{\textsc{KN}}^{\mbox{\tiny\textrm{gen}}}  \right] d^3s \\
&\!\! =\!\! & 
\int_{\p \cN_\ep} \phi_{\textrm{pt}} \bn \times \grad \times \bA_{\textsc{KN}}^{\mbox{\tiny\textrm{gen}}} dS 
+ \int_{\cN_\ep} \phi_{\textrm{pt}}  \De_{\cN}^{} \bA_{\textsc{KN}}^{\mbox{\tiny\textrm{gen}}}  d^3s \\
&\!\! =\!\! & 
- \int_{\p \cB_\ep} \phi_{\textrm{pt}} \bn \times \grad \times \bA_{\textsc{KN}}^{\mbox{\tiny\textrm{gen}}} dS 
- \int_{\p \cT_\ep} \phi_{\textrm{pt}} \bn \times \grad \times \bA_{\textsc{KN}}^{\mbox{\tiny\textrm{gen}}} dS 
+ 0 \\
&\!\! =\!\! & 
- \int_{\cB_\ep}\! \grad \phi_{\textrm{pt}}  \times \grad \times \bA_{\textsc{KN}}^{\mbox{\tiny\textrm{gen}}} d^3s
- \int_{\cT_\ep}\! \left[\grad \phi_{\textrm{pt}}  \times \grad \times \bA_{\textsc{KN}}^{\mbox{\tiny\textrm{gen}}} 
- \phi_{\textrm{pt}}  \De_{\cN}^{} \bA_{\textsc{KN}}^{\mbox{\tiny\textrm{gen}}}\right]  d^3s \\
&\!\! =\!\! & 
O(\ep) + O(\ep^{1/2})
+ \int_{\cT_\ep} \phi_{\textrm{pt}}  \De_{\cN}^{} \bA_{\textsc{KN}}^{\mbox{\tiny\textrm{gen}}}  d^3s .
\end{eqnarray*}
 Once again, the $O(\ep^p)$ vanish in the limit $\ep \to 0$, thus leaving only the
last integral on the torus to consider.
  But, by the Maxwell--Amp\'ere law, 
$$
 - \De_{\cN}^{} \bA_{\textsc{KN}}^{\mbox{\tiny\textrm{gen}}} = 4\pi \bJ_{\textsc{KN}}^{\mbox{\tiny\textrm{gen}}},
$$
where $\bJ_{\textsc{KN}}^{\mbox{\tiny\textrm{gen}}}$ is the electric current density (a measure) concentrated 
on the ring singularity (for this interpretation we need to consider the continuous extension of $\cN$ into its 
ring singularity, which of course is no longer a differentiable manifold, but a geometrical space.
  Now we recognize that $\phi_{\textrm{pt}}$ is nothing but $\textsc{q}'\times$ the Green function 
for the negative Laplacian on $\cN$, and so we have
\begin{equation}
\int_{\cN} \bE_{\textrm{pt}} \times \bB_{\textsc{KN}}^{\mbox{\tiny\textrm{gen}}} d^3s 
=
4\pi \textsc{q}' \bA_{\textsc{KN}}^{\mbox{\tiny\textrm{gen}}}(\bq_{\textrm{pt}}) .
\end{equation}
 Dividing by $4\pi$ completes the proof.
\end{proof}

 Thus for $\cP$ we have recovered the KN electromagnetic potential (possibly with an anomalous magnetic moment),
evaluated at the location of the point charge and multiplied by its charge.
 This is exactly the minimal coupling formula for the electromagnetic potential of the point charge
(treated as a test particle) in the electromagnetic field of the (generalized) KN ring singularity.
 At the same time, the symmetry w.r.t. the two fields of our interaction integral formula makes it plain that we are
really computing the \emph{mutual interaction} of a (generalized) z$G$KN ring singularity with a point charge, and the same
formula should also be used in a true two-body problem, not just the two alternate one-body problems in ``external fields''
that we have kept talking about.

\begin{rem}\emph{
 In the non-relativistic limit our result reduces to nothing but Newton's ``actio equals re-actio'' principle, which implies that 
the potential energy of object one in the force field of object two equals the potential energy of object two in the force field of 
object one.
 With hindsight, one could have elegantly argued this way up front, yet a non-relativistic principle in a relativistic
context would certainly have been viewed with some suspicion.
 Our relativistic energy-momentum formula now completely vindicates this heuristic extension of Newton's principle.
 In this vein, in the interpretation of \refeq{eq:DirEqA} where $\Psi$ is a  bi-spinor wave function for the z$\,G$KN 
ring singularity the interaction term is properly called a ``minimal re-coupling'' term, as explained earlier.}
\end{rem}

\subsubsection{The frame formulation of the Dirac equation}\label{sec:Cartan}

 Equation \refeq{eq:DirEqA} is formulated in a nicely compact manner which, however, is not very useful for computations.
 Using Cartan's frame method (see \cite{BrillCohen66} and refs. therein) one can express the covariant derivative on spinors 
in terms of standard derivatives:
\beq
\tilde{\ga}^\mu \nab_\mu    = \ga^\mu\be_\mu  +\frac{1}{4} \Om_{\mu\nu\la} \ga^\la\ga^\mu\ga^\nu.
\eeq
Here $\{\be_\mu\}_{\mu=0}^3$ is a {\em Cartan frame}, i.e. an orthonormal frame of vectors spanning the tangent space at each point
of the spacetime manifold.
We thus have
\beq 
(\be_\mu)^\nu (\be_\la)^\ka {g}_{\nu\ka} = {{\eta}}_{\mu\la},
\eeq
where 
\beq
({\boldsymbol{\eta}}) = \diag(1,-1,-1,-1)
\eeq 
is the matrix of the Minkowski metric in rectangular coordinates.
 On the one hand, it follows that
\beq
\tilde{\ga}^\mu = (\be_\nu)^\mu \ga^\nu,
\eeq
where the $\ga^\nu$ are Dirac matrices for the Minkowski spacetime, satisfying 
$\ga^\nu\ga^\mu+\ga^\mu\ga^\nu = 2{{\eta}}^{\mu\nu}{\boldsymbol{1}}_{4\times4}$.
  On the other hand, let $\{\boldsymbol{\om}^\mu\}_{\mu=0}^3$ denote the {\em dual} frame to $\{\be_\mu\}$,
 i.e. the orthonormal basis for the cotangent space at each point of the manifold that is dual to the basis for the tangent space:
\beq 
\boldsymbol{\om}_\mu \big(\be^\nu\big)  = \be^\nu \big( \boldsymbol{\om}_\mu\big) = \de_\mu^\nu.
\eeq
Then the $\Om_{\mu\nu\la}$ are by definition the {\em Ricci rotation coefficients} of the frame $\{\boldsymbol{\om}^\mu\}_{\mu=0}^3$, 
defined in the following way:  Let the one-forms $\bOm^\mu_\nu$ satisfy
\beq
d \boldsymbol{\om}^\mu + \bOm^\mu_\nu \wedge \boldsymbol{\om}^\nu = 0.
\eeq
This does not uniquely define the $\bOm^\mu_\nu$.
  However, there exists a unique set of such one-forms satisfying the extra condition
\beq
\bOm_{\mu\nu} = - \bOm_{\nu\mu},
\eeq 
where the first index is lowered by the Minkowski metric:   $\bOm_{\mu\nu} := {{\eta}}_{\mu\la}\bOm^{\la}_\nu$.
 Since $\left\{\boldsymbol{\om}^\mu\right\}$ forms a basis for the space of one-forms, we have  
$\bOm_{\mu\nu} = \Om_{\mu\nu\la}\boldsymbol{\om}^\la$, which defines the rotation coefficients $\Om_{\mu\nu\la}$.

The Dirac equation (\ref{eq:DirEqA})
on a spacetime $(\cM,\bg)$ with an electromagnetic four-potential $\bA$ can thus be written in the following form:
\beq\label{eq:DirEqAstandard}
\ga^\mu \left(\be_\mu + \Ga_\mu - i \textsc{q}' \tilde{A}_\mu\right)\Psi  + im\Psi = 0;
\eeq
here, the $\Ga_\mu$ are connection coefficients,
\beq
\Ga_\mu := \frac{1}{4} \Om_{\nu\la\mu} \ga^\nu\ga^\la = \frac{1}{8} \Om_{\nu\la\mu}[\ga^\nu,\ga^\la],
\eeq
and the $\tilde{A}_\mu$ are the components of the potential $\bA$ in the ${\boldsymbol{\om}^\mu}$ basis, 
i.e. $\bA =\tilde{A}_\mu \boldsymbol{\om}^\mu$, or,
\beq
\tilde{A}_\mu  := (\be_\mu)^\nu A_\nu.
\eeq
 Moreover, recall that $\tilde{A}_\mu$ in \refeq{eq:DirEqAstandard} is the ($\mu$ component of the) generalized z$G$KN potential evaluated 
at $\bq_{\textrm{pt}}$ while $\Psi$ is the bi-spinor of the z$G$KN ring singularity evaluated at $(t,\bq_{\textrm{pt}})$, with 
``$\bq_{\textrm{pt}}$'' shorthand for the oblate spheroidal coordinates $(r,\theta,\varphi)$.

As mentioned earlier, a single chart of oblate spheroidal coordinates $(t,r,\theta,\varphi)$ covers the whole 
zero-$G$ Kerr--Newman spacetime $(\cM,\bg)$, and in these coordinates the electromagnetic Appell--Sommerfeld one-form
$\bAKN$ is everywhere on  $(\cM,\bg)$ given by the simple formula (\ref{def:AKN}); also, the extension to incorporate
a Kerr--Newman-anomalous magnetic moment is equally simple, see (\ref{def:AKNanomalous}).
 It is therefore only natural that one would like to write Dirac's equation \refeq{eq:DirEqA} in these coordinates as well, 
in the hope of achieving at least some partial separation of variables.\footnote{The idea of using special frames adapted 
  to a coordinate system in order to separate spinorial wave equations in those coordinates goes back to Kinnersley \cite{Kin69} 
  and Teukolsky \cite{Teu72}.}

 However, unlike Cartesian coordinates $(x^\mu)$ in Minkowski spacetime, oblate spheroidal coordinate derivatives do not give 
rise to an orthonormal basis for the tangent space at each point of a zero-$G$ Kerr spacetime.
 Thus, to bring \refeq{eq:DirEqA} into the Cartan form \refeq{eq:DirEqAstandard} using oblate spheroidal coordinates, one
also needs to construct a suitable Cartan frame. 
  Following Chandrasekhar \cite{ChandraDIRACinKERRgeom,ChandraDIRACinKERRgeomERR}, Page \cite{Page76}, Toop \cite{Too76} (see also
Carter-McLenaghan \cite{CarMcL82}), we  introduce a special orthonormal frame $\{\be_\mu\}_{\mu=0}^3$ on the 
tangent bundle $T\cM$ which is adapted to the oblate spheroidal coordinates in order for the Dirac equation to 
take a comparatively simple form.

We begin by introducing a Cartan (co-)frame $\{\boldsymbol{\om}^\mu\}_{\mu=0}^3$ for the cotangent 
bundle\footnote{This particular frame is called a {\em canonical symmetric tetrad} in \cite{CarMcL82}.}:
\beq
\boldsymbol{\om}^0 := \frac{\varpi}{|\rho|} (dt - a \sin^2\theta\, d\varphi),\quad
\boldsymbol{\om}^1 := |\rho|d\theta,\quad
\boldsymbol{\om}^2 := \frac{\sin\theta}{|\rho|} (-a dt + \varpi^2 d\varphi),\quad
\boldsymbol{\om}^3 := \frac{|\rho|}{\varpi}dr,
\eeq
with the abbreviations
\beq
\varpi := \sqrt{r^2 + a^2}, \quad \rho:= r + i a \cos\theta.
\eeq
Let us denote the oblate spheroidal coordinates $(t,r,\theta,\varphi)$ collectively by $(y^\nu)$.
  Let $g_{\mu\nu}$ denote the coefficients of the spacetime metric \refeq{def:metricBL} in oblate spheroidal coordinates, 
i.e. $g_{\mu\nu} = \bg\Big(\frac{\p}{\p y^\mu},\frac{\p}{\p y^\nu}\Big)$.
     One easily checks that written in the $\{\boldsymbol{\om}^{\mu}\}$ frame, the spacetime line element is
\beq 
ds_{\bg}^2 = g_{\mu\nu}dy^\mu dy^\nu = {{\eta}}_{\alpha\beta} \boldsymbol{\om}^{\alpha}\boldsymbol{\om}^{\beta}.
\eeq
 This shows that the frame $\{\boldsymbol{\om}^\mu\}_{\mu=0}^3$ is indeed orthonormal.
 With respect to this frame the electromagnetic Sommerfeld potential \refeq{def:AKNanomalous} becomes
$\bA = \tilde{A}_\mu \boldsymbol{\om}^\mu$, with
\beq\label{generalKNelmagAtilde}
\tilde{A}_0 =  -\textsc{q}\frac{r}{|\rho|\varpi} -  \left(\textsc{q}-{\textsc{i}\pi a}\right)\frac{a^2r\sin^2\theta}{\varpi |\rho|^3},\quad
\tilde{A}_1 = 0, \quad
\tilde{A}_2 = - \left(\textsc{q}-\textsc{i}\pi a\right)\frac{ar\sin\theta}{|\rho|^3},\quad 
\tilde{A}_3 = 0.
\eeq

\begin{rem}\textit{
We observe that for $\textsc{q} =\textsc{i}\pi a$, all but one of the quantities $\tilde{A}_\mu$ vanish, and the non-vanishing one,
$\tilde{A}_0$, reduces to $-{\textsc{q}r}/{|\rho|\varpi}$.}
\end{rem}

  Next, let the frame of vector fields $\{{\be}_{\mu}\}$ be the {\em dual} frame to $\{\boldsymbol{\om}^{\mu}\}$.
 Thus $\{\be_\mu\}$ yields an orthonormal basis for the tangent space at each point in the manifold:
\beq
\be_0 = \frac{\varpi}{|\rho|} \p^{}_t + \frac{a}{\varpi|\rho|} \p^{}_\varphi,\quad
\be_1 = \frac{1}{|\rho|}\p^{}_\theta,\quad
\be_2 = \frac{a\sin\theta}{|\rho|} \p^{}_t + \frac{1}{|\rho|\sin\theta} \p^{}_\varphi,\quad
\be_3 = \frac{\varpi}{|\rho|} \p^{}_r\;.
\eeq

  Next, the anti-symmetric matrix $\big(\Om_{\mu\nu}\big) = \big({{\eta}}_{\mu\la}\Om^\la_\nu\big)$ is computed to be
\beq
(\Om_{\mu\nu}) = \left(\begin{array}{cccc}
0 & -C\boldsymbol{\om}^0 - D\boldsymbol{\om}^2 & D\boldsymbol{\om}^1 - B\boldsymbol{\om}^3 & -A\boldsymbol{\om}^0- B \boldsymbol{\om}^2 \\
& 0 & D\boldsymbol{\om}^0 + F\boldsymbol{\om}^2 &-E \boldsymbol{\om}^1 - C\boldsymbol{\om}^3  \\
&\textrm{(anti-sym)}& 0 & -B \boldsymbol{\om}^0 - E\boldsymbol{\om}^2 \\
& & & 0
\end{array}\right),
\eeq
with
\beq
A := \frac{a^2 r \sin^2\theta}{\varpi |\rho|^3},\,
B := \frac{a r \sin\theta}{|\rho|^3},\,
C := \frac{a^2 \sin\theta\cos\theta}{|\rho|^3},\,
D := \frac{a\cos\theta\varpi}{|\rho|^3},\,
E := \frac{r\varpi}{|\rho|^3},\,
F := \frac{\varpi^2\cos\theta}{|\rho|^3\sin\theta}.
\eeq

 With respect to this frame on a zero-$G$ Kerr spacetime the covariant derivative part of 
the Dirac operator (\ref{eq:DirEqA}) can be expressed with the help of the operator
\beq
\fO := \tilde{\ga}^\mu\nab_\mu = \left(\begin{array}{cc} 0 & \fl'+\fm'\\ \fl+\fm &0 \end{array}\right),
\eeq
where
\beq
\fl := 
\frac{1}{|\rho|} \left(\begin{array}{cc} D_+ & L_- \\ L_+ & D_-\end{array}\right)
\eeq
and
\beq
\fl' := 
\frac{1}{|\rho|} \left(\begin{array}{cc} D_- & -L_- \\ -L_+ & D_+\end{array}\right),
\eeq
with
\beq\label{eq:DpmLpm}
D_\pm := \pm \varpi \p^{}_r + \left( \varpi \p^{}_t + \frac{a}{\varpi} \p^{}_\varphi\right),
\qquad 
L_\pm :=\p^{}_\theta  \pm i \left(a \sin\theta\,\p^{}_t+\csc\theta \p^{}_\varphi\right),
\eeq
while
\beq
\begin{aligned}
\fm &:=  \half\bigl[ (-2C+F+iB)\si^{}_1+(-A+2E+iD)\si^{}_3\bigr] \\
&\ = \frac{1}{2|\rho|} \left(\begin{array}{cc} \frac{r}{\varpi}+ \frac{\varpi}{{\rho^*}} &\cot\theta + \frac{ia\sin\theta}{{\rho^*}}\\
\cot\theta + \frac{ia\sin\theta}{{\rho^*}} &  -\frac{r}{\varpi} - \frac{\varpi}{{\rho^*}}\end{array}\right)
\end{aligned}
\eeq
and 
\beq
\fm': = \half\bigl[ (2C-F+iB)\si^{}_1+(A-2E+iD)\si^{}_3\bigr] = -\fm^\dagger,
\eeq
and the $\si^{}_k$ are Pauli matrices \refeq{eq:PAULIsigma}, viz.
\beq
 \si^{}_1 = \left(\begin{array}{cc} 0 & 1\\ 1 & 0\end{array}\right),\quad
 \si^{}_2 = \left(\begin{array}{cc} 0 & -i\\ i &\ 0\end{array}\right),\quad
 \si^{}_3 = \left(\begin{array}{cc} 1 & \ 0\\ 0 & -1\end{array}\right).
\eeq

  We note that the principal part of $|\rho|\fO$ has an additive separation property:
\beq\label{eq:DDprincipal}
\begin{aligned}
|\rho|\left(\begin{array}{cc} 0 & \fl'\\ \fl & 0\end{array}\right) 
=
\left[
 \ga^3 \varpi \p^{}_r + \ga^0\left(\varpi \p^{}_t + \frac{a}{\varpi} \p^{}_\varphi\right)\right]
 + 
\Bigl[ \ga^1 \p^{}_\theta + \ga^2(a\sin\theta \p^{}_t + \csc\theta\,\p^{}_\varphi)\Bigr],
\end{aligned}
\eeq
where the coefficients of the two square-bracketed operators are functions of only $r$, respectively only $\theta$.
 Moreover, it is possible to transform away the lower order term in $\fO$, so that exact separation can be achieved for $|\rho|\fO$.
 Namely, let 
\beq
\chi(r,\theta) := \half \log( \varpi {\rho^*}\sin\theta).
\eeq  
 It is easy to see that
\beq
\fm = \fl\chi,\qquad \fm' = \fl'{\chi^*}.
\eeq
 Let us therefore define the diagonal matrix
\beq\label{def:D}
\fD := \diag( e^{-\chi},e^{-\chi}, e^{-{\chi}^*}, e^{-{\chi}^*})
\eeq
and a new bi-spinor $\hat{\Psi}$ related to the original $\Psi$ by
\beq
\Psi = \fD \hat{\Psi}.
\eeq
 Denoting the upper and lower components of a bi-spinor $\Psi$ by $\psi_1$ and $\psi_2$ respectively, it then follows that
\beq
(\fl + \fm)\psi_1 = 
(\fl + \fm)(e^{-\chi}\hat{\psi}_1) =
 e^{-\chi} \left[ \fl - \fl\chi + \fm\right]\hat{\psi}_1 =
 e^{-\chi} \fl\hat{\psi}_1,
\eeq
and similarly
\beq
(\fl'+ \fm')\psi_2 = e^{-{\chi}^*} \fl'\hat{\psi}_2.
\eeq

 We now put it all together.
 We set 
\beq
\fR := \diag(\rho,\rho,{\rho}^*,{\rho}^*)
\eeq 
and note that $|\rho|\fD^{-\dagger}\fD = \fR$ while $\fD^{-\dagger}\ga^\mu\fD = \ga^\mu$; here,
$\fD^{-\dagger}$ is shorthand for $(\fD^{-1})^\dagger$.
 Thus, setting $\Psi = \fD \hat{\Psi}$ in \refeq{eq:DirEqA} and left-multiplying the equation by the diagonal matrix 
$\fD' := |\rho|\fD^{-\dagger}$  we  conclude that $\hat{\Psi}$ solves a new Dirac equation
\beq\label{eq:newDir}
\left(|\rho|\ga^\mu (\be_\mu - i\textsc{q}'\tilde{A}_\mu) + im\fR\right) \hat{\Psi} = 0.
\eeq

\subsection{The Dirac Hamiltonian for a general-relativistic zero-gravity Hydrogen atom (in the Born--Oppenheimer approximation)}

  To make contact with the physical Hydrogen problem, we henceforth identify the electric  charge $\textsc{q}$ of the z$G$KN 
spacetime (in its $r>0$ sheet) with the electron's empirical negative elementary charge, $\textsc{q}=-e$; note that in the other 
sheet, the  charge is automatically that of the positron, $+e$.
  The point charge $\textsc{q}'$ with which the z$G$KN singularity is interacting electromagnetically is chosen to be
the proton's charge, $\textsc{q}'=+e$, which we treat as classical; thus we will speak of the general-relativistic zero-gravity 
Born--Oppenheimer Hydrogen problem.
  Lastly, the mass parameter $m$ in Dirac's equation we identify with the ``empirical mass of the electron, 
$m_{\mbox{\tiny{e}}}$.''\footnote{Incidentally, one would therefore also like to identify $m_{\mbox{\tiny{e}}}$ 
	with the ``mass of the z$G$KN singularity,'' but since the z$G$KN metric does not contain a mass parameter 
	this would be just a convenient ``way of speaking,'' not a concept backed up by any calculation.
	 In fact, the only ``mass'' one could possibly assign it by any calculation is its electromagnetic mass, 
	but as for the Coulomb point charge used in the familiar textbook treatments of Dirac Hydrogen, 
	one finds that the electromagnetic field energy of the z$G$KN ring singularity is infinite. 
	 This is not a problem for computing quantum-mechanical spectra \emph{as long as} $G=0$ because then the self-energies
	only cause an undetectable overall shift of the spectrum, which does not affect the \emph{energy differences}, viz. 
	the emission / absorption frequencies of the model.
 	The issue becomes important once $G$ is switched on; see the end of our last section.}

\subsubsection{The Dirac Equation in Hamiltonian Form}\label{sec:HamDir}

  Let us compute the Hamiltonian form of \refeq{eq:newDir}. 
  Let matrices $M^\mu$ be defined by
\beq
|\rho|\ga^\mu\be_\mu = M^\mu\p^{}_\mu.
\eeq
  In particular,
\beq
M^0 = \varpi\ga^0 + a \sin\theta\, \ga^2.
\eeq
 We may thus rewrite \refeq{eq:newDir} as
\beq
M^0 \p^{}_t\hat{\Psi} = - \left( M^k\p^{}_k - ie|\rho|\ga^\mu\widetilde{\AKNanom}_\mu + im\fR\right)\hat{\Psi}.
\eeq
 Finally, restoring $\hbar$ and $c$, and the argument of $\AKNanom$, we define
\beq\label{def:Hhat}
\hat{H} := 
(M^0)^{-1} \left( M^k (-i \hbar \p^{}_k) 
- \textstyle{\frac1c} e|\rho|\ga^\mu\widetilde{\AKNanom}_\mu(\bq_{\textrm{pt}}) + mc\fR\right),
\eeq
and can now rewrite the Dirac equation \refeq{eq:newDir} for $\Psi(t,\bq_{\textrm{pt}})$ 
in Hamiltonian form:
\beq\label{eq:DIRACeqHAMformat}
i \hbar \p^{}_t \hat{\Psi} = \hat{H}\hat{\Psi}.
\eeq

\subsubsection{A Hilbert space for $\hat{H}$}\label{sec:Hilbert}

 The correct positive-definite inner product for the space of bi-spinor fields defined on the z$G$KN spacetime can be
extracted from the action for the original Dirac equation \refeq{eq:DirEqA}, which is obtainable from this equation 
upon left-multiplying it by the conjugate bi-spinor $\overline{\Psi}$, defined as
\beq
\overline{\Psi} := \Psi^\dag \ga^0,
\eeq
and integrating the result over a slab of the spacetime. 
 Thus,
\beq
\cS[\Psi] = \int_{t_1}^{t_2} dt \int_{\Si_t} \Psi^\dag \ga^0 \left[ \tilde{\ga}^\mu \nab_\mu \Psi + \dots \right] d\mu^{}_{\Si_t},
\eeq
where $d\mu^{}_{\Si_t}$ is the volume element of $\Si_t\equiv\cN$, any spacelike $t=$ constant slice of z$G$KN.
 Using oblate spheroidal coordinates, $d\mu^{}_{\cN} = |\rho|^2\sin\theta  d\theta d\varphi dr$, so
the natural inner product for bi-spinors on $\Si_t=\cN$ reads
\beq
\langle \Psi,\Phi\rangle = \int_\cN \Psi^\dag\ga^0\tilde{\ga}^0 \Phi d\mu^{}_\cN 
= \int_0^{2\pi}\int_0^\pi \int_{-\infty}^\infty \Psi^\dag M \Phi |\rho|^2 \sin\theta d\theta d\varphi dr,
\eeq
with
\beq
M := \ga^0 \tilde{\ga}^0 = \ga^0 \be_\nu^0 \ga^\nu = \frac{\varpi}{|\rho|} \al^0 + \frac{a\sin\theta}{|\rho|} \al^2.
\eeq
Here, $\alpha^2$ is the second one of the three Dirac alpha matrices in the Weyl (spinor) representation
\beq
\al^k = \ga^0 \ga^k = \left(\begin{array}{cc} \si^{}_k & 0 \\ 0 & -\si^{}_k\end{array}\right),\qquad k=1,2,3,
\eeq
and for notational convenience the $4\times4$ identity matrix has been denoted by
\beq
\al^0 =  \left(\begin{array}{cc} \boldsymbol{1}_{2\times2} & 0 \\ 0 & \boldsymbol{1}_{2\times2}\end{array}\right).
\eeq

Now, let $\Psi = \fD \hat{\Psi}$ and $\Phi = \fD \hat{\Phi}$, with $\fD$ as in \refeq{def:D}.
  Then we have
\beq
\langle \Psi,\Phi \rangle = \int_0^{2\pi}\int_0^\pi \int_{-\infty}^\infty \hat{\Psi}^\dag \hat{M} \hat{\Phi} d\theta d\varphi dr,
\eeq
where
\beq
\hat{M} := \al^0 + \frac{a\sin\theta}{\varpi} \al^2.
\eeq
The eigenvalues of $\hat{M}$ are $\la_\pm = 1 \pm \frac{a\sin\theta}{\varpi}$, both of which have multiplicity 2 and
are positive everywhere on this space with Zipoy topology.
	(Note that $\la_- \to 0$ on the ring, which is not part of the space time but at its boundary.)
	We may thus take the above as the definition of a positive definite inner product
 given by the matrix $\hat{M}$ for bi-spinors defined on the $t=const.$ section of $\cM$, a rectangular cylinder
$\rcyl :=\RR\times [0,\pi]\times  [0,2\pi]$  with its natural measure:
\beq\label{def:innerPROD}
\langle \hat{\Psi},\hat{\Phi}\rangle_{\hat{M}} := \int_{\rcyl} \hat{\Psi}^\dag\hat{M} \hat{\Phi} d\theta d\varphi dr.
\eeq
 The corresponding Hilbert space is denoted by ${\sf H}$, thus
\beq
{\sf H} 
:= \left\{ \hat\Psi:\rcyl \to \Cset^4\ | \ \|\hat\Psi\|_{\hat{M}}^2 := \langle \hat\Psi,\hat\Psi\rangle_{\hat{M}} < \infty \right\}.
\eeq
 Note  that ${\sf H}$ is \emph{not equivalent} 
to $L^2(\rcyl)$ whose inner product has the identity matrix in place of $\hat{M}$.

 After these preparations we are now ready to list our main results which are proved in \cite{KieTah14a}.

 Our results about  the symmetry of the spectrum are valid with or without the presence of a KN-anomalous magnetic moment, i.e.
for \emph{any} self-adjoint extension of $\hat{H}$, whatever $\textsc{q}$ and $\textsc{i}$.
 The essential self-adjointness, and the location of essential and point spectra, are stated only for the proper z$G$KN
singularity (i.e. no KN-anomalous magnetic moment) interacting with a point ``proton;'' however, we conjecture that these 
results will continue to hold as long as the KN-anomalous magnetic moment is sufficiently small.

\subsubsection{Symmetry of the spectrum of $\hat{H}$}

 Let $\hat{S}: {\sf H} \to {\sf H}$ denote the {\em sheet swap} map $\hat{S}\hat\Psi(x) = \hat\Psi(\varsigma(x))$ and $\hat{K}$  the 
complex conjugation operator $\hat{K}\hat\Psi(x) = \hat\Psi^*(x)$. 
 Let
$$
\hat{C} := \ga^0 \hat{K}\hat{S}.
$$
 One readily checks that $\hat{C}$ anti-commutes with the full Hamiltonian, and thus if $\hat\Psi$ is an eigen-bi-spinor of $\hat{H}$
with eigenvalue $E$, then $\hat{C}\hat\Psi$ is an eigen-bi-spinor of $\hat{H}$ with eigenvalue $-E$.  

 More generally, with the help of the operator $\hat{C}$, which anti-commutes with any self-adjoint extension of
the formal Dirac operator $\hat{H}$ on $\sf H$, and an argument of Glazman \cite{Gla65}, in \cite{KieTah14a} we prove:

\begin{thm}\label{thm:sym}
  Let any self-adjoint extension of the formal Dirac operator $\hat{H}$ on $\sf H$ be denoted by the same letter.
  Suppose $E\in\mathrm{spec}\,\hat{H}$. 
  Then $-E\in\mathrm{spec}\,\hat{H}$. 
\end{thm}

\begin{rem}\emph{
 The operator $\hat{C}$ should not be confused with the {\em charge conjugation} operator
$$
\tilde{C} := i\ga^2 \hat{K}.
$$
 One easily checks that if $\hat\Psi$ solves $i\hbar\p_t\hat\Psi = (\hat{H}_0 + e\cA)\hat\Psi$, then $\tilde{C}\hat\Psi$ solves
$i\hbar\p_t(\tilde{C}\hat\Psi) = (\hat{H}_0 - e\cA)(\tilde{C}\hat\Psi)$.
 In particular, if $\hat\Psi$ is an eigen-bi-spinor of $\hat{H}_0 + e\cA$ with eigenvalue $E$, then $\tilde{C}\hat\Psi$ is an eigen-bi-spinor of 
$\hat{H}_0 - e\cA$ with eigenvalue $-E$ (note that the two Hamiltonians here are different).}
\end{rem}

\subsubsection{Essential self-adjointness of $\hat{H}$}

 Let $\rcyl^\ast$ denote $\rcyl$ with the ring singularity $\{r,\theta,\varphi|r=0,\theta=\pi/2\}$ deleted.
 By adapting an argument of Winklmeier--Yamada \cite{WINKLMEIERc}, in \cite{KieTah14a} we prove:
\begin{thm}\label{thm:esa}
For $\textsc{q} = -e = \textsc{i}\pi a$,  the operator $\hat{H}$ with domain $C^\infty_c(\rcyl^\ast)$ is e.s.a. on~$\sf H$.  
\end{thm}

\subsubsection{The continuous spectrum of $\hat{H}$}

With the help of the Chandrasekhar--Page--Toop formalism to separate variables, and an argument of Weidmann \cite{Wei82},
in \cite{KieTah14a} we prove:

\begin{thm}\label{thm:essspec}
For $\textsc{q}=-e=\textsc{i}\pi a$, the continuous spectrum of $\hat{H}$ on $\sf H$ is $\Rset\setminus(-m,m)$.
\end{thm}

\subsubsection{The point spectrum of $\hat{H}$}

Using the Chandrasekhar--Page--Toop formalism, and the Pr\"ufer transform, in \cite{KieTah14a} we prove:

\begin{thm}\label{thm:ptspec}
  Let $\textsc{q} =-e = \textsc{i}\pi a$.
  Then, if                                      
$|a|m<\half$ and $e^2<\sqrt{2|a|m(1-2|a|m)}$,
the point spectrum of $\hat{H}$ on $\sf H$ is nonempty and located in $(-m,m)$; the end points are not included.
 Moreover, the eigenvalues stand in one-to-one correspondence with pairs of heteroclinic orbits connecting two saddle
points in two parameter-dependent flows on truncated cylinders.
\end{thm}

\begin{rem}\emph{
 In \cite{KieTah14a} we surmise that the winding numbers of the heteroclinic orbits enumerate the energy levels (or possibly 
certain finite families of levels) of the z$\,G$KN-Dirac Hamiltonian, and that their right- vs. left-handedness corresponds to positive, 
resp. negative energy eigenvalues.}
\end{rem}

\begin{rem}\emph{
Resorting to units with $\hbar$ and $c$ restored, our smallness conditions become 
$2|a|<\frac{\hbar}{mc}$ and $\frac{e^2}{\hbar c} <\sqrt{\frac{2|a|}{\hbar/mc}\big(1-\frac{2|a|}{\hbar/mc}\big)}$.
 The condition on the ring diameter $2|a|$ demands that it be smaller than the electron's Compton wavelength.\footnote{For convenience, 
we have dropped the adjective ``reduced,'' even though ``Compton wavelength'' usually refers to $2\pi$ times the reduced expression, i.e.
to $h/mc$.}
 Now it would seem natural to identify the zero-$G$ Kerr--Newman magnetic moment $-{e}a$ with the negative of the Bohr magneton
(as done for instance by Carter \cite{Car68}), which yields $2|a| = \frac{\hbar}{mc}$, and then our condition would be violated.
 However, we shall see below that this choice, while suggestive --- and indeed entertained by us for a while --- is too naive!
 Instead, we will compellingly argue that $2|a|$ should only be about one-thousands of the electron's Compton wavelength.
 As to the second condition, note that $\frac{e^2}{\hbar c} = \al_{\mbox{\tiny\textsc{s}}}^{}\approx\frac{1}{137.036}$, and
this condition is satisfied for our proposed choice of $|a|$.
 More interesting in this regard is the \emph{hydrogenic} problem where the point ``proton'' charge $e$ is replaced by the
charge $Ze$ of a  point ``nucleus,'' with $Z>1$, in which case we get a point spectrum in the gap of the continuum as long as
$Z< 137.036 \sqrt{\frac{2|a|}{\hbar/mc}\big(1-\frac{2|a|}{\hbar/mc}\big)}$; this indicates that our estimate is presumably 
not sharp (for our upper bound on $Z$ goes $\downarrow 0$ as $|a|\downarrow 0$, while the familiar Dirac bound $Z<137.036$
for the existence of a point spectrum in the  hydrogenic problem on Minkowski spacetime should be obtained instead).}
\end{rem}

This completes the summary of our main results from \cite{KieTah14a}.

\subsubsection{On the computability of the eigenvalues}

 If the KN-anomalous magnetic moment is absent, i.e. if we assume that $\textsc{q} =-e=\textsc{i}\pi a$, 
then $|\rho|\ga^\mu\tilde{A}_\mu$ reduces to  $|\rho|\ga^0\tilde{A}_0 = -({\textsc{q}r}/{\varpi})\ga^0$, which is a function of only $r$.
 The separation of variables Ansatz of Chandrasekhar \cite{ChandraDIRACinKERRgeomERR,ChandraDIRACinKERRgeom}, Page \cite{Page76}, 
and Toop \cite{Too76} now yields a system of ordinary differential equations which facilitates the computability of the
point spectrum.

\begin{rem}\textit{
For the convenience of the reader, in the Appendix B we also recall how to separate the variables for the z$\,G$KN Dirac equation, 
in the absence of the anomalous terms, applying the method of Chandrasekhar--Page--Toop.}
\end{rem}

However, unlike the familiar ODE system for the Dirac equation of Born--Oppenheimer Hydrogen in Minkowski spacetime, 
the (say) z$G$KN-Dirac ODE system does not have a ``triangular'' structure which would allow its solution one ODE at a time in what
may be called the bottom-up direction. 
 Instead, two of the equations are intrinsically coupled, and their joint solution seems to be feasible only numerically with
the aid of a computer, for a judicious choice of parameter values of $a$ (see further below) and a few energy eigenvalues, e.g. the 
positive and negative ground states and a dozen or so excited states.
 We hope to report on such a study in the not too distant future.

 Next we note that for $\textsc{q} \ne \textsc{i}\pi a$ the quantity $|\rho|\ga^\mu\tilde{A}_\mu$ is a function of both $r$ and $\theta$, 
and unlike the other terms in the Dirac equation \refeq{eq:newDir}  it does {\em not} separate into a sum of two 
terms each depending only on one of these variables. 
It follows that the Dirac equation will not be completely separable when the magnetic moment is different from $-ea$ (viz. $\textsc{q}a$).
A two-dimensional ``vector'' PDE problem needs to be solved to obtain the energy and angular-momentum eigenvalues.
This would with near certainty be feasible only on a computer, too.
 In the meantime, we have to be content with a few conclusions that can be drawn based on our theorems and some further analysis.

\subsubsection{The spectrum in the limit $a\to0$ of the electromagnetic spacetime}

 Under the plausible (but yet to be proved) hypothesis that the limit $a\to 0$ of the z$G$KN-Dirac spectrum coincides with
the spectrum of the Dirac Hamiltonian evaluated directly on the $a\to 0$ limit of the z$G$KN spacetime (or rather, its
$t=0$ spacelike slice), we obtain the following results.

 First, since the $a\neq0$ spectrum is symmetric about zero, it is symmetric also in the limit. 

 Second, the continuous spectrum is always $(-\infty,-m]\cup[m,\infty)$, also in the limit $a\to0$. 

 Coming thus to the point spectrum, we note that since the spacetime is topologically non-trivial, with a perfect (anti-)symmetry 
between its two electromagnetic sheets, also the limiting spacetime when $a\to0$ will be topologically non-trivial, with a 
perfect (anti-)symmetry between its two electromagnetic sheets; however, the geometry degenerates in the limit $a\to0$: as the 
ring radius $|a|$ shrinks to $0$, the ring collapses to a point, and the limit $a\to0$ of the z$G$KN spacetime thus becomes two 
copies of Minkowski spacetime with a straight worldline deleted, but with the continuous extension of the two copies into the
removed worldline identified at that worldline (to visualize this, remove the time and one space dimension and
think of a familiar double cone in $\Rset^3$, pushed flat).
 Moreover, inspection of the generalized KN field \refeq{def:AKNanomalous} reveals that in the limit $a\to0$ with all other parameters
kept fixed the electromagnetic field becomes two copies of a pure Coulomb field, corresponding to a negative point charge $-e$ in
the $r>0$ sheet, and a positive point charge $+e$ in the $r<0$ sheet.
 Lastly, we note that the eigenvalue problem now decouples in the sense that the variations can be carried out restricted to either
the $r>0$ sheet or the $r<0$ sheet. 
 Each of these subproblems leads just to the familiar special-relativistic Born--Oppenheimer ``Hydrogen'' problem, where the quotes
around Hydrogen indicate that in one calculation the electron charge $-e$ is replaced by the positron charge $+e$, and it is well-known
that this produces a negative copy of the familiar Born--Oppenheimer Hydrogen spectrum computed from the Dirac equation by 
Darwin \cite{DarwinDIRACspec} and Gordon \cite{GordonDIRACspec} with the electron charge $-e$; cf. \cite{MessiahBOOKii}.
 Thus: 

\smallskip
\noindent\emph{In the limit $a\to0$ the point spectrum of $\hat{H}$ is the union of a Sommerfeld fine structure
spectrum with a negative copy of the same (both with proper quantum-mechanical labeling of the eigenvalues).}
 
 \begin{rem}\emph{
Incidentally, our discussion implicitly explains why in the usual special-relativistic calculations one only obtains half of the
symmetric point spectrum: the symmetry of the point spectrum is broken because one tacitly breaks the symmetry of the underlying
spacetime by restricting the variations to be supported on only half of it. 
 Of course, nobody at the time of Dirac and Darwin and their contemporaries should have anticipated that!}
 \end{rem}

\subsubsection{On the choices of ring radius $|a|$ and ring current $\mbox{\small{I}}$}

 Temporarily switching to physical units with $\hbar$ and $c$ restored, we note that the most suggestive a-priori
choice for $|a|$ would seem to be ``half of the electron's Compton wavelength,'' $|a| = \hbar/2mc$, 
which is obtained by equating the KN magnetic moment $-ea$ with the negative Bohr magneton $-\hbar e/2mc$; note this implies $a>0$.
 Thus suggestion has been made as early as in \cite{Car68}, where it was observed that the Kerr--Newman spacetime features
a $g$-factor $g_{\mbox{\tiny\textrm{KN}}}=2$.
 However, as will be explained next, the corresponding z$G$KN-Dirac spectrum will most likely 
deviate appreciably from the empirical Hydrogen spectrum for this choice of $a$. 
 Thus the identification of the zero-$G$ Kerr--Newman ring singularity with a binary electron/anti-electron particle structure would 
seem to receive a devastating blow. 
 However, as already emphasized above, our reasoning will also show that the choice $a=\hbar/2mc$ for the ring size is ``\emph{too naive}.'' 

 Namely, we have argued that as $a\to 0$ the positive point spectrum converges to the familiar Sommerfeld fine structure spectrum 
(with the correct Dirac labeling), and the negative one to the negative thereof.
 This is consistent with the traditional narrative that in the conventional Dirac Hydrogen problem a \emph{point-like} structureless 
electron is assumed, with an electric charge $-e$ but no magnetic moment at all.
 However, the same zero-$a$ Dirac equation, now with a homogeneous magnetic $\bA$ term added to the Hamiltonian, is also known to 
produce the correct spectrum of a ``Dirac electron in an applied homogeneous magnetic field.'' 
 Thus, the effects traditionally said to come from the ``magnetic moment of the electron'' are really supplied by the structure of 
the Dirac matrices (with their physical coefficients) which act on the bi-spinorial Dirac wave function, and not by an additional 
magnetic dipole structure of the electron.\footnote{Of course, 
        since Pauli one knows this, but it is still difficult to not let the conventional narrative 
	of the ``magnetic moment of the electron,'' which suggests it were a true property of the particle 
        like its charge, get in one's way.}
 The same will therefore be true in the $a\to0$ limit of our z$G$KN-Dirac Hamiltonian. 
 This in turn implies that for finite-$a$ the z$G$KN magnetic moment $-ea$ will actually make an \emph{additional} contribution 
to the ``magnetic moment of the electron,'' hence $-ea$ should \emph{not} be identified with the Bohr magneton itself but at most 
with the ``\textbf{anomalous magnetic moment of the electron},'' to guarantee that the finite-$a$ spectrum will continue to agree 
reasonably well with the Hydrogen spectrum. 
 The upshot is: 
\smallskip

\centerline{
\emph{The most plausible a-priori choice for the z$\,G$KN ring radius $|a|$ is $a\approx 5.83\times 10^{-4}\hbar/mc (>0)$.}}
\smallskip

\noindent
So much for the a-priori choice of $a$.

 We now come to the choice of the KN-anomalous magnetic moment, and in concert with it a possibly more refined choice of $a$ as well.
 If, as we have just argued, the finite size of the ring is associated with the anomalous magnetic moment of the electron,
then there would seem to be no need to add any further anomalous magnetic moment in form of the KN-anomalous magnetic moment. 
However, it is certainly conceivable, even likely, that a better agreement of the spectral data with the empirical spectrum 
will be obtained if both $a$ and $\textsc{i}$ are used as parameters to compute corrections to the Sommerfeld fine structure
spectrum. 
In that case one would set $\textsc{i}\pi a^2 = -  5.83\times 10^{-4} \hbar e/mc$ 
to identify the total generalized KN magnetic
moment with the ``electron's anomalous magnetic moment'' and use the remaining one-parameter freedom to optimize 
the generalized z$G$KN-Dirac spectrum in regard to the empirical one. 

\subsubsection{On the perturbative computation of finite-$a$ effects on the Hydrogen spectrum}

 Our proposed a-priori value for $a$ is a tiny value (compared to the ``naive choice''); and so, since the relevant atomic length 
scale for Hydrogen is the Bohr radius $\al_{\mbox\tiny\textsc{s}}^{-1}  \hbar/mc$, one has a tiny dimensionless parameter 
$a\al_{\mbox\tiny\textsc{s}} mc / \hbar \approx 5\times10^{-6}$.
 Hence perturbation theory could be sufficient to compute its effects.
 Furthermore, all eigenfunctions of the Dirac Hamiltonian for Born--Oppenheimer Hydrogen are known when $a=0$, so
it is tempting to conclude that perturbative computations of finite-$a$ effects on the Hydrogen spectrum for small $|a|$ 
are straightforward. 
 However, this will not be a problem of the usual perturbation-theoretical type.

 In contrast to the standard situation where to a quantum Hamiltonian an extra potential is added 
which acts on the same Hilbert space domain on which the unperturbed Hamiltonian is defined, here the configuration space itself 
changes, and with it the Hilbert space domain of the Hamiltonian.
 For instance, the symmetry of the z$G$KN-Dirac spectrum implies that one definitely would have to start from the two anti-symmetric
copies of the familiar special-relativistic Dirac-``Hydrogen'' problem mentioned above and then ``switch on'' $a$. 
 After this doubling of the Hydrogen problem one gains the advantage of being able to work with the same global coordinate chart 
throughout.
 Unfortunately, the metric and with it the highest order terms in the Hamiltonian change, which raises the problem of 
dominating them in terms of the unperturbed Hamiltonian. 

 Be that as it may, as a physicist one may nevertheless want to proceed on a purely formal level and expand\footnote{A 
   non-zero radius of convergence is not to be expected, though.}  
the Hamiltonian and the metric in powers of $a$ about $a=0$,
then use formal first-order perturbation formalism to compute corrections of $O(a)$ to the Sommerfeld fine structure
spectrum (with correct quantum labeling) --- the first order expressions are finite, as we have checked.
 However, as long as a rigorous justification is missing one needs to be aware of 
the possibility that this type of naive calculation may lead to incorrect results.
 If, on the other hand, a numerical evaluation of the z$G$KN-Dirac ODE system should confirm such formal perturbative calculations
for small $a$, then one could have confidence in such calculations, and furthermore the mathematical physicists would be called
upon to work out its rigorous foundations.
  Such numerical studies are currently under way and we hope to report on the results in the near future.

\subsubsection{On the perturbative treatment of a KN-anomalous magnetic moment}

 In contrast to the technical-conceptual problems of treating the small-$a$ regime perturbatively about $a=0$, once a complete
set of eigenfunctions (or at least a sufficiently large subset thereof) has been computed numerically by integrating the
z$G$KN-Dirac ODE system, the regime of a small KN-anomalous magnetic moment can presumably be treated with conventional
perturbation theory. 
 We write ``presumably,'' for the addition of a small KN-anomalous magnetic moment to the z$G$KN-Dirac Hamiltonian brings in the 
terms with factor $e(-ea - \textsc{i}\pi a^2)$ in the electromagnetic one-form \refeq{generalKNelmagAtilde}, which 
are more divergent at the ring than the KN terms, and thus need to be controlled by the momentum part of the Hamiltonian. 
 We have not yet established this rigorously, but have made significant progress;
we hope to report some definitive results in a follow-up publication.

\section{A de Broglie--Bohm law of evolution for the ring singularity}

 As a quantum-mechanical wave equation our Dirac equation is formulated on the configuration space of the generic position
variables of the z$G$KN-type ring singularity.
 To make contact with the empirical world one has to specify what the Dirac bi-spinor wave function is supposed to say about
the actual position as found in experiments. 
 The conventional, or (in the words of Max Born) orthodox quantum-mechanical interpretation is to assign $\Psi^\dagger\Psi$ the 
meaning of a probability density for this actually found position; however, in order to make it intelligible why $\Psi$ should
have such a significance Born argued that $\Psi$ somehow guides the actual particles in their motion in the right way, 
yet otherwise he did not bother to formulate a suitable guiding equation; cf. \cite{BornSCATTERa,BornSCATTERb}.
 A suitable guiding equation (in the non-relativistic formalism) was supplied by de Broglie \cite{deBroglieSOLVAY} and rediscovered and 
clarified 25 years later by Bohm \cite{Bohm52a,Bohm52b}, who subsequently showed that also a relativistic guiding equation for the worldpoint
$Q^\mu$ of a Dirac point particle can be formulated using Dirac's bi-spinors \cite{BohmHileyBOOK}; see \cite{HBDmodel} for a modern discussion 
and new insights.
 For a general introduction into the de Broglie--Bohm theory, see the books \cite{BohmHileyBOOK,HollandBook,DTbook}, each one of which
offers a  somewhat different perspective.
 We also recommend \cite{DGZbook} for a collection of profound papers on the de Broglie--Bohm type foundations of quantum mechanics, and
see also \cite{KieFoP} for a review of common misunderstandings regarding this theory.

 In the following we explain that a de Broglie--Bohm type law can be formulated for the actual position of the ring singularity
relative to the fixed point charge with a Dreibein attached.
 To obtain this (sometimes so-called) de Broglie--Bohm--Dirac law of motion for the z$G$KN-type ring singularity, we resort again first to the
alternate (and main-stream) interpretation of our Dirac equation (which we use in \cite{KieTah14a}) as the bi-spinor wave equation for a 
spin-1/2 point fermion in the z$G$KN spacetime which interacts with its ring singularity electromagnetically. 

\subsection{The guiding laws for a Dirac point (anti-)electron in the z$G$KN spacetime}

\subsubsection{The four-current and velocity vector fields of $\Psi$}

 Associated to a bi-spinor field $\Psi$ is the four-vector field $\bj_\Psi^{}$ 
(orthodoxly called {\em quantum probability four-current}) with components
$$
j^\mu_\Psi := \overline{\Psi}\ga^\mu\Psi;
$$
equivalently, 
$$
\bj_\Psi^{} 
= \left(\begin{array}{c} \Psi^\dag\Psi\\ \Psi^\dag \boldsymbol{\al}\Psi\end{array}\right),
$$
where $\boldsymbol{\al} = (\al_1^{},\al_2^{},\al_3^{})^t$, with
$$
\al_i := \left(\begin{array}{cc} \si_i & 0 \\ 0 & -\si_i \end{array}\right).
$$
 Evaluating on the generalized Cayley--Klein representation \refeq{canspin} of $\Psi$, we have
$$
\bj_\Psi^{} =   R^2 \left(\begin{array}{c}  1 \\ \cos^2\frac{\Si}{2}\bn_1- \sin^2\frac{\Si}{2}\bn_2\end{array}\right).
$$

 Now we define the {\em density} $\brho$ and {\em velocity field} $\bv_{\scriptscriptstyle{\psi}}$ associated to $\Psi$ as follows:
$$
\brho := R^2,\qquad \bv_{\scriptscriptstyle{\psi}} := \cos^2\frac{\Si}{2}\bn_1- \sin^2\frac{\Si}{2}\bn_2.
$$
 We then have $(j^\mu) = (\brho,\brho\bv_{\scriptscriptstyle{\psi}})^t$.
 One readily checks that
$$
\|\bv_{\scriptscriptstyle{\psi}}\|^2 = \half\left(1+\cos^2\Si - (\bn_1 \cdot \bn_2)\sin^2\Si \right) \leq 1,
$$
with equality if and only if either $\Si = 0$ or $\pi$; or if $\bn_1$ and $\bn_2$ are {\em anti-parallel}, i.e. $\bn_1\cdot \bn_2 =-1$. 
 Therefore the current $(j^\mu)$ is always {\em causal}, i.e. either timelike or null, since
\beq\label{eq:JisCAUSAL}
{\eta}_{\mu\nu} j^\mu j^\nu = \brho^2(1 - \|\bv_{\scriptscriptstyle{\psi}}\|^2) \geq 0.
\eeq
 Moreover, the case of equality is exceptional, see \cite{RodiDaniel}. 

\begin{rem}\emph{
 Note that \refeq{eq:JisCAUSAL} supplies Einstein's relativistic gamma factor of a massive point particle moving at the speed $\bv$.
 More precisely,}
\beq
 \ga[\bv_{\scriptscriptstyle{\psi}}] = \frac{1}{\sqrt{1-|\bv_{\scriptscriptstyle{\psi}}|^2}} = \frac{\brho}{\sqrt{{\eta}_{\mu\nu} j^\mu j^\nu }}.
\eeq
\end{rem}

 By a similar analysis $\bv_{\scriptscriptstyle{\psi}} =0$ if 
$\bn_1 = \bn_2$ and $\Si = \frac{\pi}{2}$, or equivalently if $\psi_2 = e^{i\Phi}\psi_1$. 

\subsubsection{The de Broglie--Bohm-type law of motion for the point positron}

 With our present choice of charges $\textsc{q}=-e$ and $\textsc{q}'=e$, and with $m=m_{\mbox\small{e}}$ the electron mass,
\emph{this traditional interpretation} of the Dirac equation is that of Dirac's point positron in the field of an infinitely 
massive z$G$KN ring singularity. 
 The de Broglie--Bohm--Dirac law of motion for the four components $Q^\mu(\tau)$ of the \emph{actual} worldpoint of the point positron
then reads
\beq
\label{eq:dBBlawPOSITIONcomponents}
\frac{d Q^\mu}{d\tau} = \left.\frac{1}{\sqrt{{\eta}_{\al\beta} j^\al j^\beta }}\overline{\Psi}\ga^\mu\Psi\right|_{\{q^\kappa=Q^\kappa\}} ,
\eeq
or equivalently,
\beq
\label{eq:dBBlawPOSITION}
\left(\frac{d Q^\mu}{d\tau}\right)_{\mu=0}^3 = 
\left.\frac{1}{\sqrt{{\eta}_{\al\beta} j^\al j^\beta }} \left(\begin{array}{c} \Psi^\dag\Psi\\ \Psi^\dag \boldsymbol{\al}\Psi\end{array}\right)
\right|_{\{q^\kappa=Q^\kappa\}};
\eeq
here, $\tau$ is the proper-time variable along the worldline of that point.
 An illustration of a (fictitious) trajectory of such a de Broglie--Bohm motion of a Dirac point positron in the z$G$KN spacetime 
with generic oblate spheroidal coordinates $(t,r,\theta,\varphi)$ is shown in Fig.~6.

\begin{figure}[ht]
\centering
\includegraphics[scale=0.25]{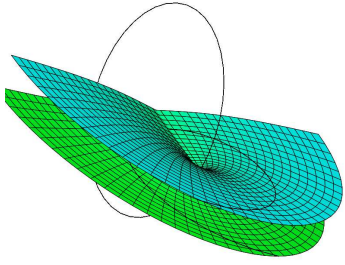}
\caption{\small An illustration of a fictitious de Broglie--Bohm world line of a ``traditional Dirac point positron,'' projected onto a
constant-$\varphi$ section of a constant-$t$ snapshot of the z$G$KN spacetime $\cM$, immersed 
into three-dimensional Euclidean space; again, the two sheets of this flat Sommerfeld space are slightly bent for 
the purpose of visualization. Shown is also the ring singularity of the constant-$t$ snapshot of $\cM$.}
\end{figure}

 If $j^\mu$ is a null vector then (\ref{eq:dBBlawPOSITION}) is ill-defined. 
 In that case, for some affine parameter $s$, 
\beq
\label{eq:dBBlawPOSITIONcomponentsNULL}
\frac{d Q^\mu}{ds} = \left.\overline{\Psi}\ga^\mu\Psi\right|_{\{q^\kappa=Q^\kappa\}} .
\eeq
 Since this case is exceptional \cite{RodiDaniel}, it will not be considered further.

\subsubsection{The Dreibein attached to the Dirac point positron}

 As shown in section 3, 
$\Psi$ evaluated at the worldpoint $(Q^\mu)(\tau) =(t,r,\theta,\varphi)_\tau$ of the point positron (in the interpretation of 
\cite{KieTah14a}) defines a normalized Dreibein ``attached'' to the point positron at $(Q^\mu)(\tau)$, consisting of
the the normalized versions of the orientation vector field $\bn_\Psi$ given by \refeq{eq:classSPIN}, and its orthogonal 
complementary vector fields $\bl_\Psi$ and $\bm_\Psi$ (constructed in an earlier subsection of section 3), all evaluated 
at the worldpoint $(Q^\mu)(\tau)$.

 Explicitly, expressing $\bn_\Psi$ in terms of $\Psi$ (see section 3.1.3), and evaluating at $(Q^\mu)(\tau)$, 
we define
\beq
\label{eq:dBBlawORIENTATION}
\bN(\tau) = 
\left.\frac{\Psi^\dag \mathbf{S} \Psi}{\Psi^\dag \Psi}\right|_{(t,\bq) = (Q^\mu)(\tau)}
 = \left.\left(\cos^2({\textstyle\frac{1}{2}}\Si) \bn_1 + \sin^2({\textstyle\frac{{1}}{2}}\Si) \bn_2\right)\right|_{(t,\bq) = (Q^\mu)(\tau)},
\eeq
and we let $\check{\bN}(\tau) = \bN(\tau)/\|\bN(\tau)\|$ denote the actual unit normal vector; similarly we define
${\bL}(\tau)$ and $\bM(\tau)$ as the evaluation of $\bl_\Psi$ and $\bm_\Psi$ at the worldpoint $(Q^\mu)(\tau)$, and set
$\check{\bL}(\tau) = \bL(\tau)/\|\bL(\tau)\|$ and $\check{\bM}(\tau) = \bM(\tau)/\|\bM(\tau)\|$. 
 Then at each $\tau$, the Dreibein $(\check{\bN}(\tau), \check{\bL}(\tau), \check{\bM}(\tau))$
defines a unique element of $SO(3)$, say $\cR(\tau)$, such that 
$\check{\bL}(\tau) = \cR(\tau)\check{\bL}(0)$,
$\check{\bM}(\tau) = \cR(\tau)\check{\bM}(0)$, and 
$\check{\bN}(\tau) = \cR(\tau)\check{\bN}(0)$.

\subsection{The law of motion for center and orientation
of the ring singularity relative to the inertial Dreibein of an infinitely massive point proton at rest}

 We want to re-interpret the guiding law \refeq{eq:dBBlawPOSITION} as supplying also the de Broglie--Bohm--Dirac motion
of a spin-1/2 z$G$KN-type ring singularity of mass $\textsc{m}=m\; (=m_{\mbox{\small e}})$ and charge $\textsc{q}=-e$ (as 
seen in the $r>0$ sheet) which interacts with an infinitely massive point ``proton'' of charge $\textsc{q}'=e$. 
 Indeed, if in the interpretation of \cite{KieTah14a} $(Q^\mu)_0=(t,r,\theta,\varphi)_0$ is the initial worldpoint of the 
point positron, with $\tau=0$ (say), and $(Q^\mu)(\tau) =(t,r,\theta,\varphi)_\tau$ is its evolved worldpoint at proper time $\tau$, 
then, as explained in subsection \ref{sec:DiracEQzGKNrg}, each quadruplet of these BL coordinates implicitly also describes the actual 
location of the z$G$KN-type ring singularity of mass $m$ and charge $-e$ relative to the position of that point ``positron'' ---
which in the interpretation of this paper becomes an infinitely massive point ``proton'' with a fixed location --- with a Dreibein
attached to it which is always a translate of the Dreibein at $\tau=0$.

 However, since the Dirac equation for a  point positron in the static z$G$KN spacetime
the Dirac bi-spinor also supplies a frame for the point positron, when re-interpreting the Dirac equation in our sense as an equation for
a z$G$KN-type ring singularity moving with respect to a point proton with an inertial $\alpha$ frame attached to it, it is clear that 
there will be a contribution to the velocity of the center of the ring singularity coming from $\cR(\tau)$. 
 More precisely, in the interpretation of our current paper, the ring normal $\check{\bN}_{\mathrm{rg}}(\tau)$ now is evolved by the inverse 
rotation, viz. 
\beq
 \check{\bN}_{\mathrm{rg}}(\tau) = \cR^{-1}(\tau) \check{\bN}_{\mathrm{rg}}(0).
\eeq
 This is the law of evolution for the orientation of the ring singularity.
 The position of the center of the moving ring singularity, $\qV(\tau)$, w.r.t. the $\alpha$ frame, is simply
\beq
\qV(\tau) = - \cR^{-1}(\tau) \QV(\tau).
\eeq
 Here, $\tau$ is still the proper time along the world line of the ``Dirac point positron'' from the traditional interpretation; yet since
$Q^0(\tau) = t$, we can express the evolutions in terms of the reference time $t$ of the $\alpha$ frame, which coincides with the time $t$
of the BL coordinates. 

\begin{rem}
 The validity of this re-interpretation of the law of motion is restricted to quasi-static motions relative to the point charge.
\end{rem}

\subsection{Comments on the de Broglie--Bohm-type laws}\label{sec:dBBcomm}

\subsubsection{Stationary states}\label{sec:dBBcommSTATIONARY}

 Interestingly, in contrast to the non-relativistic de Broglie--Bohm law of electron motion generated from a Schr\"odinger wave function 
$\psi$ which yields that for eigenstates $\psi$ the electron sits still (relative to the point ``proton''),\footnote{Many physicists 
  have expressed their unease upon learning that de Broglie--Bohm theory predicts that particles don't move 
  when guided by an eigenstate Schr\"odinger wave function.
 However, when compared with Bohr's famous postulate that the accelerated point charges in the stationary Kepler orbits of
  his model of the Hydrogen atom move but don't radiate, thereby denying the predictions of classical electromagnetic theory for such 
  motions even though classical electromagnetism was used in his model to compute the interactions, it seems to us to 
  actually be an asset of the de Broglie--Bohm theory that it resolves Bohr's dilemma in such a clean way, i.e.
  in complete agreement with the Maxwell--Lorentz field equations.}
the de Broglie--Bohm--Dirac law \refeq{eq:dBBlawPOSITION} and the orientation law \refeq{eq:dBBlawORIENTATION} 
together imply that even for eigenstate bi-spinor wave functions $\Psi$ the ring singularity generally does not 
sit still (relative to the point ``proton''). 
 A dynamical ring singularity as source of the Maxwell equations on a z$G$K spacetime would generally be expected to emit electromagnetic
radiation; note that in our formulation of the general-relativistic zero-$G$ Hydrogen problem in the Born--Oppenheimer approximation 
radiation is neglected because we are using the \emph{quasi-static approximation} for the classical electromagnetic fields to compute the 
interaction energy-momentum vector $\cP$.
 So one tentative conclusion is that the inclusion of radiation effects, even if only semi-classically, will possibly destroy most
of the energy eigenvalues except perhaps the positive and negative ``ground states.'' 
 A similar phenomenon is known from so-called ``non-relativisitc QED;'' in the Hydrogen problem for example, 
the quantum-mechanical Pauli Hamiltonian for an electron in the Coulomb field of a point proton is additionally coupled to the quantized 
radiation field, and only the ground state eigenvalue of the Pauli Hamiltonian survives this radiation coupling.

 However, at the semi-classical level of electromagnetic coupling that we use here, another possibility is conceivable. 
 Namely, one readily computes that for any solution to Dirac's equation,
$$
\bv_{\scriptscriptstyle{\psi}} \cdot \bn_{\scriptscriptstyle{\psi}} 
= \cos^4 {\textstyle\frac{1}{2}}\Si - \sin^4{\textstyle\frac{1}{2}}\Si = \cos\Si,
$$
so that $\bv_{\scriptscriptstyle{\psi}}$ is orthogonal to $\bn_\Psi$ if $\Si = \pm \frac{\pi}{2}$, 
which for non-zero $\psi_1$ and $\psi_2$ is equivalent to $|\psi_1| = |\psi_2|$.
 It is readily checked (see Appendix B) that for eigenstate bi-spinor wave functions $\Psi$, indeed one has $\Si = \pm \frac{\pi}{2}$.
 Now for any actual quadruplet $(t,r,\theta,\varphi)$ the evaluation of $\bn_\Psi$ can be rotated by an element of $SO(3)$ to 
become parallel to $\bn_{\textrm{rg}}$ computed from this quadruplet (we used this fact already once before, to fix the initial conditions;
see above).
 By the stationarity of the bi-spinor wave function, $\bn_\Psi$ evaluated with this quadruplet $(t,r,\theta,\varphi)$ 
will then remain parallel to $\bn_{\textrm{rg}}$ computed from this quadruplet for all time.
 Now, Bohm and Hiley \cite{BohmHileyBOOK} showed in the Pauli limit that the de Broglie--Bohm--Dirac velocity is a vector sum of the
de Broglie--Bohm expression one obtains with the Schr\"odinger equation, plus a gyration that is additionally supplied by the bi-spinor 
structure of the Dirac equation.
 Since, as mentioned above, the (what may be called) de Broglie--Bohm--Schr\"odinger term vanishes for an eigenstate, only the gyrational 
motion remains. 
 We surmise that this will be so also in the the fully relativistic expression; we have yet to confirm this. 
 But suppose this surmise is correct, then by the axisymmetry of the problem, with this quadruplet interpreted as the actual 
position of a point ``positron'' (using again the usual interpretation) we conclude that the point positron moves in a circle
parallel to the ring singularity. 
 In the alternate interpretation put forward in the present paper, the point charge is not a point positron but an infinitely 
massive un-accelerated point ``proton,'' so from the perspective of that point ``proton''
the ring singularity now gyrates stationarily about its axis of symmetry, and therefore generates only the static 
electromagnetic field used to compute the interaction term $\cP$ and no classical electromagnetic radiation fields, in complete 
agreement with our stationarity assumption.

\subsubsection{Quasi-Static Motions}\label{sec:dBBcommQUASISTATIC}

 In a non-stationary state the re-interpretation of our Dirac equation as covering the motion of a z$G$KN ring singularity relative to
an infinitely massive point proton with Dreibein attached is restricted to the regime of quasi-static motions.
 This means that the absolute velocities obtained from the gyrational motion need to be much smaller than the speed of light.
 We now give a rough estimate which shows that for gentle non-stationary perturbations of the Hydrogen ground state we are in the
quasi-static regime.

 Thus, we recall that the law for the evolution of the orientation vector can be decomposed into a sum of the familiar Larmor formula plus
a quantum contribution which contains $|\Psi|^2$. 
 We find that for the KN electromagnetic fields a Bohr distance away from the ring singularity the Larmor frequency times the Bohr 
distance gives a speed $c_{\mbox{\tiny{Larmor}}}$ of roughly
$$
c_{\mbox{\tiny{Larmor}}} \approx 10^{-3}\alpha_{\mbox{\tiny{S}}}^3c,
$$
while the quantum contribution is roughly
$$
c_{\mbox{\tiny{quantum}}} \approx \alpha_{\mbox{\tiny{S}}}c.
$$
 So the speed of the ring singularity relative to a point proton with fixed Dreibein, the $\alpha$ frame, attached is, in such a situation, 
roughly one-hundredth the speed of light and thus ``quasi-static.''

\subsubsection{A Feynman--St\"uckelberg dilemma avoided}\label{sec:dBBcommFeyStueckDILEMMA}

 A final remark concerns a dilemma for the de Broglie--Bohm theory when Dirac's equation is given 
Dirac's orginal point electron interpretation or the St\"uckelberg--Feynman interpretation.
 Namely the fact that with general initial conditions for $\Psi$ \emph{both} 
the putative ``electronic bi-spinors'' and the ``positronic bi-spinors'' (according to the decomposition of Hilbert space into 
positive and negative subspaces) enter $j^\mu$ --- and thus also the  guiding equation \refeq{eq:dBBlawPOSITION} for a single point ---
is perplexing with these interpretations of Dirac's equation, while it is completely natural in our interpretation of the Dirac particle 
as having a bi-particle structure in which electron and anti-electron are merely the ``two different sides of the same medal,'' realized 
through a moving zero-$G$ Kerr--Newman type ring singularity for which we have formulated a Dirac equation.

\section{Outlook}

  The most important task now would seem to be to numerically evaluate the z$G$KN-Dirac spectrum for the pertinent choice of $a$ and
see whether the modifications of the Sommerfeld spectrum are physically reasonable.
  Since the QED corrections like the Lamb shift are spectacular, it would be unwise to expect too much, but it's certainly
an important problem to compute the spectrum quantitatively.

 Other important tasks include the discussion of the Hamiltonian in the presence of a z$G$KN-anomalous magnetic moment; as
pointed out, we expect that essential self-adjointness holds if the coupling constant $(\textsc{q}a- \textsc{i}\pi a^2)e$ is
small in magnitude.
 It is also an interesting question whether more than one self-adjoint extension exists for sufficiently large coupling 
constant, or none at all.
 Incidentally, note that our essential self-adjointness result for the Dirac Hamiltonian on the z$G$KN spacetime implies
that its naked singularity does no harm. 

 We have also begun a study of the Dirac Hamiltonian for a z$G$KN ring singularity exposed to a harmonic magnetic field added 
into the z$G$KN spacetime. 
 We plan to report on this in the future.

 Another obviously important question is to tackle the problem of two interacting z$G$KN-like ring singularities, in particular 
with view toward positronium theory. 
 In that case the Born--Oppenheimer approximation is clearly non-sensical, and one has to face a two-body problem.
 Here one may hope to benefit from the standard many-body approach detailed in \cite{BeSaBOOK}.

 All these problems should be approachable with the quasi-static approximation invoked here, as long as one
is concerned with the discrete spectrum or at most gentle motions.

 To go beyond the quasi-static approximation will be a truly challenging task because then one has to contemplate deformations 
of the ring singularities: stretching, bending, twisting, curling, and who knows what else. 
 It would also require, presumably, solving Einstein's vacuum equations with such dynamical, topologically ring-like singularities as 
``boundary condition,'' plus the dynamical Maxwell equations on such a spacetime solution with the right asymptotics.
 In any event, the generalization of Dirac's equation for this situation should produce the same discrete spectrum as is 
obtained with the quasi-static approximation.

 More speculative, but very tempting and intriguing (to us), are the following thoughts.
 Namely, since in our interpretation the electron and anti-electron are merely the two different 
``topo-spin'' states of a single more fundamental particle, quite naturally one may wonder whether in the many-body 
theory there is an analog of a ferro-magnetic phase transition --- 
this would require a statistical mechanical analysis in the limit of infinitely many z$G$KN ring singularities.
 This phase transition would of course not be the magnetization phase transition, but in the topo-spin space, and amount
to the emergence of a broken symmetry phase in which only electrons (or only positrons) feature in one sheet of physical space.
 This would offer a completely novel explanation of the apparently broken particle/anti-particle symmetry in our world!

 All of the above invokes the zero-gravity limit. 
 Ultimately one wants to switch gravity back on again, not because one should expect significant corrections to the atomic spectra, 
but for conceptual reasons. 
 However, as we have recalled, for $G>0$ the maximal-analytically extended Kerr--Newman solution suffers from causal pathologies and 
strong curvature singularities which --- quite possibly --- are unphysical.
 Upon closer inspection one can trace the origin of these pathologies to the linear Maxwell--Maxwell\footnote{By speaking of the 
Maxwell--Maxwell system, the first ``Maxwell'' stands for the so-called pre-metric Maxwell equations, while the second ``Maxwell'' stands 
for the ``law of the electromagnetic vacuum'' (in condensed matter physics called the ``constitutive relations'') connecting the 
  four electromagnetic three-component fields of the pre-metric equations.
     Thus, in Minkowski spacetime, in a space~$\&$~time splitting of the Faraday tensor $\bF$ into $(\bE,\bB)$ and of the Maxwell
     tensor $\bM$ into $(\bD,\bH)$, choosing Maxwell's ``law of the pure ether'' $\bH=\bB$ and $\bE=\bD$ yields the Maxwell--Maxwell
     equations (and the general relativistic model becomes the Einstein--Maxwell--Maxwell equations).
     A physically potentially important nonlinear law of the electromagnetic vacuum is due to Born and Infeld \cite{BornInfeldBB}, in 
     which case we shall speak of Maxwell--Born--Infeld, respectively Einstein--Maxwell--Born--Infeld equations; see
     \cite{KieEMBIwPTdefects} for a recent discussion of their physical significance.
     As long as we work with the default choice of Maxwell's ``law of the pure ether'' we simply drop the second ``Maxwell,'' but
     retain it in discussions in which it is important to distinguish these from e.g. the Einstein--Maxwell--Born--Infeld equations.}
equations which for 
point and ring sources produce solutions with non-integrable field-energy densities; when coupled to the spacetime structure through
the Einstein--Maxwell--Maxwell equations, the consequences are devastating.
 Thus, before any physically meaningful deformation of the zero-$G$ results to $G>0$ can be attempted, first one has to 
redo the zero-$G$ calculations with the linear Maxwell--Maxwell equations replaced by 
better-behaved equations that yield integrable energy densities for point and ring charges, for instance the nonlinear 
Maxwell--Born--Infeld equations, then deform into the full system of Einstein--Maxwell--Born--Infeld equations.
 So far this has only been accomplished for a single point charge \cite{Tah11} in a topologically simple spacetime, 
but we don't see any reason for why this should not be possible for spacetimes with Zipoy topology and their ring singularities.
 Of course, progress may come slowly.

\medskip
\textbf{Acknowledgement}: This work did not grow in a vacuum but in the stimulating ``fundamental physics research'' atmosphere 
created by our colleages at Rutgers, in particular Joel Lebowitz, Shelly Goldstein, Rodi Tumulka, Avy Soffer, and Eric Carlen. 
We thank all of them. We thank Shelly Goldstein also for pointing out Dankel's work, and Rodi Tumulka for clarifying discussions
about the de Broglie--Bohm laws of evolution.
 Thanks go also to Friedrich Hehl for his many encouraging comments.
 We thank Barry Simon for his remarks on perturbation theory.
 We thank Ted Newman for his interesting comments, and for his permission to quote from one of his communications he sent us.
Lastly, we thank the referees for their thoughtful comments. 

\section*{Appendix}
\begin{appendix}

\subsection*{A: From $|\bq_{\mathrm{pt}}-\bq_{\mathrm{rg}}|$ and $(\bq_{\mathrm{pt}}-\bq_{\mathrm{rg}})\cdot\bn_{\mathrm{rg}}$ 
to oblate spheroidal coordinates}

Let $\bn_{\mathrm{rg}} \in \RR^3$ be a unit vector, representing the orientation of the ring, i.e. 
the ring lies in a plane with unit normal $\bn_{\mathrm{rg}}$ and  is positively oriented with respect to $\bn_{\mathrm{rg}}$.
  Next we introduce oblate spheroidal coordinates $(\xi,\eta,\varphi)$ centered at $\bq_{\mathrm{rg}}$ in order to generate a 
copy of z$G$KN with the ring singularity of radius  $|a|$  centered at $\bq_{\mathrm{rg}}$ and axis of symmetry parallel to the $q^3$ axis.
  These coordinates are defined as follows: let $\br = \bq-\bq_{\mathrm{rg}}$, and define $\xi(\br), \eta(\br)$ and $\varphi(\br)$ as follows
\begin{eqnarray}
r_1 +i r_2 &  = & a \sqrt{1+ \xi^2 } \sqrt{1-\eta^2} e^{i\varphi}\\
r_3  & = & a \xi \eta
\end{eqnarray}
The oblate spheroidal coordinates are closely related to cylindrical coordinates:  Let 
\begin{eqnarray}
z(\bq) &:=& (\bq - \bq_{\mathrm{rg}})\cdot \bn_{\mathrm{rg}} \\
\rho(\bq)& :=& |(\bq - \bq_{\mathrm{rg}})\times \bn_{\mathrm{rg}}| = \sqrt{ |\bq-\bq_{\mathrm{pt}}|^2 - z(\bq)^2}\\
\varphi(\bq) &:=& \tan^{-1} \frac{q^2 - q_{\mathrm{rg}}^2}{q^1 - q_{\mathrm{rg}}^1}+\pi H(-(q^1-q_{\mathrm{rg}}^1))\sgn(q^2-q_{\mathrm{rg}}^2)
\end{eqnarray}
($H$ is the Heaviside function) be cylindrical coordinates centered at $\bq_{\mathrm{rg}}$ and with symmetry axis parallel to $\bn_{\mathrm{rg}}$.
  We  then have the following change of coordinate formula
$$ 
z = a \xi\eta,\qquad \rho = a \sqrt{1+\xi^2}\sqrt{1-\eta^2},\qquad \varphi=\varphi.
$$
In particular, the set of points where $\xi=0$ is a 2-disc in $\RR^3$,
$$ 
\cD = \{ \bq \in \RR^3 \ |z(\bq) = 0, \rho(\bq) \leq a\}
$$
and the ring $\cR := \p \cD$ is the boundary of this disc.
 The level sets of $\xi$ are oblate spheroids, and those of $\eta$ are one-sheeted hyperboloids.
 As was mentioned before, the coordinate $\xi$ can be extended to take on negative values, and the 
physical space is likewise extended to become a double-sheeted branched Riemann space \cite{Evans51}.
  The set of points $ (\xi,\eta,\varphi)\in \RR\times [-1,1]\times [0,2\pi]$ gives us a copy of a 
spatial slice of the maximal extension of z$G$KN, which we denote by $\cN$.

Consider now a point particle with location $\bq_{ \textrm{pt}}$.  The locus of 
possible ring centers $\bq_{ \textrm{rg}}$ with ring normal $\bn_{\mathrm{rg}}$ such that the point particle sits somewhere on that ring, 
is itself a ring of the same radius, this one around the point particle.
  Suppose then that we introduce cylindrical coordinates $z',\varrho',\varphi'$ centered at $\bq_{ \textrm{pt}}$ with symmetry 
axis $\bn_{\mathrm{rg}}$, in the same manner as in the above. Let $\bq'_{\textrm{rg}}$ denote the location of the geometric center of the ring in the primed coordinates based at the point particle.
  One easily checks that 
$$
z'(\bq'_{ \textrm{rg}}) = -z(\bq_{ \textrm{pt}}),\qquad
\varrho'(\bq'_{ \textrm{rg}}) = \varrho(\bq_{ \textrm{pt}}),\qquad
\varphi'(\bq'_{ \textrm{rg}}) = \varphi(\bq_{ \textrm{pt}}) +\pi
$$
  Let $(r',\theta',\varphi')$ be oblate spheroidal coordinates based at $\bq_{ \textrm{pt}}$. 
 Then likewise we have
$$
r'(\bq'_{ \textrm{rg}}) = r(\bq_{ \textrm{pt}}),\qquad
\theta'(\bq'_{ \textrm{rg}}) = \pi - \theta(\bq_{ \textrm{pt}}).
$$
 It now follows that the interaction potential can be  equally well expressed in the primed coordinates:
$$
\bAKNanom(\bq_{\textrm{pt}}) = {\bAKNanom}'(\bq'_{\textrm{rg}})
$$
Indeed the only difference in the expression for the interaction potential would be that the sign of $a$ is flipped. However, since the geometrical center $\bq'_{ \textrm{rg}}$ of the ring singularity has no ontological significance, we invoke 
$\bq'_{ \textrm{rg}}$ only as auxiliary variable and therefore also not use its oblate spheroidal coordinates based at $\bq_{ \textrm{pt}}$,
but instead work with the original oblate spheroidal coordinates $(r,\theta,\varphi)$.

\subsection*{B: Separation of variables for the Dirac Equation on a z$G$KN spacetime}

 When $\textsc{q}=\textsc{i}\pi a$ the Dirac equation \refeq{eq:DIRACeqHAMformat} for the bi-spinor $\hat{\Psi}$ allows 
a clear separation also for the remaining $r$ and $\theta$ derivatives (commonly referred to in 
the literature as ``radial" and ``angular" derivatives, even though $r$ is 
not a radius and $\theta$ is not an angle, except at infinity).
 Thus, when $\textsc{q}=\textsc{i}\pi a$ the Dirac equation \refeq{eq:DIRACeqHAMformat} becomes
\beq\label{eq:DirSep}
(\hat{R} +\hat{A}) \hat{\Psi} = 0,
\eeq
where 
\bna
\hat{R}& := & \left(\begin{array}{cccc} imr & 0 &D_--ieQ\frac{r}{\varpi} & 0 \\
0 & imr & 0 & D_+-ieQ\frac{r}{\varpi}\\
D_+-ieQ\frac{r}{\varpi} & 0 & imr & 0 \\
0 & D_--ieQ\frac{r}{\varpi} & 0 & imr \end{array} \right),\\
\hat{A} &:= & \left(\begin{array}{cccc} -m a \cos\theta & 0 & 0 & -L_- \\
0 & -ma\cos\theta & -L_+ &0 \\
0 & L_- & ma\cos\theta & 0 \\
L_+ & 0 & 0 & ma\cos\theta
\end{array}\right),
\ena
where $D_\pm$ and $L_\pm$ have been given in (\ref{eq:DpmLpm}).
 Once a solution $\hat{\Psi}$ to \refeq{eq:DirSep} is found, the bi-spinor $\Psi := \fD\hat{\Psi}$
solves the original Dirac equation \refeq{eq:DirEqA}.

Following Chandrasekhar we make the Ansatz that a solution $\hat{\Psi}$ of \refeq{eq:DirSep} is of the form
\beq\label{chandra-ansatz} 
\hat{\Psi} = e^{-i(Et-\kappa \varphi)} \left( \begin{array}{c}R_1S_1\\ R_2 S_2\\ R_2 S_1\\ R_1 S_2 \end{array}\right),
\eeq
with $R_k$ being complex-valued functions of $r$ alone, and $S_k$ real-valued functions of $\theta$ alone.  
Let
\beq
\vec{R} := \left(\begin{array}{c} R_1\\ R_2\end{array}\right),\qquad \vec{S} := \left(\begin{array}{c} S_1\\ S_2\end{array}\right).
\eeq
Plugging the Chandrasekhar Ansatz \eqref{chandra-ansatz} into \eqref{eq:DirSep} one easily finds that there must be 
$\la\in\Cset$ such that 
\beq\label{eq:rad} 
T_{rad}\vec{R} =  E\vec{R},
\eeq
\beq\label{eq:ang}
T_{ang}\vec{S} = \la \vec{S},
\eeq
where
\bna
T_{rad} & :=  \label{eq:Trad} 
& \left(\begin{array}{cc} d_- 
&m\frac{r}{\varpi} - i\frac{\la}{\varpi} \\ m\frac{r}{\varpi}+i\frac{\la}{\varpi} 
& -d_+ \end{array}\right)
\\
T_{ang}& := \label{eq:Srad} 
& \left(\begin{array}{cc}  ma\cos\theta & l_- \\
 l_+ &ma\cos\theta  \end{array}\right)
\ena
The operators $d_\pm$ and $l_\pm$ are now ordinary differential operators in $r$ and $\theta$ respectively, 
with coefficients that depend on the unknown $E$, and parameters $a$, $\kappa$, and $eQ$:
\bna\label{opdefs}
d_\pm & := & -i \frac{d}{dr} \pm \frac{-a\kappa + eQ r}{\varpi^2}\\
l_\pm & := & \frac{d}{d\theta} \pm \left( aE\sin\theta - \kappa \csc\theta\right)
\ena
 The angular operator $T_{ang}$ in \refeq{eq:ang} is easily seen to be essentially self-adjoint on
 $(C^\infty_c((0,\pi),\sin\theta d\theta))^2$ and in fact is self-adjoint on its domain inside $(L^2((0,\pi),\sin\theta d\theta))^2$
(e.g. \cite{SufFacCos83,WINKLMEIERa}) with purely point spectrum $\la=\la_n(am,aE,\kappa)$, $n\in \Zset\setminus 0$.
  Thus in particular  $\la \in \RR$.
It then follows that the radial operator $T_{rad}$ is also essentially self-adjoint on $(C^\infty_c(\RR, dr))^2$ and in 
fact self-adjoint on its domain inside $(L^2(\RR,dr))^2.$

Suppose $\vec{R} = (R_1,R_2)^T \in (L^2(\RR))^2$ is a non-trivial solution to $T_{rad} \vec{R} = E\vec{R}$, with $E\in \RR$.
Then
\bea
\frac{dR_1}{dr} - i\left(E -  \frac{ a \kappa-eQr}{\varpi^2}\right)R_1 +\frac{1}{\varpi} (imr + \la) R_2 & = & 0\\
-\frac{d R_2}{dr} -i \left(E- \frac{ a \kappa-eQr}{\varpi^2}\right) R_2  +\frac{1}{\varpi} (imr - \la)R_1 & = & 0.
\eea
Multiply the first equation by ${R}_1^*$ and the second equation by ${R}_2^*$,  add them and take the real part, to obtain
\beq
\frac{d}{dr} \left(|R_1|^2 - |R_2|^2\right) = 0.
\eeq
Thus the difference of the moduli squared of $R_1$ and $R_2$ is constant, hence zero since they need to be integrable at infinity. 
I.e.,
\beq
|R_1| = |R_2| := R.
\eeq
Let $R_j = R e^{i\Phi_j}$ for $j=1,2$.
  Multiply the first equation by ${R}_2^*$, multiply the complex conjugate of the second equation  by $R_1$, and add them to obtain
\beq
\frac{d}{dr} \left(\frac{R_1}{ {R}_2^*}\right) = 0.
\eeq
Thus the ratio $R_1 / {R}_2^*$, and hence the sum of the arguments $\Phi_1+ \Phi_2$ must be a constant, say $\de$.
 Thus $R_1 = {R}_2^*e^{i\de}$.
Since multiplication by a constant phase factor is a gauge transformation for Dirac bi-spinors, we can replace $\hat{\Psi}$ 
with $\hat{\Psi}' = e^{-i\de/2}\hat{\Psi}$ without changing anything.
  The spinor thus obtained has the same form as \refeq{chandra-ansatz}, now with $R'_1 = {R'_2}^*$.
  Thus without loss of generality we can assume $\de = 0$ and $R_1 = {R_2^*}$.

This motivates us to set
\beq
R_1 =\frac{1}{\sqrt{2}}( u-iv),\qquad R_2  =\frac{1}{\sqrt{2}}( u + iv)
\eeq
 for real functions $u$ and $v$.
   Consider the unitary matrix
\beq
U := \frac{1}{\sqrt{2}}\left(\begin{array}{cc} 1 & -i \\ \ 1 & i \end{array}\right).
\eeq
A change of basis using $U$ brings the radial system \refeq{eq:rad} into the following standard (Hamiltonian) form
\beq\label{eq:hamil}
(H_{rad} -E)\left(\begin{array}{c} u \\ v \end{array}\right) = \left(\begin{array}{c}0 \\ 0 \end{array}\right),
\eeq
where
\beq\label{eq:Hrad}
H_{rad} := \left(\begin{array}{cc} m \frac{r}{\varpi} + \frac{\ga r+a\kappa}{\varpi^2} & -\p^{}_r + \frac{\la}{\varpi} \\[20pt]
 \p^{}_r +\frac{\la}{\varpi} & -m\frac{r}{\varpi} + \frac{\ga r+a\kappa}{\varpi^2}  \end{array}\right),
\eeq
(cf. \cite{ThallerBOOK}, eq (7.105)) with
\beq
\ga := -eQ <0.
\eeq  
Consider now the equations \refeq{eq:hamil} and \refeq{eq:ang} for unknowns $(u,v)$ and $(S_1,S_2)$.
  Let us define new unknowns $(R,\Om)$ and $(S,\Theta)$ via the Pr\"ufer transform \cite{Pru26}
\beq\label{eq:prufer}
u =\sqrt{2} R \cos\frac{\Om}{2},\quad v = \sqrt{2} R \sin\frac{\Om}{2},\quad S_1 = S \cos\frac{\Theta}{2},\quad S_2 = S \sin\frac{\Theta}{2}.
\eeq
Thus
\beq
R =\half\sqrt{u^2+v^2},\quad\Om = 2\tan^{-1}\frac{v}{u},\quad S = \sqrt{S_1^2+S_2^2},\quad \Theta = 2\tan^{-1}\frac{S_2}{S_1}.
\eeq
As a result, $R_1 = Re^{-i\Om/2}$ and $R_2 = Re^{i\Om/2}$.  
Hence $\hat{\Psi}$ can be re-expressed in terms of the Pr\"ufer variables, thus
\beq\label{genCK}
\hat{\Psi}(t,r,\theta,\varphi) = R(r)S(\theta)e^{-i(Et-\ka \varphi)} \left(\begin{array}{l}
\cos(\Theta(\theta)/2)e^{-i\Om(r)/2}\\
\sin(\Theta(\theta)/2) e^{+i\Om(r)/2}\\
\cos(\Theta(\theta)/2)e^{+i\Om(r)/2}\\
\sin(\Theta(\theta)/2)e^{-i\Om(r)/2}\end{array}\right),
\eeq
and we obtain the following equations for the new unknowns
\bna
\frac{d}{dr}\Om    &=& 2 \frac{mr}{\varpi} \cos\Om + 2\frac{\la}{\varpi} \sin\Om +2\frac{a\kappa + \gamma r}{\varpi^2} - 2E ,\label{eq:Om}\\
\frac{d}{dr} \ln R &=& \frac{mr}{\varpi}\sin\Om - \frac{\la}{\varpi} \cos\Om .\label{eq:R}
\ena
Similarly,
\bna
\frac{d}{d\theta}\Theta &=& -2ma\cos\theta\cos\Theta + 2\left(aE \sin\theta - \frac{\kappa}{\sin\theta}\right)\sin\Theta + 2\la,\label{eq:Theta}\\
\frac{d}{d\theta} \ln S &=& -ma \cos\theta\sin\Theta - \left(aE\sin\theta - \frac{\kappa}{\sin\theta}\right)\cos\Theta. \label{eq:S}
\ena
In \cite{KieTah14a} we show that the above equations, under suitable assumptions on the parameters, have solutions that give rise to
 an eigenstate $\hat{\Psi}\in {\sf H}$.

We conclude by noting that, via a slight change in notation, \refeq{genCK} can be brought to the generalized Cayley-Klein form 
\refeq{canspin}, namely, setting
$$
\tilde{R} := \sqrt{2}R(r)S(\theta),\quad
\tilde{S} := -Et+\ka \varphi
$$
one obtains
\beq\label{eigenstate}
\hat{\Psi} = \frac{1}{\sqrt{2}}\tilde{R}e^{i\tilde{S}}\left(\begin{array}{l}
 \cos(\Theta/2)e^{-i{\Om}/2}\\
\sin(\Theta/2) e^{i{\Om}/2}\\
\cos(\Theta/2)e^{i{\Om}/2}\\
\sin(\Theta/2)e^{-i{\Om}/2}\end{array}\right) =\frac{1}{\sqrt{2}} \tilde{R}e^{i\tilde{S}} \left(\begin{array}{l} \check{\psi}\\ 
\check{\psi}^*\end{array}\right)
\eeq
In particular, for bi-spinors of this type, $\Si = \pi/2$, so that $\bv \cdot \bn = 0$ and thus, per discussion in \ref{sec:dBBcomm},
 the resulting motion of the ring is a stationary circulation.

\subsection*{C: Other proposals to link Dirac's equation and the Kerr--Newman spacetime.}

   To be sure, we are 
not the first ones to suspect a connection between the Kerr--Newman (KN) spacetime and the Dirac electron,\footnote{Judged 
  by his early speculations \cite{EinsteinGRandELEMENTARYparticles}, Einstein could have been tempted, too, had he been around in
the 60s.}
recall footnote \ref{Ted}.
  Speculations in this direction seem to have appeared first in Carter's  paper \cite{Car68}, where he observed that 
the Kerr--Newmann solution features a $g$-factor of 2, just like Dirac's wave equation for the electron.
 Subsequently this idea was picked up by others, see in particular \cite{PekFra87,Lyn04,ArcPei04,Bur08}, and related papers mentioned below.
	However, all these earlier proposals have run into grave difficulties.
        To some measure these are caused by the physically questionable character of the KN solution with $G>0$, 
unveiled by Carter \cite{Car68} (see also \cite{HehlREVIEW,ONeillBOOK}), while the $G=0$ proposals in \cite{PekFra87,Lyn04}
suffer from artificial infinities introduced by truncating the z$G$KN manifold to obtain a topologically simple spacetime.

        As to its questionable character, the maximal-analytical extension of the axisymmetric and stationary ``outer'' KN
spacetime with $G>0$ \cite{Car68} has a very strong 
curvature singularity on a timelike cylindrical surface whose cross-section with constant-$t$ hypersurfaces is a circle; here, $t$ is a 
coordinate pertinent to the \emph{asymptotically (at spacelike $\infty$) timelike} Killing field that encodes the stationary character 
of the ``outer regions'' of the KN spacetime.
	This circle is commonly referred to as \emph{the ``ring'' singularity}.
	The region near the ring is especially pathological since it includes closed timelike loops.\footnote{The timelike
		ring singularity of the KN manifold is itself the limit of closed timelike loops.}
	Anybody unfortunate enough to be trapped in this region is doomed to repeat the same mistakes over and over again.
	In the black-hole sector of the parameter space of the KN family of spacetimes this acausal region is 
``hidden'' behind an event horizon, but the electron parameter values are not in this black-hole sector, the ring singularity 
is then ``naked,'' and the closed timelike loops turn the entire manifold into a \emph{causally vicious set}.

	As pointed out earlier, it was also discovered in \cite{Car68} that the ring singularity is associated with the 
topologically non-trivial feature that the maximal-analytical extension of the KN spacetime consists of two 
asymptotically flat ends which are ``doubly linked through the ring.''
	Carter showed that this Zipoy topology survives the vanishing-charge limit of the KN manifold, which 
yields the maximal-analytic extension \cite{BoyLin67} of Kerr's solution \cite{Ker63} to Einstein's vacuum equations, cf.~\cite{HawEll73}. 
	He furthermore showed that this topology survives the vanishing-mass limit of the Kerr manifold, which yields an otherwise
flat vacuum spacetime, which is also the vanishing-mass limit of a static spacetime family
discovered and completely described a few years earlier by Zipoy \cite{Zip66}.

        We emphasize that its nontrivial topology is \emph{not necessarily} a physically questionable feature of the maximal-analytically
extended Kerr--Newman solution to the Einstein--Maxwell equations, in the sense that it is not an arbitrary mathematical construct, does
not seem to lead to paradoxical conclusions, and no empirical data seem to rule it out yet.
        Yet it is frequently \emph{perceived} as a physically questionable feature.
        In particular, some general relativists seem to abhor such topologically non-trivial spacetimes,\footnote{Unfortunately
          this attitude can be encountered even nowadays, despite the ubiquity of topologically non-trivial Calabi--Yau manifolds 
in string theory.}
as exemplified by the following quote from \cite{BonSack68}: ``[Zipoy] endowed [his spacetimes] with rather terrifying topological 
properties.''
	Guided by similar sentiments, and following \cite{BonSack68} in their treatment of the Zipoy spacetimes, Israel \cite{Isr70} 
engineered a single-leafed  electromagnetic spacetime, obtained from half of the maximal-analytically extended KN manifold 
by identifying the limit points of sequences which approach the disc spanned by the ring from above 
with those approaching it from below, in \emph{one and the same} leaf.
	However, this ``short-circuiting'' procedure produces a spacetime with an \emph{ultra-singular} disc source:\footnote{Incidentally, 
                for astrophysical interpretations of the Kerr--Newman spacetime (other than a black hole) another cut-and-short-circuiting 
                procedure has been proposed in attempts to find a singular disc source for the spacetime and its electromagnetic fields;
		see \cite{BizakETalB,GRGdisc}, and also \cite{BizakETalA} for the Kerr spacetime. 
		This results in an infinitely extended, infinitely thin disc source which is not as physically bizarre as the 
		``disc source'' spanned by the KN ring singularity.}
the disc's interior
carries an infinite total charge and current, and with the wrong sign on top of that, which have to be ``partly compensated'' by an infinite 
opposite charge and current on its rim to leave a finite net amount of charge and current as diagnosed by the asymptotic flux integrals; 
see \cite{Isr70}, and also \cite{PekFra87}.
	Furthermore, the disc holds a negative infinite amount of mass in its interior, to be ``compensated partly'' by a positive infinite 
amount of mass on its rim in order to produce the ADM mass, and it rotates at superluminal speeds \cite{Isr70}.
     The mathematically artificial construction of a topologically simple spacetime with a ``disc source'' of such
physically bizarre proportions is, in our opinion, too high a price to pay to satisfy --- what to us seems to have been --- 
just a prejudice against topologically non-trivial spacetimes,\footnote{It is quite an irony of sorts that 35 years earlier Einstein, 
		prejudiced at the time against singularities, together with Rosen \cite{EinsteinRosenA} envisaged a topologically 
                non-trivial linkage (the so-called Einstein--Rosen bridge) of two copies of an outer region of the Schwarzschild 
                manifold in order to avoid its singularities.
		As known nowadays, ``the Einstein--Rosen bridge cannot be crossed,'' yet it is part of the maximal-analytical extension 
		of the outer region of the Schwarzschild manifold \cite{Kruskal,Szekeres}, 
		which --- ironically, too --- is topologically simple, yet singular, 
		having the past and future, spacelike curvature singularities which Einstein hoped to avoid.}
expressed in \cite{BonSack68,Isr70}.
 	Lastly, there still is a region of closed timelike loops in Israel's single-leafed $G>0$ spacetime.\footnote{To 
          the extent that it has been addressed at all in the pertinent literature, it has only been pointed out 
        that in the zero-charge limit the acausal region is excised by the cutting and short-circuiting process \cite{Isr70}, 
        yet in a footnote Israel notes that in the charged situation some acausal region remains after the cutting and short-circuiting.}

	By contrast, Arcos and Pereira \cite{ArcPei04} embrace the topologically non-trivial character of the maximal analytically extended
Kerr--Newman spacetime.
	Yet, also these authors cut and re-glue the $G>0$ KN manifold, though in a different way.
        Namely, in \cite{ArcPei04} it is argued that the acausal region has to be ``cut out'' and its boundary ``re-glued.'' 
        Thereby the ring singularity is removed as well, leading to a manifold in the spirit of John Wheeler's ``charge without charge'' 
\cite{Whe62}.
        Unfortunately, as noted in \cite{ArcPei04}, their manifold has the non-orientable topology of a Klein bottle.
        Furthermore, the authors of \cite{ArcPei04} only claim the continuity of the metric but do not investigate 
the higher regularity of the metric and electromagnetic fields across the re-glued cut, and in fact their gluing process 
may well have introduced some artificial singularities.

        Also in contrast to \cite{Isr70,PekFra87,Lyn04}, the authors of \cite{ArcPei04} actually try to establish a mathematical 
connection between the electromagnetic Kerr--Newman manifold and the Dirac equation for the electron beyond the fact that
both feature a $g$-factor of 2.
        Interestingly enough, they propose to identify their cut-and-re-glued Kerr--Newman spacetime itself with a bi-spinor 
solution of Dirac's equation,
by mapping the metric components of their manifold into the components of a Dirac bi-spinor which satisfies a Dirac equation for the free 
electron; a similar proposal was made in \cite{Bur08}.
	However, such an identification is conceptually unstable under perturbations: first, a Dirac bi-spinor does not 
have enough degrees of freedom to map into a solution of the Einstein--Maxwell equations in a generic neighborhood of the 
KN solution;
and second, a Dirac bi-spinor satisfying a Dirac equation for a non-free electron will hardly produce a Lorentzian metric 
(whether using the mapping of \cite{ArcPei04} or of \cite{Bur08}) that solves the Einstein--Maxwell equations. 
        Thus, the possibility of such an identification for the KN manifold, should it pan out rigorously,
would rather seem to be a mathematical curiosity without deeper physical implications.

  By considering the zero-$G$ limit of the maximal-analytically extended KN spacetime we are avoiding all 
the acausal and other pathological features of that spacetime which have prompted the ad-hoc surgical procedures in the studies 
of \cite{ArcPei04} and \cite{Bur08} that involve the removal of the ring singularity.
  In particular, we are able to compute the electromagnetic interaction between the z$G$KN ring singularity and a point
charge at rest elsewhere in the z$G$KN spacetime which provides the interaction term in Dirac's equation --- something that
is manifestly impossible in the approaches of  \cite{ArcPei04} and \cite{Bur08}.

  Also others, in particular Lynden-Bell \cite{Lyn04} and co-workers \cite{GaiLyn07}, have contemplated the zero-$G$ limit of the
Kerr--Newman spacetime as a better candidate to provide a link between general relativity and the electron.
   However, these authors do not consider the maximal-analytical extension but, following Israel's lead \cite{Isr70}, 
introduce an artificial branch cut at the disc spanned by the ring.
	{As pointed out in \cite{Isr70}, it is possible to 
		interpret the KN electromagnetic fields on Minkowski spacetime, which is the zero-$G$ limit of Israel's single-leafed
                short-circuited truncation of the maximal-analytically extended KN spacetime. 
		Kaiser \cite{Kai04} finds that interpreted in this way the superluminal speeds are absent, and the rim of the disc rotates 
                exactly at the speed of light;
		still the interior of the disc carries an infinite total charge and current  which are ``partly compensated'' 
		by an infinite opposite charge and current on the rim of the disc, leaving a finite net amount of charge and 
                current as diagnosed by the usual asymptotic formulas of classical electromagnetism.
	Undeterred by such warning signs, it has been speculated in \cite{Lyn04} ``that something similar may occur 
in the quantum electrodynamic charge distribution surrounding the point electron.''}

	Since nature has her own ways, perhaps she does associate electrons with such ultra-singular discs.
	However, we have serious doubts that there is any fundamental physical truth in this notion.
	In particular, from a mathematical point of view the ``disc singularity of the Kerr--Newman manifold'' is
an \emph{artificial construct}, akin to a branch cut,  obtained by \emph{arbitrarily}\footnote{There 
		are uncountably many other ways to cut and then short-circuit one of the so obtained leaves, and each of these 
		artificially-so-created singularities could be claimed to be ``the source of the Kerr--Newman fields.''
		One may want to argue that the disc is special because
		the maximal-analytically extended Kerr--Newman manifold has a reflection symmetry which leaves the disc spanned by the ring 
		singularity invariant, and this is inherited by the ``Kerr--Newman manifold with disc singularity;'' but that symmetry 
                remains intact also if one cuts along a sphere 
                with the ring singularity as equator, then identifies points in the two hemispheres by 	reflection.}
choosing the disc spanned by the ring singularity for cutting the maximal-analytically extended Kerr--Newman manifold, 
then discarding half of it, and short-circuiting the remaining part at the cut.

  Besides introducing mathematically artificial pathologies, in Israel's one-leaved truncation of the maximal-analytically extended
zero-$G$ Kerr--Newman spacetime the crucial electromagnetic bi-particle structure of the z$G$KN
ring singularity is lost, while for us it is the main feature needed to coherently interpret the spectral properties of 
``Dirac's wave equation for the electron.''

 Furthermore, as far as we can tell, all previous proposals to link the KN manifold with ``Dirac's equation for the electron''
invariably invoke Carter's observation that the {\em gyromagnetic ratio} of this
electromagnetic spacetime (i.e. by definition the ratio of its total magnetic moment to its total angular momentum)
is equal to $\textsc{q/m}$, which coincides with the value predicted for the electron by the one-body Dirac equation.
 Yet note that in the zero-$G$ limit the KN spacetime becomes \emph{static} so that it becomes nonsensical to speak of 
its gyro-magnetic ratio, indicating that the $g$-factor may be a false lead. 
 Of course, one may still assign an angular momentum to the ring singularity, but how to do it unambiguously is not clear.

 Lastly, to the extend that an estimate of the radius $|a|$ of the KN singularity has been attempted in earlier proposals that try
to link the KN manifold with ``Dirac's equation for the electron,'' the suggestive identification of the Kerr--Newman 
magnetic moment with the Bohr magneton is made, which leads to $2a= \hbar/mc$.
  In contrast to such estimates, our spectral analysis shows that the Bohr magneton is supplied already by the structure 
of Dirac's equation for a point particle featuring an electric monopole but no magnetic dipole, so that the magnetic moment 
carried by the z$G$KN ring singularity is to be identified with the electron's {\em anomalous} magnetic moment, yielding a 
significantly smaller ring radius $|a|$ than in all the other studies.
\end{appendix}

\baselineskip=13pt
\bibliographystyle{plain}
\bibliography{KTZzGKNDpapIIa}
\end{document}